\pdfoutput=1
\documentclass{article}

\usepackage[preprint]{neurips_2025}

\usepackage[utf8]{inputenc} %
\usepackage[T1]{fontenc}    %
\usepackage{hyperref}       %
\usepackage{url}            %
\usepackage{booktabs}       %
\usepackage{amsfonts}       %
\usepackage{nicefrac}       %
\usepackage{microtype}      %
\usepackage{xcolor}         %

\usepackage{bm}
\usepackage{xfrac}
\usepackage{amsfonts}
\usepackage{amsmath}
\usepackage{mathrsfs}
\usepackage{xspace}
\usepackage{bm}
\usepackage{cleveref}

\usepackage{algorithm,algpseudocode}
\usepackage{enumerate}
\usepackage{enumitem}
\usepackage{bbm}
\usepackage{dsfont}

\usepackage{pgfplots, pgfplotstable}
\pgfplotsset{compat=1.3}
\usepgfplotslibrary{fillbetween}
\usepackage{graphics}
\usepackage{tikz}
\usepackage[most]{tcolorbox}

\usepgfplotslibrary{groupplots}
\usetikzlibrary{matrix}

\usepackage{multicol}
\usepackage{longtable}
\usepackage{subcaption}

\definecolor{darkgreen}{rgb}{0.1, 0.6, 0.1}
\newcommand{\coloradjadam}{blue}

\newcommand{\colorzeroadam}{darkgreen}
\newcommand{\colorzerosgd}{orange}
\newcommand{\colorzeroanneal}{cyan}

\newtheorem{innerdesi}{Desideratum}

\newenvironment{desideratum}[1][]{\begin{innerdesi}[#1]}{\end{innerdesi}}

\input{widebar}

\newcommand{\safemath}[2]{\newcommand{#1}{\ensuremath{#2}\xspace}}

\safemath{\bma}{\mathbf{a}}
\safemath{\bmb}{\mathbf{b}}
\safemath{\bmc}{\mathbf{c}}
\safemath{\bmd}{\mathbf{d}}
\safemath{\bme}{\mathbf{e}}
\safemath{\bmf}{\mathbf{f}}
\safemath{\bmg}{\mathbf{g}}
\safemath{\bmh}{\mathbf{h}}
\safemath{\bmi}{\mathbf{i}}
\safemath{\bmj}{\mathbf{j}}
\safemath{\bmk}{\mathbf{k}}
\safemath{\bml}{\mathbf{l}}
\safemath{\bmm}{\mathbf{m}}
\safemath{\bmn}{\mathbf{n}}
\safemath{\bmo}{\mathbf{o}}
\safemath{\bmp}{\mathbf{p}}
\safemath{\bmq}{\mathbf{q}}
\safemath{\bmr}{\mathbf{r}}
\safemath{\bms}{\mathbf{s}}
\safemath{\bmt}{\mathbf{t}}
\safemath{\bmu}{\mathbf{u}}
\safemath{\bmv}{\mathbf{v}}
\safemath{\bmw}{\mathbf{w}}
\safemath{\bmx}{\mathbf{x}}
\safemath{\bmy}{\mathbf{y}}
\safemath{\bmz}{\mathbf{z}}
\safemath{\bmzero}{\mathbf{0}}
\safemath{\bmone}{\mathbf{1}}
\safemath{\bmpi}{\pmb{\pi}}
\safemath{\bmzeta}{\pmb{\zeta}}
\safemath{\bmalpha}{\pmb{\alpha}}
\safemath{\bmrho}{\pmb{\rho}}

\bmdefine{\biad}{a}
\bmdefine{\bibd}{b}
\bmdefine{\bicd}{c}
\bmdefine{\bidd}{d}
\bmdefine{\bied}{e}
\bmdefine{\bifd}{f}
\bmdefine{\bigd}{g}
\bmdefine{\bihd}{h}
\bmdefine{\biid}{i}
\bmdefine{\bijd}{j}
\bmdefine{\bikd}{k}
\bmdefine{\bild}{l}
\bmdefine{\bimd}{m}
\bmdefine{\bind}{n}
\bmdefine{\biod}{o}
\bmdefine{\bipd}{p}
\bmdefine{\biqd}{q}
\bmdefine{\bird}{r}
\bmdefine{\bisd}{s}
\bmdefine{\bitd}{t}
\bmdefine{\biud}{u}
\bmdefine{\bivd}{v}
\bmdefine{\biwd}{w}
\bmdefine{\bixd}{x}
\bmdefine{\biyd}{y}
\bmdefine{\bizd}{z}

\bmdefine{\bixid}{\xi}
\bmdefine{\bilambdad}{\lambda}
\bmdefine{\bimud}{\mu}
\bmdefine{\binud}{\nu}
\bmdefine{\bithetad}{\theta}
\bmdefine{\biomegad}{\omega}
\bmdefine{\biphid}{\phi}
\bmdefine{\biLd}{L}

\safemath{\bmia}{\biad}
\safemath{\bmib}{\bibd}
\safemath{\bmic}{\bicd}
\safemath{\bmid}{\bidd}
\safemath{\bmie}{\bied}
\safemath{\bmif}{\bifd}
\safemath{\bmig}{\bigd}
\safemath{\bmih}{\bihd}
\safemath{\bmii}{\biid}
\safemath{\bmij}{\bijd}
\safemath{\bmik}{\bikd}
\safemath{\bmil}{\bild}
\safemath{\bmim}{\bimd}
\safemath{\bmin}{\bind}
\safemath{\bmio}{\biod}
\safemath{\bmip}{\bipd}
\safemath{\bmiq}{\biqd}
\safemath{\bmir}{\bird}
\safemath{\bmis}{\bisd}
\safemath{\bmit}{\bitd}
\safemath{\bmiu}{\biud}
\safemath{\bmiv}{\bivd}
\safemath{\bmiw}{\biwd}
\safemath{\bmix}{\bixd}
\safemath{\bmiy}{\biyd}
\safemath{\bmiz}{\bizd}

\safemath{\bmxi}{\bixid}
\safemath{\bmlambda}{\bilambdad}
\safemath{\bmmu}{\bimud}
\safemath{\bmnu}{\binud}
\safemath{\bmtheta}{\bithetad}
\safemath{\bmomega}{\biomegad}
\safemath{\bmphi}{\biphid}
\safemath{\bmL}{\biLd}

\safemath{\bA}{\mathbf{A}}
\safemath{\bB}{\mathbf{B}}
\safemath{\bC}{\mathbf{C}}
\safemath{\bD}{\mathbf{D}}
\safemath{\bE}{\mathbf{E}}
\safemath{\bF}{\mathbf{F}}
\safemath{\bG}{\mathbf{G}}
\safemath{\bH}{\mathbf{H}}
\safemath{\bI}{\mathbf{I}}
\safemath{\bJ}{\mathbf{J}}
\safemath{\bK}{\mathbf{K}}
\safemath{\bL}{\mathbf{L}}
\safemath{\bM}{\mathbf{M}}
\safemath{\bN}{\mathbf{N}}
\safemath{\bO}{\mathbf{O}}
\safemath{\bP}{\mathbf{P}}
\safemath{\bQ}{\mathbf{Q}}
\safemath{\bR}{\mathbf{R}}
\safemath{\bS}{\mathbf{S}}
\safemath{\bT}{\mathbf{T}}
\safemath{\bU}{\mathbf{U}}
\safemath{\bV}{\mathbf{V}}
\safemath{\bW}{\mathbf{W}}
\safemath{\bX}{\mathbf{X}}
\safemath{\bY}{\mathbf{Y}}
\safemath{\bZ}{\mathbf{Z}}

\safemath{\bZero}{\mathbf{0}}
\safemath{\bOne}{\mathbf{1}}
\safemath{\bDelta}{\mathbf{\Delta}}
\safemath{\bLambda}{\mathbf{\UpLambda}}
\safemath{\bPhi}{\mathbf{\Upphi}}
\safemath{\bSigma}{\mathbf{\Upsigma}}
\safemath{\bOmega}{\mathbf{\Upomega}}
\safemath{\bTheta}{\mathbf{\Uptheta}}

\bmdefine{\biAd}{A}
\bmdefine{\biBd}{B}
\bmdefine{\biCd}{C}
\bmdefine{\biDd}{D}
\bmdefine{\biEd}{E}
\bmdefine{\biFd}{F}
\bmdefine{\biGd}{G}
\bmdefine{\biHd}{H}
\bmdefine{\biId}{I}
\bmdefine{\biJd}{J}
\bmdefine{\biKd}{K}
\bmdefine{\biLd}{L}
\bmdefine{\biMd}{M}
\bmdefine{\biOd}{N}
\bmdefine{\biPd}{O}
\bmdefine{\biQd}{P}
\bmdefine{\biRd}{R}
\bmdefine{\biSd}{S}
\bmdefine{\biTd}{T}
\bmdefine{\biUd}{U}
\bmdefine{\biVd}{V}
\bmdefine{\biWd}{W}
\bmdefine{\biXd}{X}
\bmdefine{\biYd}{Y}
\bmdefine{\biZd}{Z}

\bmdefine{\biDelta}{\Delta}
\bmdefine{\biLambda}{\Lambda}
\bmdefine{\biPhi}{\Phi}
\bmdefine{\biSigma}{\Sigma}
\bmdefine{\biOmega}{\Omega}
\bmdefine{\biTheta}{\Theta}

\safemath{\bimA}{\biAd}
\safemath{\bimB}{\biBd}
\safemath{\bimC}{\biCd}
\safemath{\bimD}{\biDd}
\safemath{\bimE}{\biEd}
\safemath{\bimF}{\biFd}
\safemath{\bimG}{\biGd}
\safemath{\bimH}{\biHd}
\safemath{\bimI}{\biId}
\safemath{\bimJ}{\biJd}
\safemath{\bimK}{\biKd}
\safemath{\bimL}{\biLd}
\safemath{\bimM}{\biMd}
\safemath{\bimN}{\biNd}
\safemath{\bimO}{\biOd}
\safemath{\bimP}{\biPd}
\safemath{\bimQ}{\biQd}
\safemath{\bimR}{\biRd}
\safemath{\bimS}{\biSd}
\safemath{\bimT}{\biTd}
\safemath{\bimU}{\biUd}
\safemath{\bimV}{\biVd}
\safemath{\bimW}{\biWd}
\safemath{\bimX}{\biXd}
\safemath{\bimY}{\biYd}
\safemath{\bimZ}{\biZd}

\safemath{\bimDelta}{\biDelta}
\safemath{\bimLambda}{\biLambda}
\safemath{\bimPhi}{\biPhi}
\safemath{\bimSigma}{\biSigma}
\safemath{\bimOmega}{\biOmega}
\safemath{\bimTheta}{\biTheta}

\safemath{\setA}{\mathcal{A}}
\safemath{\setB}{\mathcal{B}}
\safemath{\setC}{\mathcal{C}}
\safemath{\setD}{\mathcal{D}}
\safemath{\setE}{\mathcal{E}}
\safemath{\setF}{\mathcal{F}}
\safemath{\setG}{\mathcal{G}}
\safemath{\setH}{\mathcal{H}}
\safemath{\setI}{\mathcal{I}}
\safemath{\setJ}{\mathcal{J}}
\safemath{\setK}{\mathcal{K}}
\safemath{\setL}{\mathcal{L}}
\safemath{\setM}{\mathcal{M}}
\safemath{\setN}{\mathcal{N}}
\safemath{\setO}{\mathcal{O}}
\safemath{\setP}{\mathcal{P}}
\safemath{\setQ}{\mathcal{Q}}
\safemath{\setR}{\mathcal{R}}
\safemath{\setS}{\mathcal{S}}
\safemath{\setT}{\mathcal{T}}
\safemath{\setU}{\mathcal{U}}
\safemath{\setV}{\mathcal{V}}
\safemath{\setW}{\mathcal{W}}
\safemath{\setX}{\mathcal{X}}
\safemath{\setY}{\mathcal{Y}}
\safemath{\setZ}{\mathcal{Z}}
\safemath{\emptySet}{\varnothing}

\safemath{\colA}{\mathscr{A}}
\safemath{\colB}{\mathscr{B}}
\safemath{\colC}{\mathscr{C}}
\safemath{\colD}{\mathscr{D}}
\safemath{\colE}{\mathscr{E}}
\safemath{\colF}{\mathscr{F}}
\safemath{\colG}{\mathscr{G}}
\safemath{\colH}{\mathscr{H}}
\safemath{\colI}{\mathscr{I}}
\safemath{\colJ}{\mathscr{J}}
\safemath{\colK}{\mathscr{K}}
\safemath{\colL}{\mathscr{L}}
\safemath{\colM}{\mathscr{M}}
\safemath{\colN}{\mathscr{N}}
\safemath{\colO}{\mathscr{O}}
\safemath{\colP}{\mathscr{P}}
\safemath{\colQ}{\mathscr{Q}}
\safemath{\colR}{\mathscr{R}}
\safemath{\colS}{\mathscr{S}}
\safemath{\colT}{\mathscr{T}}
\safemath{\colU}{\mathscr{U}}
\safemath{\colV}{\mathscr{V}}
\safemath{\colW}{\mathscr{W}}
\safemath{\colX}{\mathscr{X}}
\safemath{\colY}{\mathscr{Y}}
\safemath{\colZ}{\mathscr{Z}}

\safemath{\opA}{\mathbb{A}}
\safemath{\opB}{\mathbb{B}}
\safemath{\opC}{\mathbb{C}}
\safemath{\opD}{\mathbb{D}}
\safemath{\opE}{\mathbb{E}}
\safemath{\opF}{\mathbb{F}}
\safemath{\opG}{\mathbb{G}}
\safemath{\opH}{\mathbb{H}}
\safemath{\opI}{\mathbb{I}}
\safemath{\opJ}{\mathbb{J}}
\safemath{\opK}{\mathbb{K}}
\safemath{\opL}{\mathbb{L}}
\safemath{\opM}{\mathbb{M}}
\safemath{\opN}{\mathbb{N}}
\safemath{\opO}{\mathbb{O}}
\safemath{\opP}{\mathbb{P}}
\safemath{\opQ}{\mathbb{Q}}
\safemath{\opR}{\mathbb{R}}
\safemath{\opS}{\mathbb{S}}
\safemath{\opT}{\mathbb{T}}
\safemath{\opU}{\mathbb{U}}
\safemath{\opV}{\mathbb{V}}
\safemath{\opW}{\mathbb{W}}
\safemath{\opX}{\mathbb{X}}
\safemath{\opY}{\mathbb{Y}}
\safemath{\opZ}{\mathbb{Z}}
\safemath{\opZero}{\mathbb{O}}
\safemath{\identityop}{\opI}

\safemath{\veca}{\bma}
\safemath{\vecb}{\bmb}
\safemath{\vecc}{\bmc}
\safemath{\vecd}{\bmd}
\safemath{\vece}{\bme}
\safemath{\vecf}{\bmf}
\safemath{\vecg}{\bmg}
\safemath{\vech}{\bmh}
\safemath{\veci}{\bmi}
\safemath{\vecj}{\bmj}
\safemath{\veck}{\bmk}
\safemath{\vecl}{\bml}
\safemath{\vecm}{\bmm}
\safemath{\vecn}{\bmn}
\safemath{\veco}{\bmo}
\safemath{\vecp}{\bmmp}
\safemath{\vecq}{\bmq}
\safemath{\vecr}{\bmr}
\safemath{\vecs}{\bms}
\safemath{\vect}{\bmt}
\safemath{\vecu}{\bmu}
\safemath{\vecv}{\bmv}
\safemath{\vecw}{\bmw}
\safemath{\vecx}{\bmx}
\safemath{\vecy}{\bmy}
\safemath{\vecz}{\bmz}

\safemath{\veczero}{\bmzero}
\safemath{\vecone}{\bmone}
\safemath{\vecxi}{\bmxi}
\safemath{\veclambda}{\bmlambda}
\safemath{\vecmu}{\bmmu}
\safemath{\vecnu}{\bmnu}
\safemath{\vecL}{\bmL}

\safemath{\vecomega}{\bmomega}
\safemath{\vecphi}{\bmphi}
\safemath{\vecpi}{\bmpi}
\safemath{\vecalpha}{\bmalpha}
\safemath{\vecrho}{\bmrho}

\safemath{\matA}{\bA}
\safemath{\matB}{\bB}
\safemath{\matC}{\bC}
\safemath{\matD}{\bD}
\safemath{\matE}{\bE}
\safemath{\matF}{\bF}
\safemath{\matG}{\bG}
\safemath{\matH}{\bH}
\safemath{\matI}{\bI}
\safemath{\matJ}{\bJ}
\safemath{\matK}{\bK}
\safemath{\matL}{\bL}
\safemath{\matM}{\bM}
\safemath{\matN}{\bN}
\safemath{\matO}{\bO}
\safemath{\matP}{\bP}
\safemath{\matQ}{\bQ}
\safemath{\matR}{\bR}
\safemath{\matS}{\bS}
\safemath{\matT}{\bT}
\safemath{\matU}{\bU}
\safemath{\matV}{\bV}
\safemath{\matW}{\bW}
\safemath{\matX}{\bX}
\safemath{\matY}{\bY}
\safemath{\matZ}{\bZ}
\safemath{\matzero}{\bmzero}

\safemath{\matDelta}{\bDelta}
\safemath{\matLambda}{\bLambda}
\safemath{\matPhi}{\bPhi}
\safemath{\matSigma}{\bSigma}
\safemath{\matOmega}{\bOmega}
\safemath{\matTheta}{\bTheta}

\safemath{\matidentity}{\matI}
\safemath{\matone}{\matO}

\safemath{\rnda}{A}
\safemath{\rndb}{B}
\safemath{\rndc}{C}
\safemath{\rndd}{D}
\safemath{\rnde}{E}
\safemath{\rndf}{F}
\safemath{\rndg}{G}
\safemath{\rndh}{H}
\safemath{\rndi}{I}
\safemath{\rndj}{J}
\safemath{\rndk}{K}
\safemath{\rndl}{L}
\safemath{\rndm}{M}
\safemath{\rndn}{N}
\safemath{\rndo}{O}
\safemath{\rndp}{P}
\safemath{\rndq}{Q}
\safemath{\rndr}{R}
\safemath{\rnds}{S}
\safemath{\rndt}{T}
\safemath{\rndu}{U}
\safemath{\rndv}{V}
\safemath{\rndw}{W}
\safemath{\rndx}{X}
\safemath{\rndy}{Y}
\safemath{\rndz}{Z}

\safemath{\rveca}{\bimA}
\safemath{\rvecb}{\bimB}
\safemath{\rvecc}{\bimC}
\safemath{\rvecd}{\bimD}
\safemath{\rvece}{\bimE}
\safemath{\rvecf}{\bimF}
\safemath{\rvecg}{\bimG}
\safemath{\rvech}{\bimH}
\safemath{\rveci}{\bimI}
\safemath{\rvecj}{\bimJ}
\safemath{\rveck}{\bimK}
\safemath{\rvecl}{\bimL}
\safemath{\rvecm}{\bimM}
\safemath{\rvecn}{\bimN}
\safemath{\rveco}{\bomO}
\safemath{\rvecp}{\bimP}
\safemath{\rvecq}{\bimQ}
\safemath{\rvecr}{\bimR}
\safemath{\rvecs}{\bimS}
\safemath{\rvect}{\bimT}
\safemath{\rvecu}{\bimU}
\safemath{\rvecv}{\bimV}
\safemath{\rvecw}{\bimW}
\safemath{\rvecx}{\bimX}
\safemath{\rvecy}{\bimY}
\safemath{\rvecz}{\bimZ}

\safemath{\rvecxi}{\bmxi}
\safemath{\rveclambda}{\bmlambda}
\safemath{\rvecmu}{\bmmu}
\safemath{\rvectheta}{\bmtheta}
\safemath{\rvecphi}{\bmphi}

\safemath{\rmatA}{\bimA}
\safemath{\rmatB}{\bimB}
\safemath{\rmatC}{\bimC}
\safemath{\rmatD}{\bimD}
\safemath{\rmatE}{\bimE}
\safemath{\rmatF}{\bimF}
\safemath{\rmatG}{\bimG}
\safemath{\rmatH}{\bimH}
\safemath{\rmatI}{\bimI}
\safemath{\rmatJ}{\bimJ}
\safemath{\rmatK}{\bimK}
\safemath{\rmatL}{\bimL}
\safemath{\rmatM}{\bimM}
\safemath{\rmatN}{\bimN}
\safemath{\rmatO}{\bimO}
\safemath{\rmatP}{\bimP}
\safemath{\rmatQ}{\bimQ}
\safemath{\rmatR}{\bimR}
\safemath{\rmatS}{\bimS}
\safemath{\rmatT}{\bimT}
\safemath{\rmatU}{\bimU}
\safemath{\rmatV}{\bimV}
\safemath{\rmatW}{\bimW}
\safemath{\rmatX}{\bimX}
\safemath{\rmatY}{\bimY}
\safemath{\rmatZ}{\bimZ}

\safemath{\rmatDelta}{\bimDelta}
\safemath{\rmatLambda}{\bimLambda}
\safemath{\rmatPhi}{\bimPhi}
\safemath{\rmatSigma}{\bimSigma}
\safemath{\rmatOmega}{\bimOmega}
\safemath{\rmatTheta}{\bimTheta}

\newenvironment{textbmatrix}{	\setlength{\arraycolsep}{2.5pt}%
								\big[\begin{matrix}}{\end{matrix}\big]%
								\raisebox{0.08ex}{\vphantom{M}}}

\def\be{\begin{equation}}
\def\ee{\end{equation}}
\def\een{\nonumber \end{equation}}
\def\mat{\begin{bmatrix}}
\def\emat{\end{bmatrix}}
\def\btm{\begin{textbmatrix}}
\def\etm{\end{textbmatrix}}

\def\ba#1\ea{\begin{align}#1\end{align}}
\def\bas#1\eas{\begin{align*}#1\end{align*}}
\def\bs#1\es{\begin{split}#1\end{split}} 
\def\bg#1\eg{\begin{gather}#1\end{gather}} 
\def\bi#1\ei{\begin{itemize}#1\end{itemize}}

\newcommand{\lefto}{\mathopen{}\left}

\DeclareMathOperator*{\argmax}{arg\;max}		%
\DeclareMathOperator{\Prob}{\opP}			%
\DeclareMathOperator{\Exop}{\opE}			%
\DeclareMathOperator{\Varop}{\opV\!\mathrm{ar}} %
\DeclareMathOperator{\grad}{\nabla}			%
\newcommand{\Var}[1]{\ensuremath{\Varop\lefto[#1\right]}} %

\newcommand{\ind}[1]{\mathbbm{1}_{#1}}

\safemath{\dirac}{\delta}					%
\safemath{\krond}{\dirac}					%
\safemath{\upto}{\uparrow}
\safemath{\downto}{\downarrow}
\safemath{\iu}{j}							%
\safemath{\ev}{\lambda}						%
\safemath{\hilseqspace}{l^{2}}				%
\newcommand{\banachfunspace}[1]{\setL^{#1}}	%
\safemath{\hilfunspace}{\banachfunspace{2}}	%

\safemath{\SNR}{\text{\sc snr}} 				%
\safemath{\No}{N_0}							%
\safemath{\Es}{E_s}							%
\safemath{\Eb}{E_b}							%
\safemath{\EbNo}{\frac{\Eb}{\No}}
\safemath{\EsNo}{\frac{\Es}{\No}}

\DeclareMathOperator{\CHop}{\ensuremath{\opH}} %
\safemath{\tvir}{\rndh_{\CHop}}				%
\safemath{\tvtf}{\rndl_{\CHop}}				%
\safemath{\spf}{\rnds_{\CHop}}				%
\safemath{\bff}{H_{\CHop}}					%

\safemath{\ircf}{r_{h}}						%
\safemath{\tftvcf}{r_{s}}					%
\safemath{\tfcf}{r_{l}}						%
\safemath{\bfcf}{r_{H}}						%

\safemath{\tcorr}{c_h}						%
\safemath{\scf}{c_{s}}						%
\safemath{\tfcorr}{c_{l}}					%
\safemath{\fcorr}{c_{H}}						%

\safemath{\mi}{I}							%
\safemath{\capacity}{C}						%

\newcommand{\entropy}{\mathcal{H}}					%
\safemath{\normal}{\mathcal{N}}			%
\safemath{\jpg}{\mathcal{CN}}			%
\safemath{\mchain}{\leftrightarrow}		%

\newtheorem{lemma}{Lemma}
\newtheorem{theorem}{Theorem}
\newtheorem{definition}{Definition}
\newtheorem{corollary}{Corollary}
\newtheorem{remark}{Remark}

\newtheorem{assumption}{Assumption}

\newenvironment{proof}{\paragraph{Proof:}}{\hfill$\square$}
\newenvironment{proofsketch}{\paragraph{Proof sketch:}}{\hfill$\square$}

\newcommand{\myalgname}{\textsc{AMID}}

\newcommand{\Vmfg}{V}
\newcommand{\qmfg}{q}
\newcommand{\Expmfg}{\setE}
\newcommand{\gpop}{\Gamma}
\newcommand{\lpop}{\Lambda}
\newcommand{\initpop}{\mu_0}
\newcommand{\pop}{L}
\newcommand{\vecpop}{\vecL}
\newcommand{\Nash}{\operatorname{Nash}}

\newcommand{\Jdg}{J}
\newcommand{\Expdg}{\setE}

\newcommand{\lpopmu}{L_{\text{pop}, \mu}}

\newcommand{\Jauc}{J_\text{auc}}
\newcommand{\Expauc}{\setE_\text{auc}}
\newcommand{\Vmfa}{V_\text{mfa}}
\newcommand{\Expmfa}{\setE_\text{mfa}}
\newcommand{\lpopmfa}{\lpop_\text{mfa}}

\newcommand{\Gauc}{\setG_\text{auc}}
\newcommand{\Mauc}{\setM_\text{mfa}}

\newcommand{\Pauc}{P^\text{mfa}}
\newcommand{\Rauc}{R^\text{mfa}}

\newcommand{\alphamax}{\alpha_{\textrm{max}}}

\newcommand{\pwin}{p_{\text{win}}}
\newcommand{\bids}{\nu^{-\perp}}

\newcommand{\objrev}{g_{\text{rev}}}

\newcommand{\empc}[1]{\sigma(#1)}

\newcommand{\Fomdpi}{\widebar{F}_{\text{omd}}}
\newcommand{\Fomd}{F_{\text{omd}}}

\newcommand{\Fupdate}{F}

\newcommand{\kldiv}{\operatorname{D}_{\textrm{kl}}}
\newcommand{\softmax}{\operatorname{softmax}}

\safemath{\vectheta}{\theta}
\safemath{\veczeta}{\zeta}

\newcommand{\threshact}{\operatorname{Th}}

\newcommand{\paymentfn}{p}
\newcommand{\utilityfn}{u}
\newcommand{\transitionfn}{w}

\newcommand{\pmarg}{\operatorname{Marg}}

\DeclareRobustCommand
  \compactcdots{\mathinner{\cdotp\mkern-3mu\cdotp\mkern-3mu\cdotp}}

\usepackage{pifont}%
\newcommand{\cmark}{\ding{51}}%
\newcommand{\xmark}{\ding{55}}%

\usepackage{siunitx}
\title{
Scalable Neural Incentive Design with Parameterized Mean-Field Approximation
}

\makeatletter
\newcommand*{\duplicatefootnotemark}{
  \textsuperscript{\@fnsymbol{1}}
}
\makeatother

\author{%
  Nathan Corecco\thanks{Equal contribution.}  \\
  Department of Computer Science \\
  ETH Zurich \\
  \texttt{nathan.corecco@inf.ethz.ch} \\
  \And
  Batuhan Yardim\duplicatefootnotemark{} \\
  Department of Computer Science \\
  ETH Zurich \\
  \texttt{yardima@inf.ethz.ch} \\
  \And
  Vinzenz Thoma \\
  Department of Computer Science \\
  ETH Zurich \& ETH AI Center\\
  \texttt{vinzenz.thoma@ai.ethz.ch} \\
  \And
  Zebang Shen \\
  Department of Computer Science \\
  ETH Zurich \\
  \texttt{zebang.shen@inf.ethz.ch} \\
  \And
  Niao He \\
  Department of Computer Science \\
  ETH Zurich \\
  \texttt{niao.he@inf.ethz.ch}
}

\begin{document}

\maketitle

\begin{abstract}
Designing incentives for a multi-agent system to induce a desirable Nash equilibrium is both a crucial and challenging problem appearing in many decision-making domains, especially for a large number of agents $N$.
Under the exchangeability assumption, we formalize this incentive design (ID) problem as a parameterized mean-field game (PMFG), aiming to reduce complexity via an infinite-population limit.
We first show that when dynamics and rewards are Lipschitz, the finite-$N$ ID objective is approximated by the PMFG at rate $\mathcal{O}(\sfrac{1}{\sqrt{N}})$.
Moreover, beyond the Lipschitz-continuous setting, we prove the same $\mathcal{O}(\sfrac{1}{\sqrt{N}})$ decay for the important special case of sequential auctions, despite discontinuities in dynamics, through a tailored auction-specific analysis.
Built on our novel approximation results, we further introduce our Adjoint Mean-Field Incentive Design (AMID) algorithm, which uses explicit differentiation of iterated equilibrium operators to compute gradients efficiently.
By uniting approximation bounds with optimization guarantees, AMID delivers a powerful, scalable algorithmic tool for many-agent (large $N$) ID.
Across diverse auction settings, the proposed AMID method substantially increases revenue over first-price formats and outperforms existing benchmark methods.
\end{abstract}

\section{Introduction}
\label{sec:introduction}
Setting the right incentives in a game with many participants is a challenging and high-stakes problem.
Policymakers must frequently make choices that affect millions, for instance, planners must design rules for curtailing city traffic \cite{ley2025less},
set pricing to maximize effective bandwidth in telecommunications networks \cite{basar2002revenue}, design spectrum auctions between telecom operators \cite{Klemperer2002How}
or manage supply and demand in energy grids \cite{maharjan2013dependable}.

We study the \textit{incentive design} (ID) problem. 
Given an objective $G$, a parameterized $N$-player game $\setG$, and player strategies $\pi_1, \ldots, \pi_N$, ID solves the equilibrium constrained optimization problem:
\begin{tcolorbox}[colback=gray!10, arc=0pt, boxrule=0pt,colframe=white!, top=1pt,left=2pt,right=2pt,bottom=1pt,before upper=\par\noindent{}]
\begin{align}
    \textbf{Maximize } G(\theta, \pi_1, \ldots, \pi_N) \textbf{ such that } (\pi_1, \ldots, \pi_N) \in \Nash_\setG(\theta)
    \label{problem:n_player_design}
    \tag{ID}
\end{align}
\end{tcolorbox}
where $\Nash_\setG(\theta)$ denotes the set of Nash equilibria (NE) of a game for the incentive parameter $\theta$ (which is to be learned).
\eqref{problem:n_player_design} is also referred to in the literature as mathematical programming with equilibrium constraints (MPEC) \cite{Luo1996Mathematical}.
Despite their relevance, MPECs are notoriously computationally challenging and in the general case NP-hard \cite{Jeroslow1985polynomial, Luo1996Mathematical}.
Simply computing a Nash equilibrium for a fixed incentive parameter $\theta$ is also \textsc{PPAD}-complete \cite{daskalakis2009complexity}, even for 2 players \cite{chen2009settling}.
For games with many players, as in many real-world problems, the so-called \emph{curse of many agents} \cite{wang2020breaking} becomes an added challenge.
In such cases, computing $\Nash_\setG(\theta)$ becomes prohibitively expensive, let alone tackling \eqref{problem:n_player_design}.
\looseness=-1

In this work, we take the approach of mean-field approximation for games with agent exchangeability (i.e., symmetry) to tackle this problem. 
Instead of solving \eqref{problem:n_player_design} directly, we construct an appropriate mean-field game (MFG) approximation $\setM$ to $\setG$ and solve the mean-field ID (MID) problem:
\begin{tcolorbox}[colback=gray!10, arc=0pt, boxrule=0pt,colframe=white!, top=1pt,left=2pt,right=2pt,bottom=1pt,before upper=\par\noindent{}]
\begin{align*}
    \textbf{Maximize } G(\theta, \pi_{\text{MFG}}) \textbf{ such that } \pi_{\text{MFG}} \in \Nash_\setM(\theta). 
    \label{problem:mfg_player_design}
    \tag{MID}
\end{align*}
\end{tcolorbox}
For the mean-field approximation \eqref{problem:mfg_player_design} to be meaningful, its solution should closely match that of the original $N$-player incentive design problem \eqref{problem:n_player_design}, with the approximation error vanishing as $N\to\infty$. We formalize this requirement in the following desideratum.
\begin{desideratum}[Approximation] \label{des_approximation}
    The solution of \eqref{problem:mfg_player_design} should be a good (approximate) solution for \eqref{problem:n_player_design}, in particular when $N$ is large, with explicit guarantees.
\end{desideratum}
MFGs are known to approximate finite-player NEs with explicit bounds for large $N$ \cite{saldi2018markov, carmona2013probabilistic, cui2021approximately, yardim2024meanfield}.
Under Lipschitz continuity, a non-asymptotic bound of $\mathcal{O}(\sfrac{1}{\sqrt{N}})$ in exploitability is obtained with a propagation-of-chaos type argument.
By contrast, our problem \eqref{problem:mfg_player_design} not only contains the classical MFG as a subproblem but also an incentive objective, accordingly, we must show that the approximation error still vanishes as $N\to\infty$.
Establishing this result is our first contribution.

\noindent\textbf{Contribution 1: Lipschitz PMFGs and Approximation.} 
We formalize \eqref{problem:mfg_player_design} as a parameterized MFG (PMFG), and show that for Lipschitz PMFGs the \eqref{problem:mfg_player_design} problem approximates \eqref{problem:n_player_design} with a rate of $\mathcal{O}(\sfrac{1}{\sqrt{N}})$, both in exploitability and optimality of the design objective.

While Lipschitz continuous MFGs cover a broad class of real-world games and are well-studied,
they exclude many important applications, notably large-scale auction design, where the transition dynamics are inherently correlated and non-Lipschitz. 
Motivated by the ubiquity and impact of auctions, we analyze \eqref{des_approximation} beyond Lipschitz PMFGs, which is our second contribution.

\noindent\textbf{Contribution 2: Approximation beyond Lipschitz.} 
We identify a PMFG for sequential batched auctions (BA-MFG) with many bidders, and establish \eqref{des_approximation} with an $\mathcal{O}(\sfrac{1}{\sqrt{N}})$ approximation rate, by identifying a set of “well-behaved” policies that is dense in the set of Nash equilibria.

Although \eqref{des_approximation} guarantees that the mean-field approximation captures the fundamental solution structure of the $N$-player incentive design problem,
realizing practical benefits from this framework requires the following desideratum.
\begin{desideratum}[Optimization] \label{des_optimization}
We would like \eqref{problem:mfg_player_design} to be computationally tractable or easier to solve than \eqref{problem:n_player_design}, specifically, admit an efficient first-order oracle that can be used for optimization.
\end{desideratum}
Our algorithmic contribution is to satisfy this desideratum in both settings: (i) PMFGs under Lipschitz continuity assumptions, and (ii) mean-field auction design.

\noindent\textbf{Contribution 3: Algorithmic Contribution.} 
We formulate our adjoint mean-field incentive design (\myalgname) algorithm for solving \eqref{problem:mfg_player_design} efficiently. 
\myalgname{} is a modification of backpropagation based on the adjoint method for computing (approximate) derivatives with Nash equilibrium constraints. 
While naive autodiff-based approaches incur a large $\mathcal{O}(T)$ memory footprint, our reformulation of the gradient computations reduces the memory footprint to $\mathcal{O}(\sqrt{T})$, with up to 80\% savings in practice. 

\noindent\textbf{Contribution 4: Experimental Contribution.} 
Finally, we use \myalgname{} to (i) solve congestion pricing in the classical beach bar MFG, and (ii) design revenue-optimal mechanisms across a variety of sequential auctions with a neural network parameterization. 
Our method consistently outperforms standard first-price mechanisms used in practice and other optimization approaches to solve \eqref{problem:mfg_player_design}.

\subsection{Related Works}
We present the works most relevant to this paper, complemented by the discussion in \Cref{app:related_works}.
\paragraph{Mathematical Programming with Equilibrium Constraints (MPEC).}
 Several works have studied gradient-based approaches for MPEC \cite{Li2020Enda,Liu2022Inducing,Wang2022Coordinating,Li2023Achievinga,Fiez2020Implicit}. Some of these assume that the equilibrium problem satisfies strong monotonicity to compute the gradients. Others use explicit differentiation, an approach we also follow in this work. Motivated by the success of reinforcement learning (RL), many works have focused on the optimal design of (multi-agent) RL environments \cite{Thoma2024Contextuala,Chen2022Adaptive,Zheng2022AI,Gerstgrasser2023Oracles,Curry2022Finding,Wang2023Differentiable}.  An important instance of designing games with desirable outcomes is \emph{automated mechanism design}, capturing many real-world economic problems \cite{Conitzer2004Self,Curry2024Automated,Likhodedov2004Methods}.

\paragraph{Steering and Equilibrium Selection.}
Complementary to \ref{problem:n_player_design} problems, a related strand of work has focused on steering and equilibrium selection. For mean-field games, \cite{Guo2023MESOB} considers a problem of choosing equilibria with high social welfare. Steering learning dynamics towards desirable equilibria was studied by \cite{Huang2024Learning,Canyakmaz2024Learning,Zhang2024Steering} for Markovian and no-regret learners and extended to MFGs by \cite{Widmer2025Steering}.

\paragraph{Mean-Field Games (MFG).}
MFGs, first formulated in the seminal works of Lasry \& Lions \cite{lasry2007mean} and Huang et al. \cite{huang2006large}, analyze symmetric competitive agents at the many-agent limit.
Recently, many works have studied RL in mean-field settings, 
such as stationary MFGs \cite{guo2019learning, zaman2023oracle, yardim2023policy, cui2021approximately}, monotone MFGs \cite{perrin2020fictitious, perolat2022scaling, perrin2022generalization, yardim2024exploiting}, static MFGs \cite{yardim2025variational}, and mean-field control \cite{carmona2023model, bensoussan2013mean}.
While general MFGs remain a theoretical challenge \cite{yardim2024meanfield}, under structured settings, they have shown empirical and algorithmic efficiency \cite{cui2021approximately,carmona2023model, lauriere2022learning}.
MFGs have also been studied in Stackelberg equilibria, closer to our setting \cite{carmona2022mean, dayanikli2023machine, aurell2022optimal, dayanikli2023machine}.
While these works have similar objectives to ours, rather than letting a leader influence a population through interactions with a static environment, we aim to design parameterized MFGs directly by explicit differentiation.
In this sense, \eqref{problem:mfg_player_design} can be seen to differ in objective from these works.
Moreover, these results do not readily apply to auction design, a foundational problem for incentive design.
In \Cref{table:selected_works}, we provide a comparison with selected works and our results.

\begin{table*}[t]
\centering
\caption{Comparison of selected works on ID in literature. 
  \textbf{Large $N$:} results scale to many-agent (symmetric) $\setG$,
  \textbf{Dynamic: } solves ID on games with dynamics,
  \textbf{Explicit diff.:} explicit differentiation for first-order information, 
  \textbf{Approx.:} finite-agent approximation in $N$, 
  \textbf{Auctions:} applies to auctions.
  }
  \begin{tabular}{lcccccc} \toprule
    \textbf{Work} & \textbf{Model} & \textbf{Large $N$} & \textbf{Dynamic} & \textbf{Approx.}  & \textbf{Explicit diff.}  & \textbf{Auctions} \\ \midrule
    \cite{Li2020Enda} & VI  & \xmark & \xmark & - & \cmark & \xmark \\
    \cite{Liu2022Inducing} & VI  & \cmark & \xmark & - & \xmark & \xmark \\
    \cite{dayanikli2023machine} & Cont. time & \cmark & \cmark & \cmark &  \xmark & \xmark \\
    \cite{moon2018linear} & LQ & \cmark & \cmark & \cmark &  \xmark & \xmark \\
    \cite{sanjari2025incentive} & Maj.-min. & \cmark & \cmark & \cmark & \xmark  & \xmark \\
     \midrule
    \textbf{Ours}  & \textbf{PMFG}  & \cmark & \cmark & \cmark-\textbf{Theorem~\ref{theorem:mainapprox}} & \cmark-\textbf{\Cref{sec:explicitdifferentiation}} & \cmark - \textbf{\Cref{sec:auctions_maintext}} \\
    \bottomrule
  \end{tabular}
  \label{table:selected_works}
\end{table*}

\section{Designing Games for Large Populations: Lipschitz Case}
\label{sec:gamedesign}

We first formalize parameterized $N$-player dynamic games and the corresponding ID problem.

\textbf{Notation.}
We use $\setS, \setA$ to denote (finite) state-actions spaces.
For the horizon $H$, define policy space $\Pi_H := \{\pi: [H] \times \setS \rightarrow \Delta_{\setA} \}$, abbreviate $\pi_h(a|s) := \pi(h,s)(a)$ and also treat $\Pi_H$ as a subspace of $\mathbb{R}^{[H]\times\setS\times\setA}$.
For a finite set $\setX$, define the ``empirical distribution'' $\empc{\vecx}(x') = \sfrac{1}{N}\sum_{i=1}^N \ind{x_i = x'}$.
We also provide a full reference table for our notation in Appendix~\ref{app:notation}.

\begin{definition}[Parameterized Dynamic Games]
\label{def:dg}
    A finite-horizon parameterized dynamic game (PDG) is a tuple $\setG := (N, \setS, \setA, H, \initpop, \Theta, \{\widebar{P}_{h, \theta}\}_{h=0}^{H-1}, \{\widebar{R}_{h, \theta}\}_{h=0}^{H-1})$ of players $N\in\mathbb{N}_{\geq 1}$,
    discrete state actions sets $\setS, \setA$, 
    parameter space $\Theta$,
    parameterized transition dynamics $\widebar{P}_{h, \theta} : \setS^N \times \setA^N \rightarrow \Delta_{\setS^N}$,  
    parameterized reward functions $\widebar{R}_{h, \theta} : \setS^N \times \setA^N\rightarrow [0,1]^N$, 
    starting distribution $\initpop \in \Delta_\setS$, 
    and time horizon $H\in\mathbb{N}_{>0}$.
    For a strategy profile $\vecpi  \in \Pi_H^N$, $\tau\geq 0$ and some $\theta\in\Theta$, the expected (entropy-regularized) sum of rewards of player $i\in[N]$ is defined as
    \begin{align*}
        \Jdg^{\tau,i}_\setG(\vecpi | \theta) := \Exop \left[ \sum_{h=0}^{H-1} \widebar{R}_{h, \theta}^i(\vecs_{h},\veca_h)  + \tau \entropy(\pi_h^i(s_h^i))  \middle| \substack{\forall j : s_0^j \sim \mu_0 , a_h ^j \sim \pi_h^j(s_h^j), \vecs_h := \{s_h^j\}_j, \veca_h := \{a_h^j\}_j,\\ \vecs_{h+1} \sim \widebar{P}_{h, \theta}(\vecs_{h},\veca_h)} \right].
    \end{align*}
\end{definition}
Define $\Expdg^{\tau}_\setG (\vecpi|\theta) := \max_{i\in [N]}\Expdg^{\tau,i}_\setG (\vecpi|\theta)$ where $\Expdg^{\tau,i}_\setG (\vecpi|\theta) :=  \max_{\pi' \in \Pi} \Jdg^{\tau,i}_\setG((\pi',\vecpi^{-i}) |\theta) - \Jdg^{\tau,i}_\setG ( \vecpi|\theta )$, the exploitability.
If $\Expdg^\tau_\setG (\vecpi^*|\theta) = 0$ for $\vecpi^*\in\Pi_H^N$, we call $\vecpi^*$ a Nash equilibrium (DG-NE) with respect to parameter $\theta$.
    The set of all Nash equilibria for $\theta\in\Theta$ is denoted $\Nash^\tau_\setG(\theta)$.
One is typically interested in maximizing a function of the aggregated population behavior (e.g., revenue, negative congestion):
\begin{align}
    G(\theta, \vecpi) := \Exop [ g(\theta, \{\widehat{L}_h\}_{h=0}^{H-1})  | \vecpi, \theta], \text{ where } \widehat{L}_h := \frac{1}{N}\sum_{j=1}^N \vece_{s^j_h, a_h^j},
    \tag{ID Objective}
\end{align}
given by some $g:\Theta\times \Delta_{\setS\times\setA}^H\rightarrow\mathbb{R}$, with the constraint that $\vecpi\in\Pi_H^N$ is an (approximate) Nash equilibrium under $\theta$ (ignoring multiplicities).
The parameter space $\Theta$ and the parameterizations of $\widebar{P}_{h, \theta}, \widebar{R}_{h, \theta}$ will dictate the implicit constraints on the design, such as the available information at time $h$.
For such parameterizations, optimizing $G$ will be nontrivial and incorporate an intractable many-agent NE computation as a subproblem.
In the following, we reduce this problem to a lower-dimensional MFG (i.e. of size independent of $N$) and propose tractable alternatives.

\subsection{Parameterized Mean Field Game Design}

Below, we formalize PMFGs.
\Cref{def:mfg} generalizes the standard definition of MFGs to a parametric family of MFGs, and approximates \Cref{def:dg} on a continuum of infinitely many players.
\begin{definition}[Parameterized Mean-Field Games]
\label{def:mfg}
    A finite-horizon parameterized mean-field game (PMFG) is a tuple $\setM := (\setS, \setA, H, \initpop, \Theta, \{P_{h, \theta}\}_{h=0}^{H-1}, \{R_{h, \theta}\}_{h=0}^{H-1})$ of 
    discrete state actions sets $\setS, \setA$, 
    parameterized transition dynamics $P_{h, \theta} : \setS\times\setA \times \Delta_{\setS \times \setA}\rightarrow \Delta_\setS$,  
    parameterized reward functions $R_{h, \theta} : \setS\times\setA\times \Delta_{\setS\times\setA} \rightarrow [0,1]$, 
    initial distribution $\initpop \in \Delta_\setS$, 
    and horizon $H\in\mathbb{N}_{>0}$.
    Define operators $\gpop_h, \lpop$ as $\gpop_h(\pop, \pi_h| \theta)(s',a') := \sum_{s, a} \pop(s,a) P_{h, \theta}(s'|s,a, \pop ) \pi_h(a'|s')$ and $\lpop(\pi|\theta) := \{\Gamma_{h-1}(\compactcdots \Gamma_1(\Gamma_0(\initpop \cdot \pi_0, \pi_1|\theta)|\theta) \compactcdots, \pi_{h-1})|\theta) \}_{h=0}^{H-1}$, called \emph{population operators}\footnote{Note that we define $\lpop(\pi|\theta)_0 := \initpop \cdot \pi_0$}.
    For $\pi  \in \Pi_H$, $\tau\geq 0$ and $\vecpop = \{ \pop_h \}_{h=0}^{H-1} \in \Delta_{\setS\times\setA}^{H}$, the total expected (entropy regularized) reward is
    \begin{align*}
        \Vmfg^\tau_\setM \left( \vecpop, \pi | \theta \right) & := \Exop \Big[ \sum_{h=0}^{H-1} R_{h, \theta}(s_h, a_h, \pop_h) + \tau \entropy(\pi_h(s_h)) \Big| \substack{s_0 \sim \initpop , \, a_h \sim \pi_h(s_h)\\ s_{h+1} \sim P_{h, \theta}(s_h, a_h, \pop_h)} \Big].
    \end{align*}
    We define mean-field exploitability as $\Expmfg^\tau_\setM (\pi|\theta) := \max_{\pi' \in \Pi} \Vmfg^\tau_\setM ( \lpop (\pi), \pi'|\theta) - \Vmfg^\tau_\setM  ( \lpop ( \pi), \pi|\theta )$.
    If $\Expmfg^\tau_\setM  (\pi^* |\theta) = 0$ for $\pi^* \in \Pi_H$, we call $\pi^*$ a MFG Nash equilibrium (MFG-NE) with respect to parameter $\theta$.
    The set of all Nash equilibria for $\theta\in\Theta$ is denoted $\Nash^\tau_\setM(\theta)$.
\end{definition}
In this section, we will make the following (standard) assumption on $P_\theta(s'|s,a,L), R_\theta(s,a,L)$, which holds in many relevant applications.
\begin{assumption}[Lipschitz continuity]
\label{assumption:lipschitz}
    For all $s,s'\in\setS, a\in\setA$, the functions $P_{h,\theta}(s'|s,a,\pop)$, $R_{h,\theta}(s,a,\pop)$, and $g(\theta, \vecpop)$ are Lipschitz continuous in $\theta, \pop$.
\end{assumption}

\Cref{theorem:mainapprox} below demonstrates that by optimizing the objective $g(\theta, \lpop(\pi^*))$, one can obtain approximation guarantees (up to a bias of $\mathcal{O}(\sfrac{1}{\sqrt{N}})$) on the performance of the PDG that has independent and symmetric state transitions. 
We have therefore established \eqref{des_approximation} for PMFGs.

\begin{theorem}
\label{theorem:mainapprox}
Let $\setM$ be a PMFG, \Cref{assumption:lipschitz} hold, and $\setG$ be the PDG such that 
$\widebar{P}_{h, \theta}(\vecs, \veca) := \bigotimes_{i\in[N]} P_{h, \theta}(s^i, a^i,\empc{\vecs, \veca})$ and $\widebar{R}_{h, \theta}^i(\vecs, \veca) = R_\theta(s^i, a^i, \empc{\vecs, \veca})$ for all $i$.
Let $\pi^*\in\Nash_\setM^\tau(\theta)$ and $\vecpi^*:= (\pi^*, \ldots, \pi^*)\in \Pi_H^N$.
Then:
\begin{enumerate}[topsep=0pt]
    \item $\Expdg^\tau_\setG(\vecpi^*|\theta) \leq \mathcal{O}(\sfrac{1}{\sqrt{N}})$, that is, $\vecpi^*$ is a $\mathcal{O}(\sfrac{1}{\sqrt{N}})$-NE of $\setG$ under $\theta$.
    \item $G(\theta, \pi^*) \geq g(\theta, \lpop(\pi^*|\theta)) - \mathcal{O}(\sfrac{1}{\sqrt{N}}).$
\end{enumerate}
\end{theorem}
\Cref{theorem:mainapprox} mirrors bounds in MFGs without ID \cite{yardim2024meanfield} and Stackelberg MFGs in other settings \cite{moon2018linear}.
For clarity, the theorem is stated as an approximation result for a PDG $\setG$ that exactly satisfies agent exchangeability, which might not always be the case.
In some applications, finding the mean-field formulation $\setM$ given a PDF $\setG$ might be nontrivial.
The \emph{converse} problem of constructing an appropriate PMFG $\setM$ that approximates a given $\setG$ was studied in \cite{yardim2024exploiting}, and this work can be trivially generalized to this case with an additional approximation bias due to asymmetries in $\setG$.
The case of auction design, which also does not satisfy \Cref{assumption:lipschitz} and the assumption of $\widebar{P}_{h, \theta}$ being a product measure, will require specific treatment in \Cref{sec:auctions_maintext}.

\subsection{Approximating the First Order Derivatives}
\label{sec:explicitdifferentiation}

Having satisfied \eqref{des_approximation} for ID with Lipschitz dynamics in the previous section, we turn to \eqref{des_optimization}--formulating algorithmic methods to solve \eqref{problem:mfg_player_design}.
We state the standard definitions of value and q-functions for PMFGs, which will be important for learning NEs:
\begin{align*}
    \Vmfg^\tau_h (s| \vecpop, \pi, \theta) &:= \Exop \Big[ \sum_{h'=h}^{H-1} R_\theta(s_{h'}, a_{h'}, \pop_{h'}) + \tau \entropy(\pi_{h'}(s_{h'})) \Big| \substack{s_h = s, \, a_{h'} \sim \pi_{h'}(s_{h'})\\ s_{h'+1} \sim P_\theta(s_{h'}, a_{h'}, \pop_{h'})} \Big] \\
    \qmfg^\tau_h (s,a | \vecpop, \pi, \theta) &:= R_\theta(s, a, L_h) + \mathbb{E}\left[ \Vmfg^\tau_{h+1} \left( s_{h+1} |\vecpop, \pi, \theta \right)  \Big| \substack{ s_h = s, a_h = a, \, a_{h'} \sim \pi_{h'}(s_{h'})\\ s_{h'+1} \sim P_\theta(s_{h'}, a_{h'}, \pop_{h'})} \right]
\end{align*}
We define the commonly used online mirror descent update rule $\Fomdpi:\Theta \times \Pi_H\rightarrow \Pi_H$ with
\begin{align*}
    \Fomdpi(\theta, \pi)(h, s) := \argmax_{u \in \Delta_\setA} \langle \qmfg^\tau_h (s,\cdot | \lpop(\pi), \pi, \theta), u\rangle + \tau \entropy(u) - \eta^{-1}(1 - \tau\eta) \kldiv(u | \pi(s)),
\end{align*}
for some given learning rate $\eta>0$ and entropy regularization $\tau\geq 0$.
$\Fomdpi$ has received particular attention in MFG literature due to its theoretical and empirical properties.
Abbreviating $\Fomdpi^{(T)}(\theta,\zeta) := \Fomdpi(\theta,\Fomdpi^{\tau}(\theta, \ldots \Fomdpi(\theta,\zeta) \ldots)),$ i.e., $\Fomdpi(\theta, \cdot)$ applied $T$ times,
the repeated iterations $\Fomdpi^{(T)}$ for $T > 0$ are known to convert to NE for monotone MFGs theoretically \cite{perolat2022scaling, zhang2023learning, isobe2024last}, and empirically find good approximations to NE \cite{cui2021approximately} for general MFGs.
Furthermore, any $\pi^* \in \Nash_\setM^\tau(\theta)$ is guaranteed to be a fixed point of the map $\Fomdpi(\theta, \cdot)$ for some learning rate $\eta$.
We formulate an explicit differentiation scheme for the PMFG using these properties of $\Fomdpi$.
Defining the softmax transform $\softmax : \mathbb{R}^{[H]\times\setS\times\setA} \rightarrow \Pi $ as $\softmax(\veczeta)(h,s,a) := \frac{\exp\{\zeta_{h,s,a}\}}{\sum_{a'} \exp\{\zeta_{h,s,a'}\}}$, the above OMD update rule can be reformulated in terms of log probabilities:
\begin{align*}
    \Fomd(\vectheta, \veczeta)(h, s, a) := (1 - \eta \tau) \zeta_{h, s,a} + \eta  \qmfg^\tau_h (s,a | \lpop(\softmax(\veczeta)), \softmax(\veczeta), \vectheta).
\end{align*}
For fixed $\theta$, the repeated iterations $\Fomd$ will converge (under technical conditions) to $\log \pi^*$ where $\pi^*$ is an NE of the MFG induced by $\theta$.
With this motivation, we reformulate the PMFG design problem \eqref{problem:mfg_player_design} as a maximization of the $T$-step approximate objective
\begin{tcolorbox}[colback=gray!10, arc=0pt, boxrule=0pt,colframe=white!, top=0pt,left=2pt,right=2pt,bottom=1pt,before upper=\par\noindent{}]
\begin{align}
    \tag{$T$-approx.}
    \label{eq:tstepobjfunction}
    G^{T}_{\textrm{approx}}: \vectheta \rightarrow g\left(\vectheta,\, \lpop(\softmax(\Fomd^{(T)}(\vectheta, \veczeta_0))|\theta)\right).
\end{align}
\end{tcolorbox}
$G^{T}_{\textrm{approx}}$ in general is a well-defined differentiable function (see \Cref{lemma:differentiability}, \Cref{app:meanfieldmechanismdesign}).
In particular, standard autograd tools can be used to compute $\nabla G^{T}_{\textrm{approx}}$.
While the behavior of $\nabla G^{T}_{\textrm{approx}}$ when $T\rightarrow\infty$ is not immediate, \Cref{lemma:omd_diff} below shows that under technical conditions $G^{T}_{\textrm{approx}}$ is a meaningful objective function to maximize, and produces low-bias estimates of the derivatives of the true NE with respect to $\theta$ for sufficiently large $T$.
\begin{lemma}[Differentiability of $\Fomd^\infty$]
\label{lemma:omd_diff}
Let $\veczeta \in \mathbb{R}^{[H]\times\setS\times\setA}, \vectheta \in \Theta$ be such that the following hold:
\begin{enumerate}[topsep=0pt, ]
    \item $\Fomd^\infty(\vectheta', \veczeta) := \lim_{T \rightarrow \infty} \Fomd^{(T)} (\vectheta', \veczeta)$ for $\vectheta'\in U$ for a neighborhood $U$ of $\vectheta$, 
     \item For $\veczeta^*$ such that $\veczeta^* = \Fomd^\infty(\vectheta, \veczeta)$, $\qmfg_{h}^\tau$ is $C^1$ in a neighborhood of $(\vectheta, \veczeta^*)$, and $\rho(\partial_{\veczeta} \qmfg_{\cdot}^\tau(\cdot,\cdot|\lpop(\softmax(\veczeta^*)), \softmax(\veczeta^*), \vectheta )) < \tau$ for all $h\in[H]$.
\end{enumerate}
Then, $\softmax(\Fomd^\infty(\vectheta', \veczeta)) \in \Nash^\tau_\setM(\vectheta') $ on $\vectheta'\in U$, $\Fomd^\infty$ is partially differentiable in $\vectheta$ at $(\vectheta, \veczeta)$, and $\lim_{T\rightarrow\infty} \partial_\vectheta \Fomd^{(T)}(\veczeta, \vectheta) = \partial_\vectheta \Fomd^\infty(\veczeta, \vectheta)$.
\end{lemma}
\Cref{lemma:omd_diff} justifies the use of $G^{T}_{\textrm{approx}}$ (and subsequently the explicit differentiation scheme) as an objective function under mild technical conditions.
If $G$ is also $C^1$, $\nabla G^{T}_{\textrm{approx}}$ converges to the derivative of a map $\theta' \rightarrow G(\theta', \pi^{\theta'})$ where $\pi^{\theta'}$ is a function of $\theta'$ such that $\pi^{\theta'}\in \Nash^\tau_\setM({\theta'})$ locally, that is, for all $\theta'$ in some neighborhood of $\theta$.
Moreover, \Cref{lemma:omd_diff} provides intuition on how to tune the parameters $\eta, T, \tau$ and characterizes their impact on the quality of explicit differentiation.

One major challenge from a computational point of view of backpropagating \eqref{eq:tstepobjfunction} will be the size of the computational graph, growing with $\mathcal{O}(T)$.
In many MFGs, finding a good approximate MFG-NE will require a large $T$, which will incur a large computational overhead.

\begin{multicols}{2}
\Cref{algorithm:adjoint_method}, which we call \textbf{adjoint mean-field incentive design} (AMID), reduces the potentially memory-intensive backpropagation through a complex computational graph to a simple forward and backward pass in $t$.
Crucially, the update operator $\Fupdate$ is typically quite complicated for PMFGs: for instance, $\Fomd$ itself involves solving forward ($\lpop$) and backward equations ($\qmfg_h^\tau$) in $h$.
Therefore, for large $T$, naive autograd will be inefficient due to the storage of many intermediate values and a large graph.

\columnbreak
\begin{algorithm}[H]
      \centering
      \caption{\myalgname{}}\label{algorithm:adjoint_method}
      \footnotesize
      \begin{algorithmic}[1]
          \Require Update rule $\Fupdate$, objective $G$, $T, \eta, \tau, \vectheta, \veczeta_{0}$
            \For{$t \in 0, \ldots, T$} \Comment{Forward pass}
                \State{$\veczeta_{t+1} = (1 - \eta \tau)\veczeta_{t} + \eta \Fupdate(\vectheta, \veczeta_{t})$}
            \EndFor
            \State{$s_{T+1} = \partial_{\vectheta}G(\vectheta, \veczeta_{T+1})$, $\veca_{T} = -\partial_{\veczeta} G(\veczeta_{T+1})$}
            \For{$t \in T, \ldots, 0$} \Comment{Backward pass}
                \State{$a_{t-1} = (1 - \eta \tau)a_{t} + \eta a_t \partial_{\veczeta} \Fupdate(\vectheta, \veczeta_t)$}
                \State{$s_{t-1} = s_t - \eta a_t \partial_{\vectheta} \Fupdate(\vectheta, \veczeta_t)$}
            \EndFor
        \State{\Return{$s_{0}$} }
      \end{algorithmic}
\end{algorithm}
\end{multicols}
\begin{lemma}[Adjoint method]
\label{lemma:adjoint}
Let $\Theta\subset \mathbb{R}^{d'}$ be an open set and $\Fupdate: \Theta \times \mathbb{R}^d \rightarrow \mathbb{R}^d$ and $G: \Theta \times \mathbb{R}^d \rightarrow \mathbb{R}^d$ be differentiable functions.
Assume \Cref{algorithm:adjoint_method} is run with inputs $\Fupdate, G$, and $T\in \mathbb{N}_{>0}, \eta, \tau > 0, \veczeta_{0}\in\mathbb{R}^d$.
Then its return value $s_0$ is equal to $\nabla_{\vectheta} G(\vectheta, \Fupdate^{(T)}(\theta, \veczeta_0))$.
\end{lemma}
\begin{remark}
In \Cref{algorithm:adjoint_method}, $\{\veczeta_t\}_t$ will need to be stored for the backward pass in memory, however, the memory footprint can be reduced to $\mathcal{O}(\sqrt{T})$ by caching every $\mathcal{O}(\sqrt{T})$ timesteps of the forward process and recomputing $\veczeta_t$ as required, maintaining a time complexity of $\mathcal{O}(T)$.
Furthermore, \Cref{algorithm:adjoint_method} can be generalized to arbitrary Bregman divergences, which permits a variety of inner loop operators to be used (\Cref{sec:appendix_adjoint_bregman}).
\end{remark}

\section{Beyond Lipschitz: Mechanism Design for Large-Scale Auctions} 
\label{sec:auctions_maintext}
Auction design is a ubiquitous and well-studied problem of extraordinary economic interest \cite{Krishna2002AuctionTheory}. 
To analyze auctions at the mean-field regime, we move beyond the Lipschitz PMFG assumptions.
Designing auctions with maximum revenue in the resulting equilibrium can be framed as an instance of \eqref{problem:n_player_design} (see our discussion in \Cref{app:auction_eq}), and with an appropriate PMFG, tackled using \myalgname{}.
Specifically, we consider the following sequential batched auction setting motivated by real-world formats used for selling government debt, broadcast licenses, mineral rights, art, fish, timber \cite{Gale2001Sequential}, and including transactions in mined Ethereum blocks \cite{Roughgarden2021Transaction}.

\paragraph{(Parameterized) batched auctions.}
An $H$-round $N$-player batched auction with incentive parameter $\theta$ is a PDG $\Gauc$ with state space
$\setS = \setV \cup \{\perp\}$ (where the value space $\setV:=\{0, ..., \sfrac{(\vert\setV\vert - 1)}{\vert\setV\vert}\}$ represents possible valuations and $\perp$ denotes non-participation in the current round) and action space $\setA  = \{0, \dots, \sfrac{(\vert \setA \vert -1)}{\vert \setA \vert }\}$ (possible bids).
Each bidder $i\in [N]$ at round $h\in[H]$ has a private state $s_h^i \in \setS$ not revealed to the auctioneer or other bidders. 
Overall, the auctioneer sells at most $\lceil \alphamax N \rceil$ goods (for some $\alphamax\in(0,1)$) and chooses $\theta$, parameterizing the allocations and payments as outlined below. 
The auction evolves for $h\in [H]$ as follows:
\begin{enumerate}[leftmargin=*,noitemsep,topsep=0pt]
\item Initial states at $h = 0$ are independently sampled from distribution $\mu_0 \in \Delta_{\setV}$.
    \item At every round $h \in [H]$, bidders for which $s_h^i \neq \perp$ submit their bids $a_h^i \in \setA$. 
    \item Observing the bid distributions $\widehat{\nu}^{-\perp}_{h} := \tfrac{1}{N}\sum_{i \in [N]}\vece_{a_h^i} \ind{s_h^i \neq \perp}$, the mechanism decides on a ratio of goods to be sold this round, $\alpha_h^\theta(\widehat{\nu}^{-\perp}_{h})$.
    Items are allocated to the highest $\lfloor \alpha_h^\theta(\widehat{\nu}^{-\perp}_{h})  N \rfloor$ bidders, with ties broken uniformly at random.
    \item Each bidder $ i $ who receives an item, makes a payment $ p_h^i = \paymentfn_h^\theta(a_h^i, \widehat{\nu}^{-\perp}_{h}) \in \mathbb{R}_{\geq 0}$.
    A winning bidder receives utility $\utilityfn_h(s_h^i, p_h^i) \in \mathbb{R}$ and transitions to state $\perp$, while non-winning and non-participating bidders receive zero utility. 
    \item Before proceeding to round $h+1$, each bidder transitions independently to a new state according to a dynamics function $ \transitionfn_h: \mathcal{S} \times \Delta_{\mathcal{S}} \rightarrow \Delta_{\mathcal{S}} $, which maps the agent’s current state and the empirical population distribution (after the allocation at round $ h $) to a distribution over next states. 
\end{enumerate}
To ensure the mechanism can not sell more goods than are available ($\alphamax$),
we assume that the parameterizations of $\alpha_h^\theta$ are such that $\sum_h \alpha_h^\theta(\widehat{\nu}^{-\perp}_{h}) \leq \alphamax$ almost surely.
This is can be ensured, for example, by parameterizing $\alpha_h^\theta$ as a fraction of \emph{remaining goods} at every $h$.
$\Gauc$ allows for complex valuation dynamics, such as single-minded bidders (who stay in $\perp$), 
time-dependent or population-dependent evolving valuations,
as well as super and subadditive valuations over bundles of goods. 
Under these dynamics, we denote the expected utility of player $i$ and exploitability as $\Jauc^{\tau, i}, \Expauc^{\tau, i}$ respectively, as defined in \Cref{def:dg}.

We note that parameterizing $\alpha_{h}^\theta,p_{h}^\theta$ fully captures the intuition of \emph{reserve prices} in the BA-MFG setting. 
In many auction formats, a reserve price, i.e., a minimum price that bidders have to bid and pay to win, has been shown to increase revenue \cite{Myerson1981Optimal}. 
A reserve price $r_h\in \setA$ at round $h$ can be implemented for example with $\alpha_{h}(\nu)=\sum_{a'\geq r_h} \nu(a')$ and $p_h(a, \nu) = a$.

\subsection{Auctions at the Mean-Field Regime}
From the above description, $\Gauc$ clearly has exchangeable agents. 
However, the corresponding one-step state evolutions are not independent, making \Cref{theorem:mainapprox} inapplicable here.
Motivated by the relevance of large-scale auction design, we show that PMFGs are still relevant and \eqref{des_approximation} holds with a refined analysis of batched auctions in the following.
We begin by defining the correct MFG that characterizes the batch auction at the limit $N\rightarrow\infty$.
While the definition is symbol-laden, we state it for completeness.
\begin{definition}[BA-MFG]
    \label{def:bamfg}
     A Batched-Auction MFG (BA-MFG) is the PMFG $ \Mauc := (\setS, \setA, H, \initpop, \Theta, \{\Pauc_{h, \theta}\}_{h=0}^{H-1}, \{\Rauc_{h, \theta}\}_{h=0}^{H-1})$, where $\Rauc_{\theta, h}, \Pauc_{\theta, h}$ are given by:
     \begin{align*}
         \Rauc_{\theta, h}(s,a,L) &:= \pwin( s, a, L, \alpha_h^\theta(\bids(L)))\utilityfn_h (s, \paymentfn_h(a, \bids(L))) \\
         \Pauc_{\theta, h}(s'|s,a,L) &:= \pwin( s, a, L, \alpha_h^\theta(\bids(L)) ) \transitionfn_h(s'|\perp, \xi^{L,\theta})  \\
            &\quad + (1 - \pwin( s, a, L, \alpha_h^\theta(\bids(L)))) \transitionfn_h(s'|s,\xi^{L,\theta}),
     \end{align*}
     where  $\bids(L) := \sum_{s\neq \perp} L(s,\cdot)$, $\xi^{L,\theta}\in\Delta_\setS$ such that $\xi^{L,\theta}(\perp) = L(\perp) + \langle L, \pwin( \cdot,\cdot , L, \alpha_h^\theta(\bids(L)))  \rangle $ and $\xi(s) = \langle L(s, \cdot), \pwin( s,\cdot , L, \alpha_h^\theta(\bids(L))) \rangle$, and $\pwin$ defined as\footnote{In this definition we take $\sfrac{\varepsilon}{0}=\infty$, for any $\varepsilon > 0$ and $\sfrac{0}{0}=0$, for convenience.}
     \begin{align*}
        \pwin( s, a, L, \alpha) := 
            \ind{s\neq \perp} \max\left\{ 0, \min\left\{1, \frac{\alpha - \sum_{s' \in \mathcal{V},\, a' > a} L(s', a')}{\bids(L)(a)}\right\}\right\}.
     \end{align*}
\end{definition}
We use $\Vmfa^\tau, \Expmfa^\tau$ to denote expected reward and exploitability in $\Mauc$, as in \Cref{def:mfg}.

The intuition behind \Cref{def:bamfg} relies on the fact that the function $\pwin$ approximately characterizes the marginal winning probability of an agent when $N$ is large.
In fact, \Cref{theorem:approximation_auction} below shows that BA-MFG is indeed the correct model for auctions with large $N$.
Existing approximation results (such as \cite{saldi2018markov, yardim2024meanfield} as well as \Cref{theorem:mainapprox}) fundamentally are incompatible with this setting due to (1) the fact that transitions in (finite-player) auctions are not independent, and (2) due to the inherent jump discontinuities in both $P_{\theta, h}^{\text{auc}}, R_{\theta, h}^{\text{auc}}$.
No zero-dominance (NZD) policies, defined below, identify a subset of policy space $\Pi_H$ where this difficulty can be circumvented.
\begin{definition}[No zero-dominance (NZD)]
Let $ \Mauc $ be a BA-MFG and $\theta\in\Theta$. 
$\pi \in \Pi_H$ is said to satisfy the NZD property for $\theta$ if at induced $\vecpop = \{L_h\}_{h=0}^{H-1} = \lpopmfa(\pi|\theta)$ there exist no $a \in \mathcal{A}, h\in[H]$ such that $\sum_{s \in \mathcal{V}} L_h(s, a) = 0$ and $\sum_{s \in \mathcal{V}, {a' > a}} L_h(s, a') = \alpha_h^\theta(\bids(L))$.
\end{definition}
While NZD is a technical condition, it is for instance satisfied by any entropy regularized MFG-NE of $\Mauc$ if $\tau > 0$, therefore, contains $\varepsilon$-NE for arbitrarily small $\varepsilon > 0$.
With this property, BA-MFGs satisfy \eqref{des_approximation} as shown below, making it a relevant model for auction design.
\begin{theorem}[Approximation for BA-MFG]
\label{theorem:approximation_auction}
Let $ \Mauc $ be a BA-MFG with Lipschitz-continuous $\utilityfn_h, \transitionfn_h, \alpha_h^\theta, \paymentfn_h^\theta$ and let $g:\Theta\times\Delta_{\setS\times\setA}^H\rightarrow\mathbb{R}$ be Lipschitz. 
Let $\pi \in \Pi_H$ be a policy that satisfies the no zero-dominance property with respect to $\theta \in \Theta$.
Then, for $\vecpi = (\pi, \dots, \pi) \in \Pi_H^N$,
\begin{enumerate}
    \item $\Expauc^\tau(\vecpi|\theta) \leq \Expmfa^\tau(\pi|\theta) + \mathcal{O}\left( \sfrac{1}{\sqrt{N}} \right),$ for any $\tau \geq 0$,
    \item $\left| g(\theta, \lpopmfa(\pi|\theta)) - G(\theta, \vecpi) \right| \leq \mathcal{O}\left(\sfrac{1}{\sqrt{N}}\right)$.
\end{enumerate}
\end{theorem}
\begin{proofsketch}
    The proof builds on (1) special handling of the correlated evolution of $s_h^{i}$ at any round $h$ and 
    (2) showing that for non-zero dominant policies $\pi$, the dynamics are locally Lipschitz continuous. 
    The two conclusions are proved separately in Appendices \ref{section:upper_bound_exploitability_proof} and \ref{section:objective_convergence}.
\end{proofsketch}

\Cref{theorem:approximation_auction} demonstrates convergence for a broad class of policies and relates to a large strand of work on equilibrium computation for auctions, which we discuss in \Cref{app:related_works}. While a true MF-NE does not necessarily satisfy the no zero-dominance property, an entropy-regularized MF-NE does.
In this regard, the results above show that the BA-MFG \emph{essentially} characterizes the limiting behavior of batched auctions.

\begin{remark}
In general, \Cref{theorem:approximation_auction} incorporates a standard worst-case exponential bound in $H$, depending on $w_h, \pi$.
However, in certain cases, such as non-expansive or population-independent $w_h$ and $\pi$ with full support, the bound becomes polynomial in $H, |\setS|, |\setA|$ (see \Cref{section:upper_bound_exploitability_proof}).
We later verify the quality of the bound in real-world experiments.
\end{remark}

Finally, we state the following differentiability result of $\Fomd^\tau$, thus satisfying \eqref{des_optimization} when combined with the adjoint method described in \Cref{sec:explicitdifferentiation}.
This result permits mechanism design by backpropagation for any entropy regularization $\tau > 0$, completing the motivation for BA-MFG.

\begin{lemma}[Differentiability on $\Gauc$]
\label{lemma:qauclipschitz}
Let $\Mauc$ be a BA-MFG on an open parameter space $\Theta\subset\mathbb{R}^d$, with Lipschitz $\utilityfn_h, \transitionfn_h, \alpha_h^\theta, \paymentfn_h^\theta$, then $\Fomd^\tau$ is almost everywhere differentiable on $\mathbb{R}^{[H]\times\setS\times\setA} \times \Theta$.
\end{lemma}
Equipped with an algorithmic tool to solve large-scale ID problems, we move to empirical demonstration on applications.

\section{Experimental Results}
\label{sec:experimental_results}

We evaluate our methodology on numerical examples of increasing complexity, using \myalgname{} to obtain gradient estimates and \textsc{Adam} \cite{Kingma2015Adam} as an update rule on parameters $\theta$.
All experiment details, including computational resources, can be found in \Cref{app:exp_details}.
We also provide reference implementations in JAX and PyTorch\footnote{The PyTorch implementation was adapted from MFGLib\cite{mfglib}.}.

First, we demonstrate the effectiveness of our approach on the prototypical MFG of the beach bar game \cite{perrin2020fictitious}.
We formulate the PMFG $\setM_{\text{bb}}$, where a large population of beachgoers starting from $\mu_0=\operatorname{Uniform}(\setS)$ can move left, stay, or move right ($\setA:\{-1,0,1\}$) on a beach ($\setS:=[K]$) over $H$ steps, while trying to minimize their distance to the bar located at $s_\text{bar}=\sfrac{K}{2}$ and avoiding busy spots.
We parameterize a pricing mechanism $\theta \in [0, \sfrac{1}{2}]^\setS$ for spots on the beach to minimize congestion (the softmax of population flow):
\begin{align*}
    R_{\text{bb}}^{h,\theta}(s,a|\vecpop, \theta) := -\frac{d(s, s_{\text{bar}})}{K} + \frac{\left|a\right|}{K} - \frac{\log L_h(s)}{3} - \theta_s, \quad g_{\text{bb}}(\theta, \vecpop) := - \sum_{h,s} \exp\{|\setS|L_h(s)\}.
\end{align*}
We report the training curves and induced flows in \Cref{fig:bb}, along with $\theta^*$ in \Cref{app:exp_details}.

\begin{figure}[H]
    \centering
    \begin{subfigure}{0.33\textwidth}
        \input{plots/beachbar/tikz.tex}
        \centering
    \end{subfigure}
    \begin{subfigure}{0.33\textwidth}
        \centering
        \includegraphics[scale=0.35]{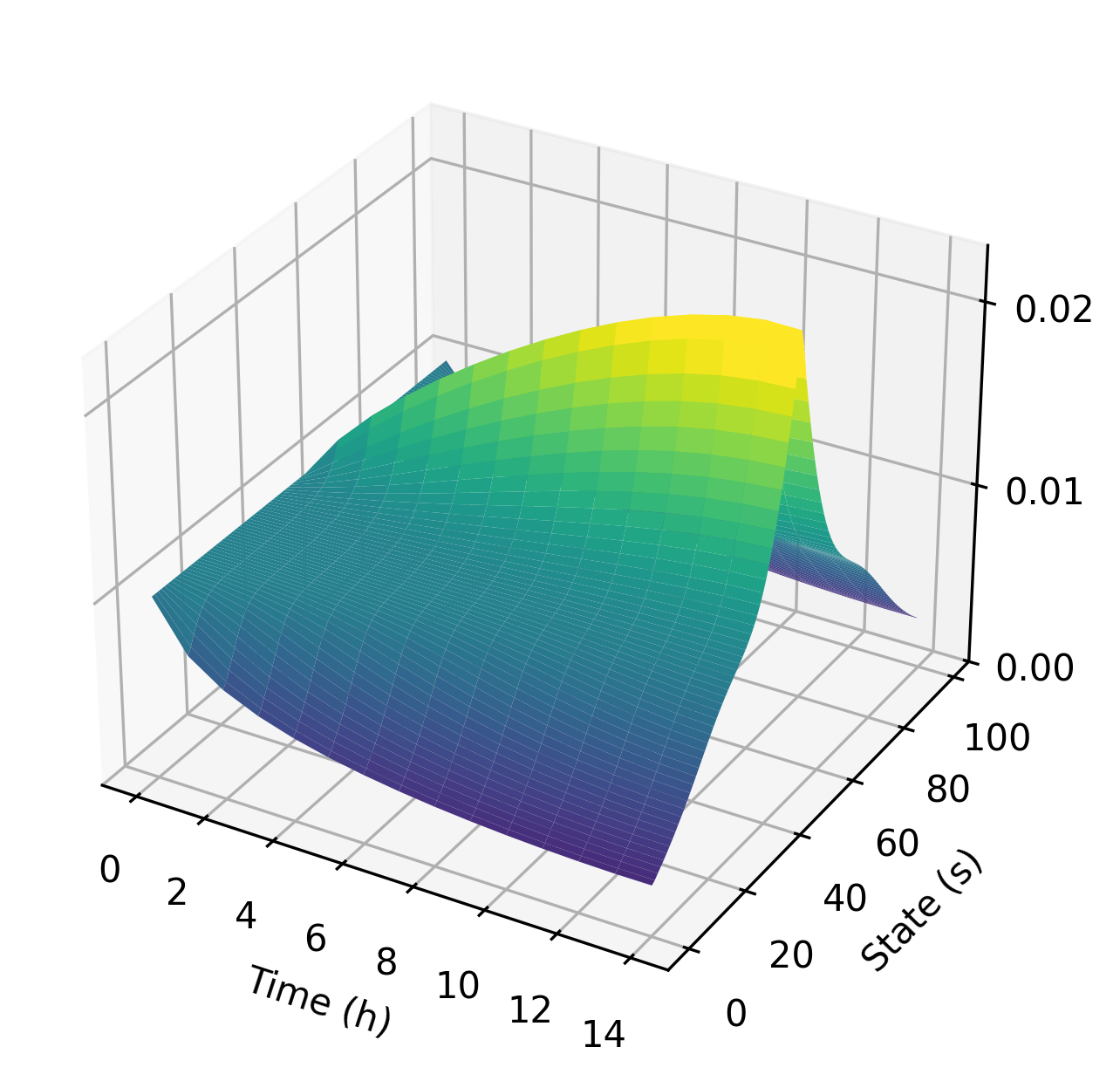}
    \end{subfigure}%
    \begin{subfigure}{0.33\textwidth}
        \centering
        \includegraphics[scale=0.35]{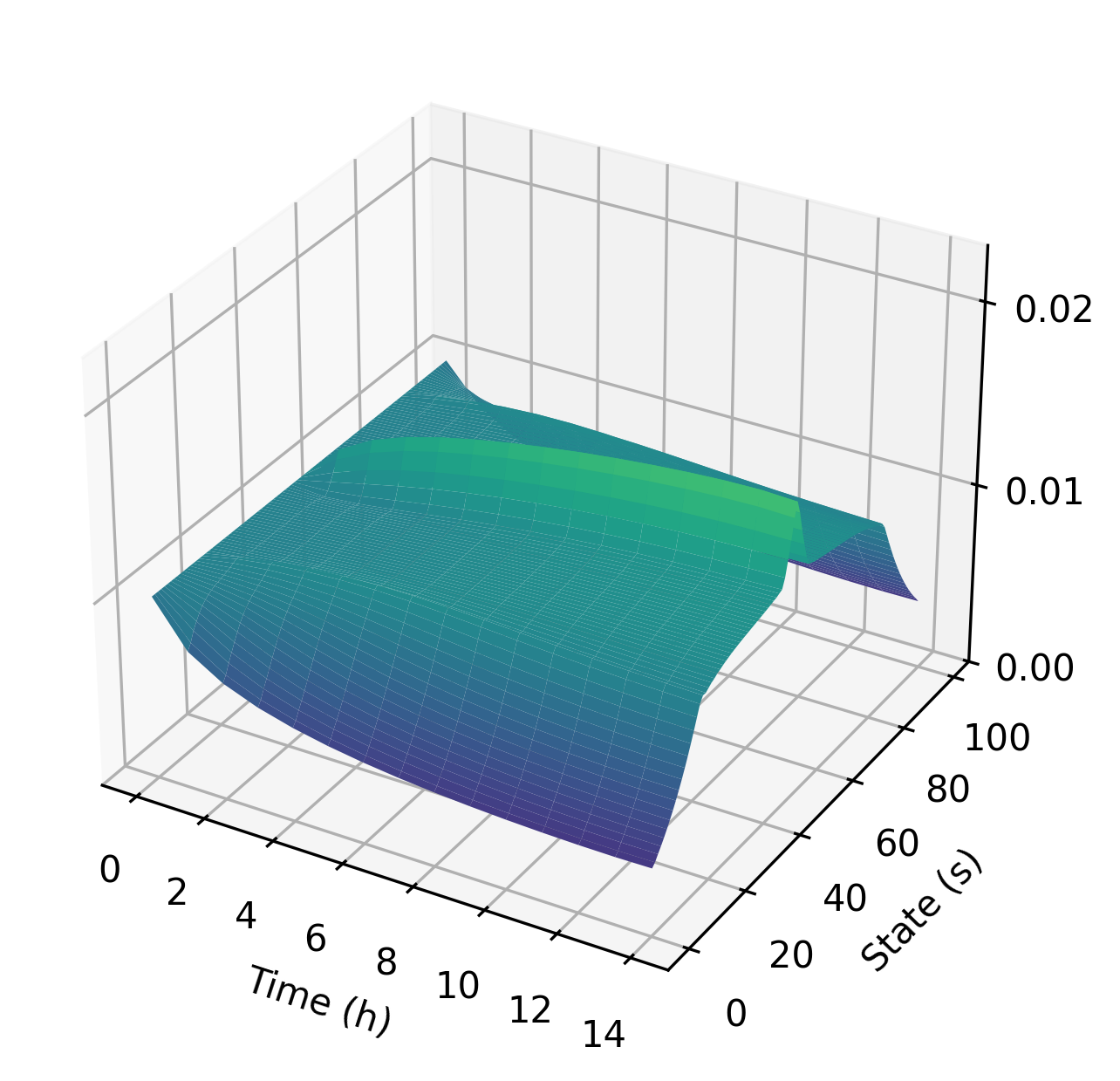}
    \end{subfigure}
    \caption{Payment design with \myalgname{} in $\setM_{\text{bb}}$. \textbf{Left: } {\color{blue} objective} and {\color{red} exploitability} throughout training iterations. \textbf{Middle-right: } population flow in time before and after learning payments.}
    \label{fig:bb}
\end{figure}

\paragraph{Dynamic auction settings.}
We move to the more challenging but relevant setting of designing neural network mechanisms BA-MFGs.
We focus on risk-neutral bidders ($u_h$ linear in payments) and on \emph{direct revelation} mechanisms, i.e. $\setS=\setA$.\footnote{The latter choice is motivated by the conceptual simplicity, widespread use in practice, and does not represent a significant restriction given the revelation principle \cite{Krishna2002AuctionTheory}.}
We set $|\setS|=100$, and maximize the revenue objective:
\begin{align*}
    \objrev(\theta, \vecpop) := \sum_{h = 0}^{H-1}\sum_{(s, a) \in \setV \times \setA}L_h(s, a)\pwin( s, a, L_h, \alpha_h^\theta(\bids(L_h))) \paymentfn_h^\theta\left(a, \bids(L_h)\right).
\end{align*}
The exact settings, labeled (A1)-(A3) are as follows:
\begin{enumerate}[leftmargin=*,noitemsep,topsep=0pt, label=\textbf{(A\arabic*)}]
    \item $H=4$, $\mu_0 = \operatorname{Uniform}(\setS)$, $\alpha_{\max}=0.8$, and single-minded bidders (after winning stay at state $\perp$) with no evolution in valuations $s_h^i$ otherwise.
    \item $H=4$, $\alpha_{\max}=0.8$, $\mu_0(s) \propto \gamma^s$ for $\gamma = 0.9$, dynamic values with $w(s'|s) \propto \exp\{\sfrac{-(3s-s')^2}{2\sigma^2}\}$ for $\sigma=0.2$, bidders are single-minded.
    \item $H=5$, $\mu_0$ uniformly sampled from $\Delta_\setS$, $\alpha_{\max}\sim \text{Uniform}([0.6, 1.2])$, participants re-enter with probability $0.3$ each round.\footnote{The setting is more challenging for two reasons. First, the neural mechanism must generalize over $\alphamax$, which it can observe. 
    Second, it must generalize over all $\mu_0$, which it does not observe, but potentially infer from $\bids(L)$. 
    This is also referred to as \textit{prior-free} mechanism design \cite{Goldberg2001Competitive,Hartline2008Optimal}}
\end{enumerate}
We parameterize $\paymentfn_h^\theta, \alpha_h^\theta$ with a residual neural network (architecture clarified in \Cref{app:exp_details}) containing $\approx 2\times 10^5$ parameters, with inputs $\vece_h, \nu_h^{-\perp},$ and remaining unsold goods at round $h$, guaranteeing by parameterization that no more than $\alphamax$ goods are sold in total.

\paragraph{Baselines.}
We evaluate \myalgname{} against several benchmarks. 
First, we compare against the results of running a simple first-price auction (\textsc{FirstPrice}), i.e., the highest bidders win and pay what they bid, to see how much more revenue we can achieve from optimizing over $\alpha^\theta_h$ and $\paymentfn^\theta_h$. 
Second, we contrast with various methods without gradient information: two methods using two-point gradient estimators ($0$-\textsc{Adam} and $0$-\textsc{SGD} respectively), and a 0-order annealing strategy (\textsc{Anneal}) using random perturbations of $\theta$.
We use $\tau=10^{-3}, \eta=10$ and $T=400$ for computing objective $G^{T}_{\textrm{approx}}$.
We report the training curves in \Cref{fig:experiments:training_curves}, where we evaluate $G^{T_{\text{val}}}_{\textrm{approx}}$ throughout training $T_{\text{val}} = 500$ for robustness.
The results indicate the effectiveness of our method against zeroth-order methods across different auctions.
Evaluations on a larger variety of settings and parameterizations (longer horizons $H$, nonlinear utilities, static mechanisms, other $g$) are also reported in \Cref{app:exp_details}.
\begin{figure}[h]
\centering
\includegraphics{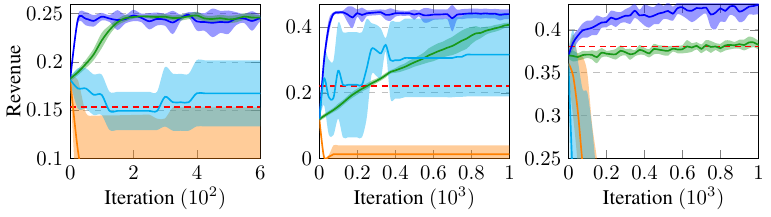}
\begin{tikzpicture}
\input{plots/legend_training_curves}
\end{tikzpicture}
\vspace{-.2cm}
\caption{$\objrev$ throughout iterations of \myalgname{} and baseline algorithms in settings (A1-3), left to right.}
\label{fig:experiments:training_curves}
\end{figure}

\paragraph{Empirically Testing \eqref{des_approximation} \& \eqref{des_optimization}.}
\Cref{fig:experiments:generalobs} illustrates that we fulfill \eqref{des_approximation} \& \eqref{des_optimization}.
Notably, the actual revenue in the $N$ player auction is very close to the optimized $\objrev$ even for $N\approx 100$.
Furthermore, the exploitability curve of OMD iterates at the optimized $\theta^*$ also suggests that the iterates are a good approximation of MF-NE, and empirically, the assumptions of \Cref{lemma:omd_diff} are valid.
Namely, OMD iterations produce a valid approximate Nash equilibrium after the end-to-end optimization process with AMID, empirically verifying that the revenue at Nash is indeed optimized.

We further provide empirical evidence supporting \eqref{des_optimization} by comparing the computational footprint of AMID against naive backpropagation through the full computational graph induced by OMD.
In \Cref{tab:experiments:computetimespace}, we report the time and memory usage of the two methods when solving (A1) with increasing time horizons $H$ on a single H100.
The results are reported for a single backpropagation step.
The modest growth in memory and computation time observed for AMID as $H$ increases highlights the scalability and practical suitability of our methodology for solving large-scale ID problems.

\begin{figure}[h]
\centering
\includegraphics{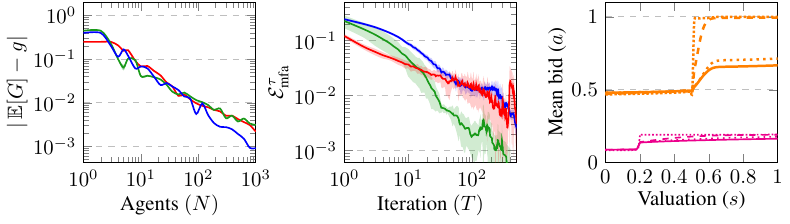}
\vspace{-.2cm}
\caption{
\textbf{Left:} deviation in revenue in $\Mauc$ vs $N$-player $\Gauc$ at $\theta^*$ as functions of $N$, and \textbf{middle:} exploitability curve of OMD iterations $\Fomd^{(T)}(\vectheta^*, \cdot)$ at optimized $\vectheta^*$ in {\color{red} (A1)}, {\color{darkgreen} (A2)}, {\color{blue} (A3)}. 
\textbf{Right:} mean bids of NE computed by $\Fomd^{(T)}$ for $h\in [4]$ {\color{magenta} before} and {\color{orange} after} optimization with \myalgname{} in $(A1)$.
}
\label{fig:experiments:generalobs}
\end{figure}

\begin{table}[ht]
\centering
\caption{Empirical compute time and memory usage for single-step naive backpropagation vs. single-step AMID across different problem horizons, in setting (A1). 
The rows with ``n/a'' indicate the method did not run on a single H100.}
\vspace{1em}
\begin{tabular}{lcccc}
\hline
\textbf{Horizon} & \textbf{Naive (time, s)} & \textbf{Naive (memory)} & \textbf{AMID (time, s)} & \textbf{AMID (memory)} \\
\hline
$H{=}5$  & $0.19 \pm 0.02$ s & 2760 MiB  & $0.09 \pm 0.01$ s  & 560 MiB  \\
$H{=}10$ & $0.25 \pm 0.10$ s & 8746 MiB  & $0.21 \pm 0.08$ s & 586 MiB  \\
$H{=}25$ & $0.71 \pm 0.15$ s & 16960 MiB & $0.67 \pm 0.12$ s  & 826 MiB  \\
$H{=}50$ & n/a & n/a & $1.72 \pm 0.41$ s  & 1076 MiB \\
\hline
\end{tabular}
\label{tab:experiments:computetimespace}
\end{table}

\section{Conclusion}
In this work, we presented a novel method for ID relying on PMFGs. 
In particular, we set forth two desiderata in order to use scalable first-order optimization to approximately solve ID problems. 
Through new analyses, we demonstrated that these conditions hold in both classical Mean Field Game (MFG) settings and batched auction environments.
For both settings, we presented a unified algorithm, called \myalgname{}, which can solve a span of ID problems, such as congestion pricing or optimal auction design.
Overall, the AMID framework offers a flexible foundation for diverse incentive design applications, paving the way for future extensions.

\begin{ack}
This project was supported by Swiss National Science Foundation (SNSF) under the framework of
NCCR Automation and SNSF Starting Grant.
V. Thoma acknowledges funding from the Swiss National Science Foundation (SNSF) Project Funding No. 200021-207343 and 
is supported by an ETH AI Center Doctoral Fellowship.
\end{ack}

\bibliographystyle{plain}
\bibliography{ref}

\newpage
\setcounter{tocdepth}{2}
\tableofcontents

\appendix

\newpage

\section{Frequently-Used Notation}
\label{app:notation}

\begin{longtable}{ p{.23\textwidth}  p{.77\textwidth} } 
        \textbf{General Notation} & \\
        $\entropy(u)$ & $:= -\sum_x u(x) \log u(x)$, entropy \\
        $\operatorname{sigmoid}(x)$ & $:= \frac{1}{1 + e^{-x}}$, sigmoid function\\
        $\nabla f$ & $\in\mathbb{R}^{d_1\times d_2}$, Jacobian of function $f:\mathbb{R}^{d_2} \rightarrow \mathbb{R}^{d_1}$\\
        $\mathbb{S}_{D-1}$ & $:= \{x\in\mathbb{R}^D:\|x\|= 1\}$, $(D-1)$-dimensional unit sphere in $\mathbb{R}^D$\\
        $\mathbb{B}_D$ & $:= \{x\in\mathbb{R}^D:\|x\|\leq 1\}$, $D$-dimensional unit closed ball in $\mathbb{R}^D$\\
        $\mathbb{V}\text{ar}$ & variance of random variable \\
        $\vece_i$ & standard unit vector with $i$-th entry 1. \\
        $\Delta_{\setX}$ & $:= \{u\in\mathbb{R}^{\setX} : \sum_x u_x=1, u_x \geq 0\}$, probability simplex on $\setX$. \\
        $\kldiv(u|v)$& := $\sum_x u(x)\log\sfrac{u(x)}{v(x)}$, Kullback–Leibler divergence\\
        $\langle \cdot, \cdot \rangle$ & dot product.\\
        $\ind{}$ & indicator function. \\
        $\| \cdot \|_1$& $\ell_1$ norm (on Euclidean space $\mathbb{R}^D$)\\
        $\| \cdot \|_2$& $\ell_2$ norm (on Euclidean space $\mathbb{R}^D$)\\
        $\bigotimes_{i\in[N]} m^i$ & product measure \\
        $\empc{\vecx}$ & empirical counts of entries of some $\vecx\in\setX^N$,  $\empc{\vecx}:= \sfrac{1}{N}\sum_{i=1}^N \vece_{x_i}$ \\
        $\rho(A)$ & spectral radius of matrix $A\in \mathbb{R}^{D,D}$ \\
        $\pmarg_\setY (d)$ & $:= \sum_{x\in\setX} d(x, \cdot) \in \Delta_\setY$, for $d \in \Delta_{\setX\times\setY}$.\\
        $\|A\|_{p\rightarrow q}$ & $:= \sup_{\|x\|_p \leq 1} \| A x \|_q$, operator norm of $A$ with norms $\|\cdot\|_p, \|\cdot\|_q$\\
         & \\
        \textbf{Generic PMFGs} & \\
        $\mathcal{G}$ & parameterized MFG \\
        $H$ & horizon (number of rounds) of auction\\
        $\setS$ & (finite) state space\\
        $\setA$ & (finite) action space \\
        $N$ & number of players/agents \\
        $g$ & ID objective \\
        $G$ & := $\Exop[ g(\theta, \widehat{\vecpop})]$, true ID objective in $N$-player game\\
        $\Nash^\tau(\setG)$ & set of $\tau$-regularized Nash equilibria \\
        $\Theta$ & parameter space \\
        $\theta$ & $\in \Theta$, ID design parameter\\
        $\Pi_H$ & $\{\pi:[H]\times\setS \rightarrow \Delta_\setA\}$, set of finite-horizon Markovian policies.\\
        $\mu_0$ & initial state distribution \\
        $\widebar{P}_{h, \theta}$ & parameterized state transition dynamics in $N$-player DG \\
        $\widebar{R}_{h, \theta}$ & parameterized reward functions in $N$-player DG \\
        $\tau$ & $\in\mathbb{R}$ entropy regularization magnitude \\
        $\vecpi$ & $\in\Pi_H^N$, $N$-tuples of policies \\
        $\Jdg^{\tau,i}_\setG$ & expected reward of player $i$ in dynamic game $\setG$ (see \Cref{def:dg})\\
        $\Expdg^{\tau, i}_\setG$ & exploitability of player $i$ in dynamic game $\setG$ (see \Cref{def:dg}) \\
        $\Expdg^{\tau}_\setG$ & maximum exploitability in dynamic game $\setG$ (see \Cref{def:dg}) \\
        $\setM$ & parameterized mean-field game (PMFG) \\
        $P_{h, \theta}$ & parameterized state transition dynamics in PMFG \\
        $R_{h, \theta}$ & parameterized reward functions in PMFG\\
        $\gpop_h$ & one-step MFG forward flow (see \Cref{def:mfg})\\
        $\lpop$ & maps policies in $\Pi_H$ to $H$-step mean-field population flow in $\Delta_{\setS\times\setA}^H$ (see \Cref{def:mfg}) \\
         &  \\ 
        \textbf{Adjoint method:} & \\
        $\Vmfg^\tau_h$ & state value function in the PMFG \\
        $\qmfg^\tau_h $ & q-value value function (on state-action pairs) in the PMFG \\
        $\Fupdate$ & $:\mathbb{R}^{[H]\times\setS\times\setA}\times \Theta \rightarrow\mathbb{R}^{[H]\times\setS\times\setA}$ generic policy update operator for computing NE, defined in log policy space\\
        $\Fomdpi$ & mirror descent update, in policy space \\
        $\Fomd$ & mirror descent update, in log-policy space \\
         &  \\ 
        \textbf{For auctions:} & \\
        $\Gauc$ & batched Auction\\
        $\perp$ & inactive state for agent not participating in the current round\\
        $\setV$ & value space\\
        $\alpha_{\max}$ & parameter for maximal amount of goods\\
        $\alpha^{\theta}_h$ & parametrized item allocation function\\
        $\paymentfn^{\theta}_h$ & parametrized payment function\\
        $\utilityfn_h$ & utility function of bidders\\
        $\transitionfn_h$ & valuation dynamics function\\
        $P^{\text{auc}}_{h, \theta}$ & transition dynamics in batched auction\\
        $R^{\text{auc}}_{h, \theta}$ & reward functions\\
        $\Jauc^{\tau, i}$ & entropy regularized ($\tau$) sum of rewards for agent $i \in [N]$\\
        $\Expauc^{\tau, i}$ & entropy regularized ($\tau$) exploitability for agent $i \in [N]$\\
        $\Mauc$ & batched auction MFG\\
        $\Pauc_{h, \theta}$ & BA-MFG transition dynamics\\
        $\Rauc_{h, \theta}$ & BA-MFG reward functions\\
        $\nu^{-\perp}$ & action marginal of active (non-$\perp$) states; maps state-action distributions to sub-probability distributions over actions\\
        $\pwin$ & winning probability function, given a bid $a$, sold goods $\alpha$ and population bid distribution $\nu^{-\perp}$\\
        $\xi^{L, \theta}$ & post allocation state distribution\\
        $\Vmfa^{\tau}$ & entropy regularized sum of rewards for BA-MFG\\
        $\Expmfa^{\tau}$ & entropy regularized exploitability of rewards for BA-MFG\\
        $\lpopmfa$ & population operator for BA-MFG\\
        $\objrev$ & revenue objective\\ 
\end{longtable}

\section{Extended Related Works}
\label{app:related_works}

\paragraph{Related works on Equilibrium Computation in Auctions.}

As outlined in \Cref{app:auction_eq}, many real-world auctions are not strategyproof. It is thus important to evaluate them at equilibrium---both from a predictive (how bidders will likely behave), as well as from a normative (bidding recommendations) standpoint. 
While some simple formats have been solved analytically \cite{Krishna2002AuctionTheory}, existing hardness results for computing exact equilibria in auctions \cite{Cai2014Simultaneous} motivate approximate computational approaches. 
In single-round auctions, a strand of work has used iterated best-response computations to calculate equilibria \cite{Reeves2004Computing,Vorobeychik2008Stochastic,Rabinovich2013Computing,bosshard2017fastBNE,Bosshard2020JAIR}. 
Other approaches rely on gradient descent \cite{Bichler2023Computing,Kohring2023Enablinga} or deep learning \cite{Bichler2021Learning,Martin2023Finding}.
For multi-round auctions, \cite{Thoma2025Computing} compute $\varepsilon$-perfect Bayesian equilibria, using best response dynamics. 
Others have used (deep) RL to find approximate Nash equilibria \cite{Greenwald2012Approximating,Pieroth2023Equilibrium,dEon2024Understandinga,Thoma2023Learning}. 
In our work, by using mean-field approximations, we circumvent the curse of dimensionality inherent to these multi-agent RL approaches. Using mean-field approaches to solve auctions has been explored previously to some extent. \cite{Balseiro2015Repeated,guo2019learning} study specific repeated ad auctions with budget constraints, and \cite{Iyer2014Mean} studies dynamic auctions, where bidders iteratively learn about their own type.

\paragraph{Other MFG works.}
Stackelberg equilibria for MFGs have also been studied in the particular case of linear-quadratic models \cite{moon2018linear, bensoussan2015mean}.
Another relevant model in this setting is mean-field incentive design with major and minor players \cite{sanjari2025incentive}, where designing incentives for a leader is studied for the purpose of influencing a population.
In continuous time, Stackelberg MFGs have been studied in applications such as regulating carbon markets \cite{carmona2022mean}, epidemics \cite{aurell2022optimal}, and advertising markets \cite{carmona2021mean}.

\section{Automated Mechanism Design as \ref{problem:n_player_design}}
\label{app:auction_eq}
We note that in contrast to our approach many works on automated mechanism design focus on designing strategyproof auctions, i.e. auctions in which bidders bid truthful in equilibrium, relying on the so-called \textit{revelation principle}\cite{Curry2024Automated,Conitzer2004Self,Likhodedov2004Methods}. The revelation principle states that any non-strategyproof equilibrium of an auction can be implemented as an outcome equivalent strategyproof equilibrium of an adapted auction \cite{Krishna2002AuctionTheory}.
Restricting to strategyproof mechanisms bypasses the need to differentiate through an equilibrium. Instead of a problem like \eqref{problem:n_player_design} where the outer objective depends on an inner equilibrium solution, the inner solution is already known--bid truthfully--and instead the problem becomes one of constrained optimization problem, where the so-called \emph{incentive compatibility} (IC) constraints ensure that bidding truthfully is in fact an equilibrium \cite{Conitzer2004Self}.

While restricting to strategyproof mechanisms foregoes the need to differentiate through an equilibrium, many real-world auctions are not strategyproof. In fact, in 2019 Google for example deliberately changed towards non-strategyproof first price auctions for selling ads, citing the increased transparency of simple, non-strategyproof format for the bidders \cite{Despotakis2021First}. In such cases without IC constraints, the question is how bidders will respond in equilibrium and in turn designing revenue-optimal auctions becomes an instance of \eqref{problem:n_player_design}, which we tackle in \Cref{sec:auctions_maintext,sec:experimental_results}.

\section{Mean-Field Mechanism Design}
\label{app:meanfieldmechanismdesign}

\subsection{Preliminary Lemmas}

\begin{theorem}[Rademacher]
\label{theorem:rademacher}
    Let $U\subset\mathbb{R}^m$ be an open set and $f:U \rightarrow \mathbb{R}^n$ be a Lipschitz continuous map.
    Then, $f$ is almost everywhere differentiable on $U$, that is, the points on which $f$ is not differentiable on $U$ for a set of measure 0.
\end{theorem}

In some cases, an explicit differentiability assumption might be useful for PMFG dynamics, which we state below.
\begin{assumption}[Differentiability]
\label{assumption:differentiable}
    For all $s,s'\in\setS, a\in\setA$, the functions $P_\theta(s'|s,a,L), R_\theta(s,a,L)$ are differentiable on $\theta, L$.
    Furthermore, $g(\theta, \vecpop)$ is differentiable on $\theta, \vecpop$ with bounded derivatives.
\end{assumption}

In particular, the following simple result is useful for the derivation of \myalgname{}.
\begin{lemma}[Differentiability of operators]
\label{lemma:differentiability}
    The maps $\gpop, \lpop, \qmfg^\tau_h, \Vmfg^\tau_h,\Fomdpi $ as well as the map $\theta, \pi \rightarrow g(\theta, \Lambda(\pi))$ are almost everywhere differentiable under \Cref{assumption:lipschitz} and differentiable everywhere under \Cref{assumption:differentiable}.
\end{lemma}
\begin{proof}
    This result is a straightforward result of the definitions of the mentioned operators, in particular, when \Cref{assumption:differentiable} is taken, the mentioned operators are also differentiable as they are the compositions of differentiable functions. 
    In the case where only \Cref{assumption:lipschitz} holds, the above-mentioned functions are also Lipschitz on every bounded domain, which implies by \Cref{theorem:rademacher} that they are almost everywhere differentiable.
\end{proof}

Finally, we state the following standard lemma from past work on the approximation of MFG dynamics by finite player games.

\begin{lemma}
\label{app:main_technical_lemma_lipschitz}
Assume that the conditions of \Cref{theorem:mainapprox} hold for the PMFG $\setM$ and DG $\setG$, let $\pi, \widebar{\pi}\in \Pi_H$ be two arbitrary policies, and let $\theta\in\Theta$ be fixed.
Let $\vecpop^{\widebar{\pi}} = \{\pop_h^{\widebar{\vecpi}}\}_h = \lpop(\widebar{\pi}|\theta)$ be the population flow induced by $\widebar{\pi}$ on the PMFG with fixed parameter $\theta$.
Take the trajectories $s_h^i, a_h^i, \widehat{L}_h$ induced by the DG with parameter $\theta$ and policy profile $(\pi, \widebar{\pi}, \ldots, \widebar{\pi}) \in \Pi_H^N$.
Then,
    \begin{align*}
        \Exop [ \|\widehat{\pop}_h -\pop_h^{\widebar{\pi}}\|_1 ] \leq
            \frac{1-\lpopmu^{h+1}}{1-\lpopmu}
            |\setS| |\setA|\sqrt{\frac{2}{N}} + \frac{1}{N}
            \sum_{i=0}^{h-1} \lpopmu^{h-i-1}
            \Delta_h,
    \end{align*}
where $\Delta_h := \sup_s \| \widebar{\pi}_h(\cdot|s) - \pi_h(\cdot|s)\|_1$, and $\lpopmu$ is a uniform bound on the Lipschitz moduli of $\Gamma_h$ in $L$.
Furthermore, denoting the random variables $s_h, a_h$ as the distributions of state-action pairs in the PMFG dynamics with population flow $\vecpop^{\widebar{\pi}}$ and policy $\pi$,
\begin{align*}
    \| \Prob[s_h=\cdot, a_h=\cdot] - \Prob[s_h^1=\cdot, a_h^1=\cdot] \|_1 \leq K_L \sum_{h'=0}^h\Exop [ \|\widehat{\pop}_h -\pop_h^{\widebar{\pi}}\|_1 ],
\end{align*}
where $K_L$ is the Lipschitz modulus of transition dynamics $P_{\theta,h}$ in $L$.
\end{lemma}
\begin{proof}
    The proof is a straightforward extension of the approximation results due to \cite{yardim2024meanfield}, to the case where transition dynamics and rewards also depend on the mean-field flow over actions.
    See in particular Theorem 3.2 in \cite{yardim2024meanfield}.
\end{proof}

\subsection{Proof of \texorpdfstring{\Cref{theorem:mainapprox}}{} }

First, we present the bound on exploitability.
As in the setting of \Cref{app:main_technical_lemma_lipschitz}, let $\vecpop^{\widebar{\pi}} = \{\pop_h^{\widebar{\vecpi}}\}_h = \lpop(\widebar{\pi}|\theta)$ be the population flow induced by $\widebar{\pi}$ on the PMFG with fixed parameter $\theta$.
Take the trajectories $s_h^i, a_h^i, \widehat{L}_h$ induced by the DG with parameter $\theta$ and policy profile $\vecpi:=(\pi, \widebar{\pi}, \ldots, \widebar{\pi}) \in \Pi_H^N$.
Similarly, denote the random variables $s_h, a_h$ as the distributions of state-action pairs in the PMFG dynamics with population flow $\vecpop^{\widebar{\pi}}$ and policy $\pi$.
\begin{align*}
    &|\Vmfg^\tau_\setM \left( \vecpop^{\widebar{\pi}}, \pi | \theta \right) -  \Jdg^{\tau,1}_\setG(\vecpi | \theta)| \\
    &= \Big|\sum_{h,s,a} \Prob[s_h=s, a_h=a] (R_{h,\theta}(\pop_h^{\widebar{\pi}},s,a|\theta) + \tau \entropy(\widebar{\pi}_h(s)) ) \\
        &\quad - \sum_{h,s,a,L} \Prob[s_h=s, a_h=a, \widehat{L}_h=L] (R_{h,\theta}(L,s,a|\theta) + \tau \entropy(\widebar{\pi}_h(s)) ) \Big | \\
    &\leq \Big|\sum_{h,s,a,L} \Prob[s_h=s, a_h=a, \widehat{L}_h=L] (R_{h,\theta}(L,s,a|\theta) - R_{h,\theta}(\pop_h^{\widebar{\pi}},s,a|\theta))\Big| \\
        &\quad + |\sum_{h,s,a} ( \Prob[s_h^i=s, a_h^i=a] - \Prob[s_h=s, a_h=a] )(R_{h,\theta}(\pop_h^{\widebar{\pi}},s,a|\theta) + \tau\entropy(\widebar{\pi}_h(s))) |.
\end{align*}
Since $|R_{h,\theta}(\pop_h^{\widebar{\pi}},s,a|\theta) + \tau\entropy(\widebar{\pi}_h(s))| \leq 1 + \tau \log|\setA|$, and $R_{h,\theta}$ is Lipschitz in $L$ (say with modulus $K$), it holds that
\begin{align*}
    &|\Vmfg^\tau_\setM \left( \vecpop^{\widebar{\pi}}, \pi | \theta \right) - \Jdg^{\tau,1}_\setG(\vecpi | \theta)| \\
    &\leq K \sum_{h}\Exop[\|\widehat{\pop}_h -\pop_h^{\widebar{\pi}}\|_1 ] + \sum_h \| \Prob[s_h=\cdot, a_h=\cdot] - \Prob[s_h^1=\cdot, a_h^1=\cdot] \|_1 (1+\tau \log |\setA|) \leq \mathcal{O}(\sfrac{1}{\sqrt{N}}),
\end{align*}
by an application of \Cref{app:main_technical_lemma_lipschitz}.
The corresponding bound on $\Expdg^\tau_\setG(\vecpi^*|\theta)$ is obtained by maximizing $\widebar{\pi}$ as the best response to the population strategy profile in the DG, and setting $\pi$ to be an MFG-NE $\pi^*$ for the parameter $\theta$.

We also prove a similar bound for the objective value.
Assume now that all $N$ players play policy $\pi^*$ in the DG.
Then, since $g$ is Lipschitz continuous,
\begin{align*}
    |G(\theta, \vecpi^*) - g(\theta, \lpop(\pi^*|\theta))| &\leq \Exop [ | g(\theta, \{\widehat{L}_h\}_{h=0}^{H-1}) - g(\theta, \{\widehat{L}_h\}_{h=0}^{H-1})| ] \\
    &\leq  \Exop [ \sum_h K_h \| \widehat{L}_h - \pop_h^{\widebar{\pi}} \|_1 ],
\end{align*}
where we denote the Lipschitz modulus of $g$ with $L_h$ as $K_h$ (with respect to norm $\|\cdot\|_1$).
An application of the technical lemma \Cref{app:main_technical_lemma_lipschitz} yields then the $\mathcal{O}(\sfrac{1}{\sqrt{N}})$ upper bound.

\subsection{Proof of \texorpdfstring{\Cref{lemma:omd_diff}}{}}

We first prove that if $\lim_{T \rightarrow \infty} \Fomd^{(T)} (\vectheta', \veczeta)$ exists and $\Fomd^\infty(\vectheta', \veczeta) = \veczeta ^*$, the $\softmax(\veczeta ^*) \in \Nash^\tau_\setM(\vectheta')$.
Since $\Fomd$ is a continuous function, it must hold that
\begin{align*}
    \Fomd(\vectheta', \veczeta^*) = \Fomd(\vectheta', \lim_{T \rightarrow \infty} \Fomd^{(T)} (\vectheta', \veczeta)) = \Fomd^\infty(\vectheta', \veczeta) = \veczeta^*,
\end{align*}
therefore $\Fomd(\vectheta', \veczeta^*) = \veczeta^*$.
Denoting $\pi^*_h(\cdot|s) := \softmax(\veczeta^*(h,s,\cdot))$, we then have the relations
\begin{align*}
    \Fomdpi(\theta, \veczeta^*)(h, s,a) &:= (1 - \eta\tau) \veczeta^*(h,s,a) + \eta \qmfg^\tau_h (s,a | \lpop(\pi^*), \pi^*, \theta), \\
    \Fomdpi(\theta, \pi)(h, s) &:= \argmax_{u \in \Delta_\setA} \langle \qmfg^\tau_h (s,\cdot | \lpop(\pi), \pi, \theta), u\rangle + \tau \entropy(u) - \eta^{-1}(1 - \tau\eta) \kldiv(u | \pi(s)).
\end{align*}
Then, for any $h$, it holds that $\veczeta^*(h,s,a) = \tau^{-1}\qmfg^\tau_h (s,a | \lpop(\pi^*), \pi^*, \theta)$.
We show that $\pi^*$ is then the best response to $\lpop(\pi^*)$ by backward induction.
Denote for convenience $\vecpop^* := \{\pop_h^*\}_h = \lpop(\pi^*)$.

At time $H-1$, $\veczeta^*(H-1,s,a) = \tau^{-1}R_{h,\theta}(s,a, \pop_{H-1}^*)$ and $\pi^*_{H-1}(a|s) = \frac{\exp\{ \tau^{-1}R_{h,\theta}(s,a, \pop_{H-1}^*)\}}{\sum_{a'}  \tau^{-1}R_{h,\theta}(s,a', \pop_{H-1}^*)}$, therefore, by first order optimality conditions, $\pi^*_{H-1}(\cdot|s)$ maximizes (uniquely) the strongly concave function $u \rightarrow \langle u, \tau^{-1}R_{h,\theta}(s,a, \pop_{H-1}^*)\rangle +\tau\entropy(u)$ on the simplex $\Delta_\setA$.
Hence, at every state $s$ at time $H-1$, the policy $\pi^*_{H-1}(\cdot|s)$ is optimal.
Assume now $\pi^*_{h'}(\cdot|s)$ is optimal for all $h' \geq h$ for some $h \in [H]$.
Then, by the inductive assumption, $\qmfg^\tau_{h'} (s,a | \lpop(\pi^*), \pi^*, \theta)$ is the optimal regularized q function for all $h' \geq h$.
Since $\veczeta^*(h,s,a) = \tau^{-1}\qmfg^\tau_{h} (s,a | \lpop(\pi^*), \pi^*, \theta)$, once again by first order optimality conditions, $\pi^*_h$ is also the optimal policy at time $h$.

We move on to the convergence of derivatives.
Firstly, since $q_h^\tau$ is given to be $C^1$ in a neighborhood of $(\theta, \zeta^*)$ and $\rho(\partial_{\veczeta} \qmfg_{\cdot}^\tau(\cdot,\cdot|\lpop(\softmax(\veczeta^*)), \softmax(\veczeta^*), \vectheta )) < \tau$, there exists an open neighborhood $\widebar{U}$ of $(\theta, \zeta^*)$ where $\rho(\partial_{\veczeta} \qmfg_{\cdot}^\tau(\cdot,\cdot|\lpop(\softmax(\veczeta')), \softmax(\veczeta'), \vectheta' )) < \tau - \delta$ for all $(\vectheta', \veczeta')\in \widebar{U}$ for some $\delta_1>0$.
Then, since 
\begin{align*}
    \Fomd(\vectheta, \veczeta)(h, s, a) := (1 - \eta \tau) \zeta_{h, s,a} + \eta  \qmfg^\tau_h (s,a | \lpop(\softmax(\veczeta)), \softmax(\veczeta), \vectheta),
\end{align*}
on $\widebar{U}$ it also holds that $\rho(\partial_{\veczeta}\Fomd(\vectheta, \veczeta)) < 1 - \delta_2$ for some $\delta_2 > 0$.
On $\widebar{U}$, the map $\Fomd^\infty$ is implicitly defined by $\Fomd(\vectheta, \Fomd^\infty(\theta, \zeta)) = \Fomd^\infty(\theta, \zeta)$, therefore by the implicit function theorem $\Fomd^\infty$ is differentiable in $\theta$ and 
\begin{align*}
    \nabla_\theta \Fomd^\infty(\theta, \zeta) = (I - \partial_\zeta \Fomd(\theta, \Fomd^\infty(\theta, \zeta)))^{-1} \partial_\theta\Fomd(\theta, \zeta), \quad \forall (\theta, \zeta) \in \widebar{U}.
\end{align*}

Fixing some direction $v\in\Theta$, define inductively the iterates
\begin{align*}
    \zeta_0 &:= \zeta', \quad \zeta_{t+1} := \Fomd(\theta, \zeta_t) \\
    x_0 &:= 0, \quad x_{t+1} := \partial_\zeta \Fomd(\theta, \zeta_t) x_{t} +  \partial_\theta \Fomd(\theta, \zeta_t) v. 
\end{align*}
Note that by the chain rule, $x_{t+1} = \nabla_u( \Fomd^{(T)}(\theta, \zeta') ) v $, therefore, if we show that $\lim_{t\rightarrow \infty} x_t = \nabla_\theta( \Fomd^{(\infty)}(\theta, \zeta') ) v$ for any choice of $v$, we are done.
Firstly, by the assumptions of the lemma, for sufficiently large $T_0$, it holds that $\zeta_t\in\widebar{U}$ for all $t > T_0$, and $\zeta^* = \lim_{t\rightarrow\infty} \zeta_t$.
Defining $x^* := (I -  \partial_\zeta \Fomd(\theta, \zeta^*))^{-1} \partial_\theta \Fomd(\theta, \zeta_t) v$, which satisfies
$x^* := \partial_\zeta \Fomd(\theta, \zeta^*)x^* + \partial_\theta \Fomd(\theta, \zeta_t) v$.
Therefore,
\begin{align*}
     \|x_{t+1} - x^*\| = &\|\partial_\zeta \Fomd(\theta, \zeta_t) x_{t} - \partial_\zeta \Fomd(\theta, \zeta^*) x^* +  \partial_\theta \Fomd(\theta, \zeta_t) v - \partial_\theta \Fomd(\theta, \zeta^*) v\| \\
     \leq & \|\partial_\zeta \Fomd(\theta, \zeta_t) x_{t} - \partial_\zeta \Fomd(\theta, \zeta^*) x_{t} \| 
      + \| \partial_\zeta \Fomd(\theta, \zeta^*) x_{t} - \partial_\zeta \Fomd(\theta, \zeta^*) x^*\| \\
     & + \| \partial_\theta \Fomd(\theta, \zeta_t) v - \partial_\theta \Fomd(\theta, \zeta^*) v\| ,
\end{align*}
which proves that $x_t \rightarrow x^*$.
By the implicit function theorem, $x^*$ is the gradient of $\Fomd^{(\infty)}$, concluding the proof.

\subsection{Discussion of Assumptions of \texorpdfstring{\Cref{lemma:omd_diff}}{}}

We briefly discuss the ramifications of \Cref{lemma:omd_diff}.
Firstly, the result is useful only assuming that the OMD iterates converge to an NE of the MFG.
NE computation in MFGs is a well-studied research topic on its own, and several positive results are known for various classes of MFGs.
We take this for granted in this lemma, as NE computation is not the main goal.

Next, we note that the lemma suggests that the derivative of \eqref{eq:tstepobjfunction} is a valid first-order oracle provided that the Jacobian of $\Fomd$ has bounded spectral radius around the NE induced by a $\theta$.
Importantly, for any PMFG, $T$-step objective approximates the fixed-point gradient provided that $\tau$ is sufficiently large.
For more structured settings, the result can be strengthened to permit $\tau=0$, for instance, the derivations in \cite{Li2020Enda} readily extend to the case where $H=1$, $|\setS|=1$, and the reward function $R$ is monotone in population distribution.
We leave as future work to generalize this for monotone PMFGs with dynamics (i.e., for PMFGs with $H>1$).

\subsection{Proof of \Cref{lemma:adjoint}}
\label{sec:appendix_adjoint}

Let $\Theta = \mathbb{R}^{D_1}, Z = \mathbb{R}^{D_2}$, and assume the functions $F:\Theta \times Z \rightarrow Z$ and $g:\Theta \times Z \rightarrow \mathbb{R}$ are differentiable.
We define $\ell: \Theta \rightarrow \mathbb{R}$ 
\begin{align}
    \ell(\vectheta) &:= g( \vectheta, \veczeta_{T+1}(\vectheta) ), \notag \\
    \veczeta_{t+1}(\vectheta) &:= (1 - \tau \eta) \veczeta_{t}(\vectheta) + \eta F(\vectheta, \veczeta_{t}(\vectheta)), \quad \forall t = 0, \ldots, T. \label{eq:problem_def_adjoint}
\end{align}
for some $\veczeta_0$ constant.
Clearly $\ell(\vectheta)$ is differentiable; our goal is to efficiently compute $\nabla_\vectheta\ell(\vectheta)$.

Define for any sequence of (row) vectors $\{\veca_t\}_{t=0}^T$ for $\veca_t \in \Theta$ the function 
\begin{align*}
    L(\vectheta, \left\{\veca_t\right\}_{t=0}^T ) := & g\left(\vectheta, \veczeta_{T+1}(\vectheta)\right)+\sum_{t=0}^T \veca_t \cdot\left(\veczeta_{t+1}-(1 - \eta\tau) \veczeta_t(\vectheta) - \eta F\left(\vectheta , \veczeta_t(\vectheta)\right)\right) \\
    = & g\left(\vectheta, \veczeta_{T+1}(\vectheta)\right) + \veca_T \cdot \veczeta_{T+1}(\vectheta) - (1 - \eta\tau) \veca_0 \cdot \veczeta_0(\vectheta) - \eta \sum_{t=0}^T \veca_t \cdot  F\left(\vectheta, \veczeta_t(\vectheta) \right) \\
        & + \sum_{t=1}^T \veczeta_t \cdot (\veca_{t-1} - (1-\eta\tau)\veca_t).
\end{align*}
Since for any $\{\veca_t\}_{t=0}^T$ it holds that $L(\vectheta, \left\{\veca_t\right\}_{t=0}^T ) = \ell (\theta)$, we have the identities
\begin{align*}
    \partial_\vectheta L(\vectheta, \left\{\veca_t\right\}_{t=0}^T ) &= \nabla_\vectheta \ell(\theta), \quad \partial_{\veca_t} L(\vectheta, \left\{\veca_t\right\}_{t=0}^T ) = 0
\end{align*}
Therefore, for any sequence of functions $\veca_t(\vectheta)$ that depend on $\vectheta$, we have
\begin{align*}
    \nabla_\vectheta\ell(\vectheta) = \partial_\vectheta L(\vectheta, \left\{\veca_t\right\}_{t=0}^T ) + \sum_t (\nabla_{\vectheta}\veca_t) \frac{\partial}{\partial \veca_t} L(\vectheta, \left\{\veca_t\right\}_{t=0}^T ) = \partial_\vectheta L(\vectheta, \left\{\veca_t\right\}_{t=0}^T ).
\end{align*}

We will therefore compute $\partial_\vectheta L(\vectheta, \left\{\veca_t\right\}_{t=0}^T )$ with a suitable choice of $\veca_t$.
By simple derivation:
\begin{align*}
    \partial_\vectheta L(\vectheta, \left\{\veca_t\right\}_{t=0}^T ) = & \partial_{\veczeta} g(\vectheta, \veczeta_{T+1}(\vectheta)) \nabla_{\vectheta} \veczeta_{T+1}(\vectheta) + \partial_{\vectheta} g(\vectheta, \veczeta_{T+1}(\vectheta))  \\
        & + \veca_T \nabla_{\vectheta} \veczeta_{T+1}(\vectheta) - (1-\eta\tau) \veca_0 \nabla_{\vectheta} \veczeta_{0}(\vectheta) \\
        & - \eta \sum_{t=0}^T \veca_t \partial_{\vectheta} F (\vectheta, \veczeta_{t}(\vectheta)) - \eta \sum_{t=0}^T \veca_t \partial_{\veczeta} F (\vectheta, \veczeta_{t}(\vectheta)) \nabla_{\vectheta} \veczeta_{t}(\vectheta)  \\
        & + \sum_{t=1}^T  (\veca_{t-1} - (1-\eta\tau)\veca_t) \nabla_{\vectheta} \veczeta_{t}(\vectheta)
\end{align*}
Since $\nabla_\vectheta \veczeta_0 = 0$, as $\veczeta_0$ is constant, rearranging the terms yields:
\begin{align*}
    \partial_\vectheta L(\vectheta, \left\{\veca_t\right\}_{t=0}^T ) = & (\nabla_{\veczeta} g(\veczeta_{T+1}(\vectheta)) + \veca_T) \nabla_{\vectheta} \veczeta_{T+1}(\vectheta) - \eta \sum_{t=0}^T \veca_t \partial_{\vectheta} F (\vectheta, \veczeta_{t}(\vectheta)) + \partial_{\vectheta} g(\vectheta, \veczeta_{T+1}(\vectheta)) \\
        & + \sum_{t=1}^T  (\veca_{t-1} - (1-\eta\tau)\veca_t - \eta \veca_t \partial_{\veczeta} F (\vectheta, \veczeta_{t}(\vectheta))) \nabla_{\vectheta} \veczeta_{t}(\vectheta).
\end{align*}
Therefore, we pick
\begin{align*}
    \widebar{\veca}_T &:= -\nabla_{\veczeta} g(\veczeta_{T+1}(\vectheta)) \\
    \widebar{\veca}_{t-1} &:= (1-\eta\tau)\widebar{\veca}_{t} + \eta \widebar{\veca}_t \partial_{\veczeta} F (\vectheta, \veczeta_{t}(\vectheta)),
\end{align*}
we obtain the equality
\begin{align}
    \label{eq:adjoint_derivative}
    \nabla_\vectheta \ell(\vectheta) = \partial_\vectheta L(\vectheta, \left\{\widebar{\veca}_t\right\}_{t=0}^T ) = - \eta \sum_{t=0}^T \veca_t \partial_{\vectheta} F (\vectheta, \veczeta_{t}(\vectheta)) + \partial_{\vectheta} g(\vectheta, \veczeta_{T+1}(\vectheta)).
\end{align}
Most importantly, \Cref{eq:adjoint_derivative} permits the computation of $\nabla_\vectheta \ell$ with $T$ with only caching the variables $\veczeta_t$.
Namely, the forward system \Cref{eq:problem_def_adjoint} can be used to iteratively compute $\veczeta_{T+1}$, and the backward system only requires evaluating Jacobians at the current step, meaning other than $\veczeta_t$, the memory requirements are kept constant.

In case memory is a bottleneck, $\veczeta_t$ can be cached during the forward step every $\beta\sqrt{T}$ steps for some constant $\beta$, meaning only $\mathcal{O}(\sqrt{T})$ memory is needed.
In this case, during the backward step, $\veczeta_t$ will need to be recomputed every $\mathcal{O}(\sqrt{T})$ steps for $\mathcal{O}(\sqrt{T})$ OMD iterations, maintaining the $\mathcal{O}(T)$ time complexity.
$\beta$ will introduce a tradeoff between time and space complexity.

\subsection{Extension of \Cref{lemma:adjoint} to General Mirror Descent}
\label{sec:appendix_adjoint_bregman}

\Cref{lemma:adjoint} analyzes the ajoint method under the specific entropy regularization scheme.
Let $h: \mathbb{R}^{D_1} \rightarrow \mathbb{R}$ be a strongly convex distance generating function, and $\nabla h :\mathbb{R}^{D_1} \rightarrow \mathbb{R}^{D_1}$ the corresponding mirror map.
Then, the generalized mirror descent update rule can be written as (in policy space $\Pi_H \subset \mathbb{R}^{[H]\times\setS\times\setA}$):
\begin{align*}
     \nabla h(\pi_{t+1}) &:= (1 - \tau \eta) \nabla h(\pi_{t}) + \eta F(\theta, \pi_{t}(\theta)), \quad \forall t = 0, \ldots, T. 
\end{align*}
Defining the iterates $\veczeta_t := \nabla h(\pi_{t+1}) $, we obtain the similar update rule
\begin{align*}
\ell(\vectheta) & := g( \vectheta, (\nabla h)^{-1}( \veczeta_{T+1}(\vectheta)) ), \\
    \veczeta_{t+1}(\vectheta) &:= (1 - \tau \eta) \veczeta_{t}(\vectheta) + \eta F(\vectheta, (\nabla h)^{-1}(\veczeta_{t}(\vectheta))), \quad \forall t = 0, \ldots, T,
\end{align*}
which reduces to the case analyzed in \Cref{sec:appendix_adjoint}, by the definitions
\begin{align*}
    \widebar{F}(\vectheta, \veczeta_{t}(\vectheta))
 := F(\vectheta, (\nabla h)^{-1}(\veczeta_{t}(\vectheta))) \\
    \widebar{g}( \vectheta, \veczeta_{T+1}(\vectheta) )
 := g( \vectheta, (\nabla h)^{-1}( \veczeta_{T+1}(\vectheta))).
\end{align*}

\section{Results on Batched Auctions}
\label{sec:batched_auction_results}

This section presents the formal analysis of parameterized batched auctions introduced in the main text.
We also provide rigorous statements and complete the proofs left out in the main text.

\Cref{sec:appendix_ba_definitions} provides formal definitions of the settings omitted in the main text and useful auxiliary constructions to assist the proofs.
\Cref{subsection:auxiliary_lemmas} presents a collection of auxiliary lemmas that support the subsequent analysis. 
In \Cref{section:upper_bound_exploitability_proof}, we prove the first part of \Cref{theorem:approximation_auction}, establishing an upper bound on the exploitability of mean field policies under the no zero-dominance condition. 
The second part of the theorem, which addresses convergence of the mechanism-level objective, is proved in \Cref{section:objective_convergence}.
Finally, \Cref{section:quaclipschitz} provides the proof of \Cref{lemma:qauclipschitz}, which establishes the Lipschitz continuity of the (entropy regularized) q values under full-support policies.

\subsection{Extended Definitions}
\label{sec:appendix_ba_definitions}

\textbf{Additional useful notation.}
To streamline the proofs, we also introduce some useful notation.
For any arbitrary finite set $\setX$ and a scalar $\beta \in \mathbb{R}_{\geq 0}$ define the sets:
\begin{align*}
    \Delta_{\setX}^{\beta} &:= \{d\in \mathbb{R}^{\setX}:\: \forall x\in \setX:\, d(x)\geq 0,\, \sum_{x\in \setX}d(x) = \beta\}, \\
    \Delta_{\setX}^{\geq \beta} &:= \{d\in \mathbb{R}^{\setX}:\: \forall x\in \setX:\, d(x)\geq 0,\, \beta\leq  \sum_{x\in \setX}d(x) \leq  1\}, \\
    \Delta_{\setX}^{\leq \beta} &:= \{d \in \mathbb{R}^{\setX}:\: \forall x\in \setX:\, d(x)\geq 0,\, \sum_{x\in \setX} d(x) \leq  \beta\}.
\end{align*}
For $x \in \mathbb{R}^\setX$, where $\setX$ has a total order, denote the cumulative mass function $S_x(d) := \sum_{x' \geq x} d(x')$.
For some $d\in \Delta_{\setX\times\setY}$, define the marginal distribution
\begin{align*}
    \pmarg_\setY (d) := \sum_{x\in\setX} d(x, \cdot) \in \Delta_\setY.
\end{align*}
For $\alpha\in[0,1]$, define the \emph{threshold bid operator} $\threshact_\alpha: \Delta_{\setA}^{\geq 0} \rightarrow \setA$ as 
\begin{align*}
    \threshact_\alpha(d) := \max \left( \, \{ a \in \setA : S_a(d) \geq \alpha, d(a) > 0 \} \cup \{a_{0}\} \right),
\end{align*}
where $a_0$ is the smallest element of $\setA$.
For $d \in \Delta_{\setV\times\setA}^{\geq \alpha}$ we also define $\threshact_\alpha(d) := \threshact_\alpha( \pmarg_\setS (d))$.

We define the operator $ \Xi_\alpha: \Delta_{\mathcal{V} \times \mathcal{A}}^{\geq \alpha} \rightarrow \mathbb{R}^{\mathcal{V} \times \mathcal{A}} $, which maps a matrix state-action distribution $ L $ to a matrix $ L' $, where each entry $ L'(s, a) $ represents the expected probability mass of agents in state $ s $ who chose action $ a $ and did not win an item this round, given that $ \alpha $ goods are allocated. 
Formally:
\begin{equation*}
    \Xi_\alpha(d)(s,a) = 
    \begin{cases}
        0 & \text{if } a > \threshact(d), \alpha > 0\\
          \frac{\sum_{s'\in \setV, a' \geq \threshact_\alpha(d)}d(s',a') -\alpha}{\sum_{s'\in \setV}d(s',\threshact_\alpha(d))}d(s, \threshact_\alpha(d)) & \text{if } a = \threshact(d), \alpha > 0\\
          d(s,a) & \text{otherwise}
    \end{cases}
\end{equation*}
$\Xi_\alpha$ is well-defined as $\sum_{s'\in \setV}d(s',\threshact_\alpha(d)) > 0$ whenever $\alpha > 0$ by definition.

Using the operator $ \Xi_\alpha $, we define the post allocation operator $ \Upsilon_\alpha: \Delta_{\mathcal{V} \times \mathcal{A}}^{\geq \alpha} \rightarrow \Delta_{\mathcal{V}} $, which maps a sub-probability distribution $d$ over active state-action pairs to the post allocation state distribution.
That is,
\begin{align*}
    \Upsilon_\alpha(L)(s) := \sum_{a \in \mathcal{A}} \Xi_\alpha(L)(s, a).
\end{align*}

The operator $ \Gamma_\alpha $, which governs the transition of the population state distribution after item allocation, can be expressed in terms of $ \Upsilon_\alpha $ by explicitly accounting for the inactive state $ \perp $. 
Specifically, let $ d \in \Delta_{\mathcal{S} \times \mathcal{A}} $ be a state-action distribution and assume $ \alpha(\nu^{-\perp}(d)) \leq \| \nu^{-\perp}(d) \|_1 $. Then $ \Gamma_\alpha: \Delta_{\mathcal{S} \times \mathcal{A}} \rightarrow \Delta_{\mathcal{S}} $ is defined as
\[
\Gamma_\alpha(d)(s) := 
\begin{cases}
\Upsilon_{\alpha(\nu^{-\perp}(d))}(d_{-\perp})(s) & \text{if } s \in \mathcal{V}, \\
\sum_{a \in \mathcal{A}} d(\perp, a) + \alpha & \text{if } s = \perp,
\end{cases}
\]
where $d_{-\perp}$ denotes the restriction of $d$ to active states $\setV$.

The BA-MFG is a parametrized PMFG, where both the payment function and the allocation threshold function are parameterized as $\alpha_h^\theta$ and $\paymentfn_h^\theta$.
Throughout this section, we assume these parameters are fixed and omit them, writing without loss of generality $\alpha_h, \paymentfn_h$ instead.

\begin{definition}[Batched Auction MFG (BA-MFG)]
A Batched-Auction MFG (BA-MFG) is a MFG defined by the tuple $\Mauc = (\setS, \setA, H, \mu_0,  \{\Pauc_h\}_{h= 0}^{H-1}, \{\Rauc_h\}_{h=0}^{H-1})$ of discrete state space $\setS = \setV \cup \{\perp\}$, discrete action space $\setA$, horizon $H \in \mathbb{N}_{>0}$, initial distribution $\mu_0 \in \setS$, transition dynamics $\Pauc_h: \setS \times \setA \times \Delta_{\setS \times \setA} \to \Delta_{\setS}$ and reward functions $\Rauc_h: \setS \times \setA \times \Delta_{\setS \times \setA} \to [0,1]$. 
The transition dynamics $\Pauc$ and the reward functions $\Rauc$ depend from the allocation functions $\{\alpha_h\}_{h=0}^{H-1}$, the dynamic functions $\{\transitionfn_h\}_{h=0}^{H-1},$ the payment functions $\{\paymentfn_h\}_{h=0}^{H-1}$ and the utility functions $\{\utilityfn_h\}_{h=0}^{H-1}$. 
Define the ``winning probability'' as
\begin{equation*}
    \pwin(s,a, L, \alpha) =  \ind{s\neq \perp}\begin{cases}
        1 & \text{if } \sum_{s'}\sum_{a'\geq a}L(s', a') \leq \alpha\\
        0 &\text{if } \sum_{s'}\sum_{a'> a}L(s', a') \geq \alpha\\
        \frac{\alpha - \sum_{s'} \sum_{a'> a}L(s', a')}{\sum_{s'}L(s, a)} & \text{otherwise}
    \end{cases}.
\end{equation*}
For the allocation $\alpha$ and valuation transition $\transitionfn$, 
define the operators $\Gamma_{\alpha}: \Delta_{\setS\times\setA} \rightarrow \Delta_\setS, \Gamma_{\transitionfn}: \Delta_\setS \times \Delta_\setS$ as
\begin{align*}
    \Gamma_\alpha(L)(s) &:= 
    \begin{cases}
        \sum_{a}L(s,a)(1 - \pwin(s,a, L, \alpha(\nu^{-\perp}(L)))) & \text{if } s \neq \perp\\
        L(\perp) + \sum_{s'}\sum_{a'} L(s',a')\pwin(s',a', L, \alpha(\nu^{-\perp}(L))) &\text{otherwise}
    \end{cases}, \\
    \Gamma_{\transitionfn}(\xi) &:= \sum_{z \in \setS}\xi(z)\transitionfn(\cdot \vert z, \xi),
\end{align*}
Define also $\Gamma_h, \lpopmfa$ as \footnote{Note that this way of writing $\Gamma_h, \lpopmfa$ is consistent with the MFG definition.}
\begin{align*}
    \Gamma_h(L, \pi)(s,a) &:= \Gamma_{\transitionfn_h}(\Gamma_{\alpha_h}(L))(s) \pi(a\vert s),\\
    \lpopmfa(\pi) &:= \{\Gamma_{h-1}(... \Gamma_1(\Gamma_0(\mu_0 \cdot \pi_0, \pi_1), \pi_2)..., \pi_{h-1})\}_{h=0}^{H-1}.
\end{align*}
The transition probability and rewards can be written as:
\begin{align*}
    \Pauc_h(s' \vert s, a, L) &:= \pwin(s,a, L, \alpha_h(\nu^{-\perp}))\transitionfn_h(s' \vert \perp, \Gamma_{\alpha_h}(L))\\
    &\quad + (1-  \pwin(s,a, L, \alpha_h(\nu^{-\perp})))\transitionfn_h(s' \vert s, \Gamma_{\alpha_h}(L))\\
    \Rauc_h(s,a,L) &:= \pwin(s,a,L, \alpha_h(\nu^{-\perp}(L)))\utilityfn_h(s, \paymentfn_h(a, \nu^{-\perp}(L))).
\end{align*}
For $\pi \in \Pi_H, \tau \geq 0$ and $\vecpop = \{L_h\}_{h=0}^{H-1}$, the total expected (entropy regularized) reward is
\begin{equation*}
    \Vmfa^\tau(\vecpop, \pi) := \mathbb{E}\left[\sum_{h=0}^{H-1}\Rauc_h(s,a, L_h) + \tau \entropy(\pi_h(s_h)) \middle\vert \substack{s_0 \sim \mu_0, a_h \sim \pi(s_h)\\
    s_{h+1} \sim \Pauc_h(s_h, a_h, L_h) } \right]
\end{equation*}
\end{definition}

For a policy $\pi \in \Pi_H$ we denote with $L^{\pi} = \lpopmfa(\pi)$ the induced state-action distribution, with $\mu^{\pi}$ the induced state distribution and with $\xi^{\pi}$ the induced hidden state distribution, with $\xi_h^{\pi} = \Gamma_{\alpha_h}(L_h^{\pi})$. 
Additionally, we define with $P_\alpha$ the transition probabilities associated with the item allocation dynamics
\begin{equation*}
    P_\alpha(s' \vert s,a, L) := \begin{cases}
        \pwin(s,a, L, \alpha(\nu^{-\perp}(L))) &\text{if }s' = \perp\\
        (1- \pwin(s,a, L, \alpha(\nu^{-\perp}(L)))) &\text{if }s' = s \\
        0 &\text{otherwise}
    \end{cases}. 
\end{equation*}
Then, the total expected (entropy regularized) reward can also be expressed as
\begin{equation*}
    \Vmfa^\tau(\vecpop, \pi) := \mathbb{E}\left[\sum_{h=0}^{H-1}\Rauc_h(s,a, L_h) + \tau \entropy(\pi_h(s_h)) \middle\vert \substack{s_0 \sim \mu_0, a_h \sim \pi(s_h), z_h \sim P_{\alpha_h}(s_h, a_h, L_h)\\
    s_{h+1} \sim \transitionfn_h(z_h, \Gamma_{\alpha_h}(L_h)) } \right]
\end{equation*}

\begin{definition}[$ N $-player Batched Auction ($N$-BA)]\label{definition:n_batched_auction}
    An $N$- player batched auction ($N$-BA) is a dynamic game $\Gauc = (N, \setS, \setA, H, \mu_0, \{P^{\text{auc}}_h\}_{h=0}^{H-1}, \{R^{\text{auc}}_h\}_{h=0}^{H-1}) $ of discrete state space $\setS = \setV \cup \{\perp\}$, discrete action space $\setA$, horizon $H \in \mathbb{N}_{>0}$, starting distribution $\mu_0 \in \setS$, transition dynamics $P^{\text{auc}}_h: \setS^N \times \setA^N \to \Delta_{\setS}^N$ and rewards functions $R^{\text{auc}}_h: \setS^N \times \setA^N  \to [0,1]^N$. The transition dynamics $P^{\text{auc}}_h$ and the reward functions $R^{\text{auc}}_h$ depend from the allocation functions $\{\alpha_h\}_{h=0}^{H-1}$, the dynamic functions $\{\transitionfn_h\}_{h=0}^{H-1},$ the payment functions $\{\paymentfn_h\}_{h=0}^{H-1}$ and the utility functions $\{\utilityfn_h\}_{h=0}^{H-1}$.
    Let $ \vecs= (s^1, \dots, s^N) \in \mathcal{S}^N $ denote the joint state of all agents, and $ \veca = (a^1, \dots, a^N) \in \mathcal{A}^N $ the joint action profile. Assume items are allocated to the top bidders according to the submitted actions, with uniform random tie-breaking at the allocation threshold. This induces a joint allocation probability kernel $ P_{N, \alpha}: \setS^N \times \setA^N \rightarrow \Delta_{\setS^N} $ formally defined as

    \begin{equation*}
    P_{N, \alpha}(\vecz \vert \vecs, \veca) =
        \begin{cases}
            \displaystyle
            \frac{1}{\binom{|\setT|}{\lfloor\alpha N \rfloor - |\setW|}} & \text{if } 
            \begin{aligned}
            & \forall i \in \setW \cup \{i \in [N]: s^{j} = \perp\},\; z^i = \perp, \\
            & \vert \{j \in \setT: z^{j} = \perp\}\vert = \lfloor\alpha(\nu)\cdot  N\rfloor - \vert \setW \vert,  \\
            &  \forall i \in \setL,\, z^{i} = s^{i}
            \end{aligned} \\[2.0em]
            0 & \text{otherwise}
        \end{cases},
    \end{equation*}
    
    where $\nu = \sum_{j}\vece_{a^j} \ind{s^j \not= \perp }$,   $a^* = \max \{a \in \setA: \sum_{j \in [N]} \ind{a^{j} \geq a, s^{j} \not= \perp} \geq \lfloor \alpha(\nu) N \rfloor\}$, $\setT = \{j \in [N]: a^{j} = a^*, s^{j} \not= \perp\}$, $\setW = \{j \in [N]: a^{j} > a^*, s^{j} \not= \perp\}$ and $\setL = \{j \in [N]: a^{j} < a^*, s^{j} \not= \perp\}$. The marginal winning probability and the rewards, similar as its corresponding MFG can be expressed as
    \begin{align*}
        &P_{N, \alpha}^{i}(\perp \vert \vecs, \veca) =: \sum_{\vecz}P_{N, \alpha}(\vecz \vert \vecs, \veca)\ind{z^i = \perp}=  \pwin(s^{i},a^{i}, \sigma(\vecs, \veca), \tfrac{\lfloor N \alpha(\nu^{-\perp}(\sigma(\vecs, \veca)))\rfloor}{N})\\
        &R_h^{\text{auc, i}}(\vecs, \veca) = \pwin(s^{i},a^{i},\sigma(\vecs, \veca), \tfrac{\lfloor N \alpha_h(\nu^{-\perp}(\sigma(\vecs, \veca)))\rfloor}{N})\utilityfn_h(s^{i}, \paymentfn_h(a^{i}, \nu^{-\perp}(\sigma(\vecs, \veca))))
    \end{align*}
    We define with $\Rauc_{N,h}$ the N-player discretization of $\Rauc_h$:
    \begin{equation*}
        \Rauc_{N,h}(s,a,L):= \pwin(s,a, L, \tfrac{\lfloor N \alpha_h(\nu^{-\perp}(L))\rfloor}{N})\utilityfn_h(s, \paymentfn_h(a, \nu^{-\perp}(L)))
    \end{equation*}
    Note that if $ L(s,a) = 0 $, the reward $ \Rauc_{N,h}(s,a,L) $ is not defined. Additionally, observe that 
    \[
    R_h^{\text{auc}, i}(\vecs, \veca) = \Rauc_{N,h}(s^{i}, a^{i}, \sigma(\vecs, \veca)).
    \]
    For a strategy profile $\vecpi  \in \Pi_H^N$, $\tau\geq 0$ the (entropy regularized) sum of rewards of player $i\in [N]$ is defined as
    \begin{equation*}
        \Jauc^{\tau, i}(\vecpi) := \mathbb{E} \left[\sum_{h=0}^{H-1}\Rauc_{N,h}(s_h^i, a_h^i,\widehat{L}_h) + \tau\entropy(\pi_h^{i}(s_h^{i}))\middle| \substack{\forall j \in [N]: s_0^j \sim \mu_0,\, a_h^j \sim \pi_h^i( s_h^j)\\
        \vecz_{h} \sim P_{N,\alpha_h}(\vecs_h, \veca_h),\, \widehat{L}_h = \frac{1}{N}\sum_{j\in [N]}\vece_{s_h^j, a_h^j}\\
        s_{h+1}^j \sim \transitionfn(z_h^j, \widehat{\xi}_h), \widehat{\xi}_h = \frac{1}{N}\sum_{j \in [N]}\vece_{z_h^j}
        } \right],  
    \end{equation*}
    exploitability $\Expauc^{\tau, i}$ as $\Expauc^{\tau, i}(\vecpi) := \max_{\pi' \in \Pi_H}\Jauc^{\tau, i}(\pi', \vecpi^{-i}) - \Jauc^{\tau, i}(\vecpi).$
\end{definition}

\subsection{Preliminary Lemmas}
\label{subsection:auxiliary_lemmas}
We present several important lemmas that will be used later to prove the main convergence theorem.

\begin{lemma}[Sensitivity of $ \Xi_\alpha $ to the Population]\label{lemma:lipschitz_Xi_pop}
Let $ \alpha \in [0,1] $ and let $ d, d' \in \Delta_{\mathcal{V} \times \mathcal{A}}^{\geq \alpha} $. Then $ \Xi_\alpha $ is Lipschitz-continuous with respect to the $ \ell_1 $-norm, with constant 1:
\begin{equation*}
    \| \Xi_\alpha(d) - \Xi_\alpha(d') \|_1 \leq \| d - d' \|_1.
\end{equation*}
\end{lemma}
\begin{proof}
It is straightforward to verify that $\Xi_\alpha$ is continuous.
We verify that it is also Lipschitz continuous with modulus 1.
Assume $\alpha > 0$, as otherwise $\Xi_\alpha(d)$ is the identity map and the claim is trivial.
Denote the bids by $\setA := \{ 1, \ldots, A \}$, and let $S_a(d) := \sum_{s'} \sum_{a' \geq a} d(s',a')$ with $S_{A+1} (d) := 0$, and $r_{a}(d) := \sum_{s'\in \setV}d(s',a)$.
For $\bar{a} \in \{ 1, \ldots, A+1\}$, define the regions
\begin{align*}
    \setR_{\bar{a}} := \{ d \in \Delta_{\mathcal{V} \times \mathcal{A}}^{\geq \alpha} : S_{\bar{a}}(d) \geq \alpha > S_{\bar{a}+1}(d) \}.
\end{align*}
On the region $\setR_{\bar{a}}$, the map $\Xi_\alpha(d)$ is differentiable, in fact, for $d\in \setR_{\bar{a}}$, it holds that $\bar{a} = \threshact_\alpha(d)$ and
\begin{align*}
    \Xi_\alpha(d)(s,a) = 
    \begin{cases}
        0 & \text{if } a > \bar{a}\\
          \frac{S_{\bar{a}}(d) -\alpha}{r_{\bar{a}}(d)}d(s, \bar{a}) & \text{if } a = \bar{a}\\
          d(s,a) & \text{otherwise}
    \end{cases}
\end{align*}
We calculate the Jacobian of $\Xi_\alpha$ and upper bound its operator norm $\|\nabla \Xi_\alpha \|_{1\rightarrow 1}$ given by the max column sum $\|\nabla \Xi_\alpha \|_{1\rightarrow 1} = \max_{s,a} \sum_{s',a'} | (\nabla \Xi_\alpha)_{s'a', sa}|.$
We upper bound the column sums corresponding to $s,a$.

\textbf{Case 1.}
If $a < \bar{a}$, then $\frac{\partial \Xi_\alpha(d)(s',a')}{\partial d_{sa}} = 0$ for any $(s',a') \neq (s,a)$ and $\frac{\partial \Xi_\alpha(d)(s,a)}{\partial d_{sa}} = 1$, therefore the column sum is
\begin{align*}
    \sum_{s',a'} \left| \frac{\partial \Xi_\alpha(d)(s',a')}{\partial d_{sa}} \right| = \left| \frac{\partial \Xi_\alpha(d)(s,a)}{\partial d_{sa}} \right| =  1.
\end{align*}

\textbf{Case 2.}
If $a = \bar{a}$, then $\frac{\partial \Xi_\alpha(d)(s',a')}{\partial d_{sa}} = 0$ if $a'\neq \bar{a}$, therefore
\begin{align*}
    \sum_{s',a'} \left| \frac{\partial \Xi_\alpha(d)(s',a')}{\partial d_{sa}} \right| &= \sum_{s'} \left| \frac{\partial \Xi_\alpha(d)(s',\bar{a})}{\partial d_{s\bar{a}}} \right| = \left| \frac{\partial \Xi_\alpha(d)(s,\bar{a})}{\partial d_{s\bar{a}}} \right| + \sum_{s'\neq s} \left| \frac{\partial \Xi_\alpha(d)(s',\bar{a})}{\partial d_{s\bar{a}}} \right| \\
    &= \left| \frac{S_{\bar{a}}(d) - \alpha}{r_{\bar{a}}(d)} + d(s,\bar{a}) \frac{r_{\bar{a}}(d) - S_{\bar{a}}(d) + \alpha }{ r_{\bar{a}}(d)^2 }  \right| \\
        & \quad +  \sum_{s'\neq s} \left| d(s',\bar{a}) \frac{r_{\bar{a}}(d) - S_{\bar{a}}(d) + \alpha}{ r_{\bar{a}}(d)^2 }  \right|
        = 1,
\end{align*}
since all terms in the absolute values are nonnegative if $\bar{a} = \threshact_\alpha(d)$.

\textbf{Case 3.}
If $a > \bar{a}$, then only the rows corresponding to the active action $\bar{a}$ has nonzero gradient, and
\begin{align*}
    \sum_{s',a'} \left| \frac{\partial \Xi_\alpha(d)(s',a')}{\partial d_{sa}} \right| &= \sum_{s'} \left| \frac{\partial \Xi_\alpha(d)(s',\bar{a})}{\partial d_{s\bar{a}}} \right| = \sum_{s'} \frac{1}{\sum_{s''} d(s'',\bar{a})} d(s,\bar{a}) = 1.
\end{align*}

To conclude, it holds that on any arbitrary region $\setR_{\bar{a}}$,
\begin{align*}
    \|\nabla \Xi_\alpha \|_{1\rightarrow 1} = \max_{s,a} \sum_{s',a'} \left| \frac{\partial \Xi_\alpha(d)(s',\bar{a})}{\partial d_{s\bar{a}}} \right| = 1.
\end{align*}
Therefore, $\|\nabla \Xi_\alpha \|_{1\rightarrow 1} \leq 1$ on all regions $\setR_{\bar{a}}$, and $\Xi_\alpha$ is non-expansive in the $\ell_1$ norm.
\end{proof}

\begin{lemma}[Sensitivity of $ \Xi_\alpha $ to the Allocation Parameter]\label{lemma:lipschitz_Xi_alpha}
    Let $d\in \Delta_{\setV \times \setA}^{\leq 1}$ arbitrary, let $\alpha_1, \alpha_2 \leq \|d\|_1$ arbitrary. Then
    \begin{equation*}
        \|\Xi_{\alpha_1}(d) - \Xi_{\alpha_2}(d)\|_1 \leq |\alpha_1 - \alpha_2|.
    \end{equation*}
\end{lemma}
\begin{proof}
    As in \Cref{lemma:lipschitz_Xi_pop}, let $\setA := \{ 1, \ldots, A \}$, and let $S_a(d) := \sum_{s'} \sum_{a' \geq a} d(s',a')$ with $S_{A+1} (d) := 0$, and $r_{a}(d) := \sum_{s'\in \setV}d(s',a)$.
    
For $\bar{a} \in \{ 1, \ldots, A+1\}$ such that $r_{\bar{a}}(d) > 0$, define the partition of the interval $[0, \|d\|_1]$ into the intervals
\begin{align*}
    \setR_{\bar{a}} := \{ \alpha \in [0, \|d\|_1] : S_{\bar{a}}(d) \geq \alpha > S_{\bar{a}+1}(d) \}.
\end{align*}
On the interval $\setR_{\bar{a}}$, the map $\Xi_\alpha(d)$ is differentiable in $\alpha$, and it holds that $\bar{a} = \threshact_\alpha(d)$ and
\begin{align*}
    \Xi_\alpha(d)(s,a) = 
    \begin{cases}
        0 & \text{if } a > \bar{a}\\
          \frac{S_{\bar{a}}(d) -\alpha}{r_{\bar{a}}(d)}d(s, \bar{a}) & \text{if } a = \bar{a}\\
          d(s,a) & \text{otherwise}
    \end{cases}
\end{align*}
Then, again, we upper bound the operator norm of the Jacobian (in this case, gradient):
\begin{align*}
    \| \nabla_\alpha \Xi_\alpha(d) \|_{1\rightarrow 1} = &\sum_{s', a'} \left| \frac{\partial \Xi_\alpha(d)(s',a') }{\partial \alpha} \right| = \sum_{s'} \left| \frac{-1}{ r_{\bar{a}}(d) } d(s', \bar{a})\right| = 1,
\end{align*}
So $\| \nabla_\alpha \Xi_\alpha(d) \|_{1\rightarrow 1} \leq 1$, implying the claim of the lemma.
\end{proof}

For completeness, we state the simple corollaries of the above two sensitivity analyses in the following. 

\begin{corollary}[Lipschitz Continuity of $\Xi$]\label{lemma:lipschitz_Xi}
    Let $d,d' \in \Delta_{\setV \times \setA}^{\leq 1}, \alpha \in [0, \|d\|_1),$ and $\alpha' \in [0, \|d'\|_1)$ arbitrary, then
    \begin{equation*}
        \|\Xi_{\alpha}(d) - \Xi_{\alpha'}(d')\|_1 \leq \|d-d'\|_1 + |\alpha - \alpha'|.
    \end{equation*}
\end{corollary}
\begin{proof}
    Let $d,d' \in \Delta_{\setV \times \setA}^{\leq 1}, \alpha \in [0, \|d\|_1),$ and $\alpha' \in [0, \|d'\|_1)$ arbitrary. Without loss of generality assume $\|d\|_1 \geq \|d'\|_1$, then by applying triangular inequality we have
    \begin{align*}
        \|\Xi_{\alpha}(d) - \Xi_{\alpha'}(d')\|_1 &\leq \|\Xi_{\alpha}(d) - \Xi_{\alpha'}(d)\|_1 + \|\Xi_{\alpha'}(d) - \Xi_{\alpha'}(d')\|_1
        \leq |\alpha - \alpha'| + \|d-d'\|_1,
    \end{align*}
    where the last step follows from \Cref{lemma:lipschitz_Xi_pop,lemma:lipschitz_Xi_alpha}.
\end{proof}

\begin{corollary}[Lipschitz Continuity of $ \Upsilon_\alpha$]\label{corollary:lipschitz_upsilon}
    Let $d,d' \in \Delta_{\setV \times \setA}^{\leq 1}, \alpha \in [0, \|d\|_1),$ and $\alpha' \in [0, \|d'\|_1)$ arbitrary, then $\Upsilon_\alpha$ is non-expansive in the $\ell_1$ norm, that is 
    \begin{equation*}
        \|\Upsilon_\alpha(d) - \Upsilon_{\alpha'}(d')\|_1 \leq |\alpha - \alpha'| +  \|d-d'\|_1.
    \end{equation*}
\end{corollary}
\begin{proof}
Let $d,d' \in \Delta^{\leq 1}, \alpha \in [0, \|d\|_1),$ and $\alpha' \in [0, \|d'\|_1)$ arbitrary, then
    \begin{align*}
        \|\Upsilon_\alpha(d) - \Upsilon_{\alpha'}(d')\|_1 &= \left\|\sum_{a \in \setA}\Xi_\alpha(d)(\cdot, a) - \sum_{a \in \setA}\Xi_{\alpha'}(d')(\cdot, a) \right\|_1\\
        &\leq \sum_{a\in \setA}\|\Xi_\alpha(d)(\cdot, a) -  \Xi_{\alpha'}(d')(\cdot, a)\|_1 = \|\Xi_\alpha(d) - \Xi_{\alpha'}(d')\|_1.
    \end{align*}
    The upper bound follows by the result of \Cref{lemma:lipschitz_Xi}.
\end{proof}

\begin{corollary}[Lipschitz Continuity of $ \Gamma_\alpha $]\label{corollary:lipschitz_Gamma_alpha}
Let $ d, d' \in \Delta_{\mathcal{S} \times \mathcal{A}} $, and suppose the threshold function $ \alpha: \Delta_{\setA}^{\leq 1} \rightarrow [0,1] $ is Lipschitz continuous with constant $ K_\alpha $, and satisfies the feasibility condition $ \alpha(\nu) \leq \|\nu\|_1 $ for all $ \nu \in \Delta_{\setA}^{\leq 1} $. Then,
\begin{equation*}
    \|\Gamma_\alpha(d) - \Gamma_\alpha(d') \|_1 \leq (K_\alpha + 1)\, \| d - d' \|_1.
\end{equation*}
\end{corollary}

\begin{proof}
    By \Cref{corollary:lipschitz_upsilon} and the Lipschitz continuity of $\alpha$, it follows that:
    \begin{align*}
         \|\Gamma_\alpha(d) - \Gamma_\alpha(d') \|_1 &\leq |\alpha(\nu^{-\perp}(d)) - \alpha(\nu^{-\perp}(d')) | + \|d-d'\|_1\\
         &\leq K_\alpha \|d-d'\|_1 + \|d-d'\|_1.
    \end{align*}
\end{proof}

In the next sequence of results, we deal with the stability of winning probabilities given by the function $\pwin$.
In general, $\pwin$ is easily seen to have discontinuous jumps, however, a local stability result can be shown if the NZD condition holds.

\begin{lemma}[Stability of Winning Probabilities]\label{lemma:deviation_winning_prob}
    Let $ \mathcal{A} = \{a_1, a_2, \dots, a_K\} $ be a finite set of actions with total order $ a_1 < a_2 < \dots < a_K $, and let $ \nu \in \Delta_{\mathcal{A}}^{\leq 1} $.
    For any $ \alpha \in [0,1] $, define the winning probability for action $ a \in \mathcal{A} $ as
    \begin{align*}
        \pwin(a, \nu, \alpha) := 
        \begin{cases}
        0 & \text{if } \sum_{a' > a} \nu(a') \geq \alpha \text{ or } \alpha = 0, \\
        1 & \text{if } \sum_{a' \geq a} \nu(a') \leq \alpha \text{ and } \alpha > 0, \\
        \frac{\alpha - \sum_{a' > a} \nu(a')}{\nu(a)} & \text{otherwise}.
        \end{cases}
    \end{align*}
    Assume that $ \alpha, \nu $ satisfy the no zero-dominance property (i.e., for any $ a \in \mathcal{A} $, $ \nu(a) = 0 \Rightarrow \sum_{a' > a} \nu(a') < \alpha(\nu) $) and $\|\nu\|_1 > 0$. 
    Then for all $ \nu' \in \Delta_{\mathcal{A}}^{\leq 1}, \alpha' \in [0,1], a \in \mathcal{A} $,  $\pwin $ satisfies
    \begin{align*}
        \left| \pwin(a, \nu, \alpha) - \pwin(a, \nu', \alpha') \right| \leq C_{\nu, \alpha} \| \nu - \nu' \|_1 + C_{\nu, \alpha} |\alpha - \alpha'|,
    \end{align*}
    where:
    \begin{align*}
        C_{\nu, \alpha} = 
        \begin{cases}
            \frac{1}{\min\{ \nu(a^*) \}} & , \text{if } \alpha = 0 \\
            \frac{1}{\min\{ \nu(a^*), \nu(a^-), \Delta_{\nu, \alpha} \}} & , \text{if } \Delta_{\nu, \alpha} > 0, a^- \neq \perp, \alpha > 0 \\
            \frac{1}{\min\{\nu(a^*), \Delta_{\nu, \alpha} \}} & , \text{if } \Delta_{\nu, \alpha} > 0, a^- = \perp, \alpha > 0 \\
            \frac{1}{\min\{ \nu(a^*), \nu(a^-)\}} & , \text{if } \Delta_{\nu, \alpha} = 0, a^- \neq \perp, \alpha > 0 \\
            \frac{1}{\min\{ \nu(a^*) \}} & , \text{if } \Delta_{\nu, \alpha} = 0, a^- = \perp, \alpha > 0 \\
        \end{cases}
    \end{align*}
    and where $a^* := \threshact_{\alpha}(\nu)$ is the threshold action, 
    $a^-$ is the action just below the threshold ($a^- := a_{k-1}$ if $a^* = a_k$ for some $k> 1$ and $a^- = \perp$ if $a^* = a_1$),
    and $\Delta_{\nu, \alpha} := \min_{a \in \mathcal{A}} \left| \alpha - \sum_{a' \succeq a} \nu(a') \right|$.
\end{lemma}

    \begin{proof}
    Importantly, $C_{\nu, \alpha}$ is finite in each case if $\nu, L$ satisfies the NZD condition.
    The proof, while notionally dense, works on a simple idea: in general, $\pwin$ incorporates discontinuities where the winning probability of an action below the threshold might jump from $0$ to $1$.
    However, this does not happen \emph{locally} when there is some probability mass on $a^-$ just below the threshold action.
    Note that when $\|\nu\|_1 \neq 0$ and NZD holds, it holds that $\nu(a^*) > 0,$ and $\nu(a^-) > 0$ if $a^* \neq a_1$ and $\alpha > 0$.
    Furthermore, by NZD, if $\alpha = 0$, it must hold that $\nu(a_K) > 0$ and $a^* = a_K$ by definition.
    Define the useful constant
    \begin{align*}
        \delta := \begin{cases}
            \Delta_{\nu, \alpha}, &\text{ if } \Delta_{\nu, \alpha} > 0, \\
            \max \{ \sfrac{1}{\nu(a^*)}, \sfrac{1}{\nu(a^-)}\} , &\text{ if } \Delta_{\nu, \alpha} = 0, a^- \neq \perp \\
            \sfrac{1}{\nu(a^*)} , &\text{ if } \Delta_{\nu, \alpha} = 0, a^- = \perp \\
        \end{cases}
    \end{align*}
    which will be the radius of the open set around which there are no discontinuities of $\pwin$.
    
    First, we show that $\left| \pwin(a, \nu, \alpha) - \pwin(a, \nu, \alpha') \right| \leq C_{\nu, \alpha} |\alpha - \alpha'|$ for any $\alpha'$.
    Without loss of generality, we can assume that $\alpha \leq \| \nu\|_1, \alpha' \leq \|\nu\|_1$, as $\pwin(a, \nu, \alpha') = \pwin(a, \nu, \min\{\alpha', \|\nu\|_1\})$ for any $\alpha'$ by definition, and $|\min\{\alpha, \|\nu\|_1\}  - \min\{\alpha', \|\nu\|_1\}| \leq |\alpha - \alpha'|$.
    If $|\alpha - \alpha'| < \delta $, then $\alpha' \in (\alpha - \delta , 0, \alpha+\delta ) \cap \mathbb{R}_{\geq 0}$.
    On the interval $(\alpha - \delta, \alpha+\delta )\cap \mathbb{R}_{\geq 0}$, $\pwin$ is continuous for any $a$ since
    \begin{align*}
        \text{For } a < a^*, a \neq a^-: \quad &\pwin(a, \nu, \alpha') = 0, \\
        \text{For } a > a^*: \quad &\pwin(a, \nu, \alpha') = 1, \\
         \text{For } a = a^*: \quad &\pwin(a^*, \nu, \alpha') = \begin{cases}
        1 & \text{if } \sum_{a' \geq a^*} \nu(a') \leq \alpha', \\
        \frac{\alpha' - \sum_{a' > a^*} \nu(a')}{\nu(a^*)} & \text{otherwise}.
        \end{cases} \\
        \text{For } a = a^-, \text{ if } a^- \neq \perp: \quad &\pwin(a^-, \nu, \alpha') = \begin{cases}
        0 & \text{if } \sum_{a' > a^-} \nu(a') \geq \alpha', \\
        \frac{\alpha' - \sum_{a' > a^-} \nu(a')}{\nu(a^-)} & \text{otherwise}.
        \end{cases}
    \end{align*}
    Moreover, from above, $\pwin$ is almost everywhere differentiable in $\alpha'$ in this interval with
    \begin{align*}
        \left|\frac{\partial \pwin(a, \nu, \alpha)}{\partial \alpha}\right| \leq  
        \begin{cases}
            \max\{ \sfrac{1}{\nu(a^-)} ,  \sfrac{1}{\nu(a^*)}\}, \text{ if } a^*\neq a_1, \alpha > 0 \\
            \sfrac{1}{\nu(a^*)}, \text{ otherwise }
         \end{cases}
    \end{align*}
    as in this interval the ``active'' threshold action will always be either $a^*$ of $a^-$.
    Therefore, 
    \begin{align*}
        \left| \pwin(a, \nu, \alpha) - \pwin(a, \nu, \alpha') \right| \leq C_{\nu, \alpha} |\alpha - \alpha'|.
    \end{align*}
    If on the other hand $|\alpha - \alpha'| \geq \delta$, then
    \begin{align*}
        \left|\pwin(a, \nu, \alpha) - \pwin(a, \nu, \alpha') \right| \leq 1 \leq \delta ^ {-1} |\alpha - \alpha'| \leq C_{\nu, \alpha} |\alpha - \alpha'|.
    \end{align*}
    So in both cases it holds that $ \left|\pwin(a, \nu, \alpha) - \pwin(a, \nu, \alpha') \right| \leq C_{\nu, \alpha} |\alpha - \alpha'|.$

    Next, we show the stability in $\nu$.
    Once again, assume that $\| \nu' - \nu \|_1 < \delta$.
    On the open set $ \{\nu' : \| \nu' - \nu \|_1 < \delta \}$ the function $\pwin(a, \cdot, \alpha)$ is once again continuous in $\nu'$ for any $a$, and almost everywhere differentiable with
    \begin{align*}
        \left| \frac{\partial \pwin(a, \nu', \alpha)}{\partial \nu'(a')} \right| \leq \begin{cases}
            \frac{1}{\nu(a^*)}, &\text{ if } a^- = \perp \text{ or } \alpha = 0 \\
            \frac{1}{\min\{\nu(a^*), \nu(a^-)\}}, &\text{ if } a^- \neq \perp, \alpha > 0 \\
        \end{cases}
    \end{align*}
    for almost every $\nu'$.
    Therefore, it holds that
    \begin{align*}
        \left| \pwin(a, \nu', \alpha) - \pwin(a, \nu, \alpha) \right| \leq C_{\nu, \alpha} \|\nu - \nu'\|_1.
    \end{align*}
    On the other hand, if  $\| \nu' - \nu \|_1 \geq \delta$, then 
    \begin{align*}
        \left| \pwin(a, \nu', \alpha) - \pwin(a, \nu, \alpha) \right| \leq 1 \leq \delta^{-1} \| \nu - \nu'\|_1.
    \end{align*}

    To complete the proof, we use the triangle inequality:
    \begin{align*}
        \left| \pwin(a, \nu', \alpha') - \pwin(a, \nu, \alpha) \right| \leq \left| \pwin(a, \nu, \alpha') - \pwin(a, \nu, \alpha) \right|  + \left| \pwin(a, \nu, \alpha') - \pwin(a, \nu', \alpha') \right| 
    \end{align*}
\end{proof}

\begin{remark}\label{remark:Lipschitz_pwin_full_support}
    The proof of \Cref{lemma:deviation_winning_prob} can be adapted to handle action distributions with full support. In this case, a refined version of the first part of the argument shows that the winning probability is Lipschitz-continuous with a constant bounded by $\frac{1}{\min\{\nu(a^-), \nu(a^*), \nu(a^+)\}}$, where $a^*, a^-, a^+$ are actions around the threshold. Consequently, for policies with full support, i.e., $\pi_h(a \vert s) \geq \epsilon$ for all $s \in \setV$, $a \in \setA$, and some $\epsilon > 0$, the deviation in winning probability is bounded by $\frac{1}{(1 - \alpha_{\max})\epsilon}\|\nu - \nu'\|_1$.
\end{remark}

\begin{lemma}\label{lemma:L_mu_relation}
    Let $ \mu, \mu' \in \Delta_{\mathcal{S}} $ be two arbitrary state distributions, and let $ \pi, \pi' \in \Pi $ be two arbitrary policies. Define the corresponding state-action distributions $ L, L' \in \Delta_{\mathcal{S} \times \mathcal{A}} $ as
    \[
    L(s,a) := \mu(s)\pi(a \vert s), \qquad L'(s,a) := \mu'(s)\pi'(a \vert s).
    \]
    Then, it holds that
    \[
    \|L - L'\|_1 \leq \sup_s \|\pi(\cdot|s) - \pi'(\cdot|s)\|_1 + \|\mu - \mu'\|_1.
    \]    \end{lemma}
    
    \begin{proof}
    We compute the total variation distance:
    \begin{align*}
    \|L - L'\|_1 &= \sum_{s,a} \left| \mu(s)\pi(a \vert s) - \mu'(s)\pi'(a \vert s) \right| \\
    &= \sum_{s,a} \left| \mu(s)\pi(a \vert s) - \mu(s)\pi'(a \vert s) + \mu(s)\pi'(a \vert s) - \mu'(s)\pi'(a \vert s) \right| \\
    &\leq \sum_s \mu(s) \sum_a \left| \pi(a \vert s) - \pi'(a \vert s) \right| + \sum_s |\mu(s) - \mu'(s)| \sum_a \pi'(a \vert s) \\
    &= \sup_s \|\pi(\cdot|s) - \pi'(\cdot|s)\|_1 + \|\mu - \mu'\|_1,
    \end{align*}
    where the last line uses that $ \sum_a \pi'(a \vert s) = 1 $ for all $ s $.
\end{proof}

\subsection{Proof of \texorpdfstring{\Cref{theorem:approximation_auction}}{}, part 1 (Approximation in Exploitability)}
\label{section:upper_bound_exploitability_proof}

The theorem considers BA-MFGs with Lipschitz-continuous reward, dynamics, and utility functions. 
Let $ K_{\transitionfn} \in [0,2] $ denote the Lipschitz modulus of the state dynamics $ \transitionfn_h $, 
$ K_\paymentfn $ the Lipschitz modulus of the payment function $ \paymentfn_h $, 
$ K_\utilityfn $ the Lipschitz modulus of the utility function $ \utilityfn_h $ with respect to its second argument, 
and $ K_\alpha $ be the Lipschitz modulus of the allocation funciton $\alpha_h$.
That is, for any $s\in\setV, a\in\setA, \nu,\nu'\in \Delta_\setA, p,p' \in \mathbb{R}_{\geq0}$,
\begin{align*}
     \left| \utilityfn_h(s, p) - \utilityfn_h(s, p') \right| \leq K_\utilityfn \cdot \left| p - p' \right|, 
     &\qquad
     \left| \paymentfn_h(a, \nu) - \paymentfn_h(a, \nu') \right| \leq K_\paymentfn \cdot \| \nu - \nu' \|_1, \\
     \left| \alpha_h(\nu) - \alpha_h(\nu') \right| &\leq K_\alpha \cdot \| \nu - \nu' \|_1.
\end{align*}
The composed function $ \utilityfn_h(s, \paymentfn_h(a, \nu)) $ then is 
also Lipschitz with modulus $ K_\utilityfn K_\paymentfn $ as
\begin{equation*}
    \left| \utilityfn_h(s, \paymentfn_h(a, \nu)) - \utilityfn_h(s, \paymentfn_h(a, \nu')) \right| \leq K_\utilityfn \cdot \left| \paymentfn_h(a, \nu) - \paymentfn_h(a, \nu') \right| \leq K_\utilityfn K_\paymentfn \cdot \| \nu - \nu' \|_1,
\end{equation*}
for all $ \nu, \nu' \in \Delta_{\mathcal{A}}^{\leq 1} $. 

Additionally, let $ B_\paymentfn $ be an upper bound on the absolute value of payments, i.e., $ |\paymentfn_h(a, \nu)| \leq B_\paymentfn $, and let $ B_\utilityfn $ be an upper bound on the absolute value of utilities, i.e., $ |\utilityfn_h(s, p)| \leq B_\utilityfn $, for all $ s, a, \nu, h $.
Such an upper bound always exists if $\utilityfn_h, \paymentfn_h$ are continuous in $\nu$, since $\Delta_{\setA}^{\leq 1}$ is a compact set.

The argument proceeds in three steps:
\begin{enumerate}
    \item First, we bound the expected deviation between the empirical distributions and their mean field counterparts. 
    That is, we show an upper bound on the deviation $\mathbb{E}\left[\|L_h^{\overline{\pi}} - \widehat{L}_h\|_1\right]$.
    \item Second, we show that for policies satisfying the no zero-dominance property, the expected transition probabilities associated with item allocation, under both the mean field and finite-population settings, differ proportionally to the deviation in population distributions.
    \item  Finally, we bound the exploitability of a single agent when all other agents follow a policy that satisfies the no zero-dominance property with respect to the mechanism.
\end{enumerate}

In our analysis, we also make use of the population distribution $ \xi_h $ after item allocation at round $ h $. In the mean field setting, this is defined as $\xi_h := \Gamma_{\alpha_h}(L_h)$,
where $ L_h \in \Delta_{\mathcal{S} \times \mathcal{A}} $ is the state-action distribution at round $ h $, and the operator $ \Gamma_{\alpha_h} $ captures the expected post-allocation state distribution under the mechanism (e.g., by marking winners as inactive). 
This quantity serves as the input to the state transition function $ \transitionfn_h $ in the MFG dynamics. In the finite-agent setting, we denote the analogous empirical quantity by $ \widehat{\xi}_h $, representing the empirical distribution over states immediately after allocation.
We also define the random variables $\{ z_{h}^{i} \}$ as $z_{h}^{i} = \perp$ if agent $i$ was not active in round $h$ (i.e. $s_h^i = \perp$) or agent $i$ won the the auction in round $h$, and $z_{h}^{i} = s_h^i$ otherwise.
With this definition,
\begin{align*}
    \widehat{\xi}_h = \frac{1}{N}\sum_{i=1}^N \vece_{z_{h}^{i}} \in \Delta_\setS.
\end{align*}

Finally, we define the constants used in our convergence analysis as follows. Define 
\begin{align*}
    K_s &:= \sup_{s, s', \xi} \| \transitionfn(s, \xi) - \transitionfn(s', \xi) \|_1, \quad K_\xi := K_\transitionfn + \frac{1}{2}K_s.
\end{align*}

\subsubsection{Step 1: Expected Deviation of Empirical Distributions}

We derive explicit bounds on the expected deviation between the empirical distributions and their mean field counterparts. In particular, we bound the deviations for the state distribution and the state-action distribution. 

\begin{lemma}\label{lemma:state_action_pop_to_mu}
    Let $\Mauc = (\setS, \setA, H, \mu_0, \{\Pauc_h\}_{h=0}^{H-1}, \{\Rauc_h\}_{h=0}^{H-1})$ define a BA-MFG. 
    Consider the corresponding finite-agent Batched Auction model with $N$ agents, $\Gauc$, which is approximated by $\Mauc$. Let $\vecpi = \{ \pi_h^i \}_{h=0,\dots,H-1,\, i \in [N]} \in \Pi_H^N$ denote the joint policy of the population. Denote by $\widehat{\mu}_h \in \Delta_{\setS}$ the empirical state distribution and by $\widehat{L}_h \in \Delta_{\setS \times \setA}$ the empirical state-action distribution at round $h$.

    Let $\overline{\pi} \in \Pi_H$ arbitrary, and define the associated mean field state-action distribution flow $L^{\overline{\pi}} := \lpopmfa(\overline{\pi})$, with corresponding marginal state distribution $\mu_h^{\overline{\pi}} := \sum_a L_h^{\overline{\pi}}(\cdot, a)$. Then, for all $h \in \{0, \dots, H-1\}$, the following bound holds:
    \begin{equation*}
        \mathbb{E}\left[\|L_h^{\overline{\pi}} - \widehat{L}_h\|_1\right]
        \leq \mathbb{E}\left[\|\mu_h^{\overline{\pi}} - \widehat{\mu}_h\|_1\right]
        + \frac{1}{N} \sum_{i \in [N]} \| \overline{\pi}_{h} - \pi_{h}^i \|_1 +  \sqrt{\frac{\vert \setA \vert  \vert \setS\vert}{N}}.
    \end{equation*}
\end{lemma}
\begin{proof}
    We decompose the deviation:
    \begin{align*}
        \mathbb{E}\left[\|L_h^{\overline{\pi}} - \widehat{L}_h\|_1\middle\vert \{s_h^i\}_{i,h}\right] &\leq \underbrace{\|L_h^{\overline{\pi}} - \mathbb{E}[\widehat{L}_h\vert \{s_h^i\}_{i,h}]\|_1}_{(\square)} + \underbrace{\mathbb{E}\left[\left\|\widehat{L}_h - \mathbb{E}[\widehat{L}_h\vert \{s_h^i\}_{i,h}]\right\|_1\middle\vert \{s_h^i\}_{i,h}\right]}_{(\triangle)}.
    \end{align*}
    We bound the two terms separately. 
    For $(\square)$, define $L \in \Delta_{\setS \times \setA}$ as $L(s,a) = \widehat{\mu}_h(s)\overline{\pi}(a\vert s) = \frac{1}{N}\sum_{i \in N}\ind{s_h^i = s}\overline{\pi}(a\vert s)$.
    We (almost surely) have:
    \begin{align*}
        \|L_h^{\overline{\pi}} - \mathbb{E}[\widehat{L}_h\vert  \{s_h^i\}_{i,h}]\|_1 &\leq \|L_h^{\overline{\pi}} - L\|_1 + \|\mathbb{E}[\widehat{L}_h\vert \widehat{\mu}_h] - L\|_1\\
        &\leq \sum_{s\in \setS}\sum_{a\in \setA}\vert \mu_h^{\overline{\pi}}(s)\overline{\pi}(a\vert s) - \widehat{\mu}_h(s)\overline{\pi}(a\vert s)\vert + \|\mathbb{E}[\widehat{L}_h\vert  \{s_h^i\}_{i,h}] - L\|_1\\
        &\leq \|\mu_h^{\overline{\pi}} - \widehat{\mu}_h\|_1 + \frac{1}{N}\sum_{i\in [N]}\sum_{a\in \setA}\vert\overline{\pi}(a\vert s_h^i) - \pi_h^i(a\vert s_h^i)\vert\\
        &\leq \|\mu_h^{\overline{\pi}} - \widehat{\mu}_h\|_1 + \frac{1}{N}\sum_{i\in [N]}\|\overline{\pi}_h - \pi_h^i\|_1.
    \end{align*}
    For the second term $(\triangle)$, by applying Jensen's inequality, we have:
    \begin{align*}
        \mathbb{E}\left[\left\|\widehat{L}_h - \mathbb{E}[\widehat{L}_h\vert  \{s_h^i\}_{i,h}]\right\|_1\middle\vert  \{s_h^i\}_{i,h}\right] &= \sum_{s \in \setS}\sum_{a\in \setA}\mathbb{E}\left[\left\vert\widehat{L}_h(s,a) - \mathbb{E}[\widehat{L}_h(s,a)\vert  \{s_h^i\}_{i,h}]\right\vert\middle\vert  \{s_h^i\}_{i,h}\right]\\
        &\leq  \sum_{s \in \setS}\sum_{a\in \setA}\sqrt{\Var{\widehat{L}_h(s,a) \middle\vert  \{s_h^i\}_{i,h}}}\\
        &=\frac{1}{N}\sum_{s \in \setS}\sum_{a\in \setA}\sqrt{\sum_{i\in [N], s_h^i = s}\pi_h^i(a\vert s)(1-\pi_h^i(a\vert s))}.
    \end{align*}
    Applying Cauchy-Schwarz's inequality, we get for any $s\in\setS$:
    \begin{align*}
        \sum_{a\in \setA}\sqrt{\sum_{i\in [N], s_h^i = s}\pi_h^i(a\vert s)(1-\pi_h^i(a\vert s))} &\leq \sqrt{\vert \setA\vert \sum_{a\in \setA}\sum_{i\in [N], s_h^i = s}\pi_h^i(a\vert s)(1-\pi_h^i(a\vert s))}\\
        &\leq \sqrt{\vert \setA\vert \sum_{a\in \setA}\sum_{i\in [N], s_h^i = s}\pi_h^i(a\vert s)}\\
        &\leq \sqrt{N\vert \setA\vert \widehat{\mu}_h(s)}.
    \end{align*}
    By integrating this result into the previous computation and using Cauchy-Schwarz's inequality, we get:
    \begin{align*}
        \mathbb{E}\left[\left\|\widehat{L}_h - \mathbb{E}[\widehat{L}_h\vert  \{s_h^i\}_{i,h}]\right\|_1\middle\vert \{s_h^i\}_{i,h}\right] &= \frac{1}{N}\sum_{s \in \setS}\sum_{a\in \setA}\sqrt{\sum_{i\in [N], s_h^i = s}\pi_h^i(a\vert s)(1-\pi_h^i(a\vert s))}\\
        &\leq \frac{\sqrt{\vert \setA \vert}}{\sqrt{N}}\sum_{s\in \setS}\sqrt{\widehat{\mu}_h(s)}
        \leq \sqrt{ \frac{\vert \setA \vert \vert \setS\vert}{N} \sum_{s \in \setS}\widehat{\mu}_h(s) } \leq \sqrt{\frac{\vert \setA \vert \vert \setS\vert}{N}}.
    \end{align*}
    Combining the upper bounds derived for terms $(\square)$ and $(\triangle)$, we obtain the desired result, as $\mathbb{E}\left[\|L_h^{\overline{\pi}} - \widehat{L}_h\|_1\right] = \mathbb{E}\left[\mathbb{E}\left[\|L_h^{\overline{\pi}} - \widehat{L}_h\|_1\middle\vert  \{s_h^i\}_{i,h}\right]\right]$.
\end{proof}

\begin{lemma}[Deviation Between Empirical and Mean Field Population]\label{lemma:state_pop_deviation}
    Let $\Mauc = (\setS, \setA, H, \mu_0, \{\Pauc_h\}_{h=0}^{H-1}, \{\Rauc_h\}_{h=0}^{H-1})$ be a BA-MFG. 
    Let $\{\alpha_h\}_{h=0}^{H-1}$, $\{\transitionfn_h\}_{h=0}^{H-1}$, $\{\paymentfn_h\}_{h=0}^{H-1}$, and $\{\utilityfn_h\}_{h=0}^{H-1}$ denote the allocation thresholds, transition dynamics, payment, and utility functions, respectively, from which $\{\Pauc_h\}_{h=0}^{H-1}$ and $\{\Rauc_h\}_{h=0}^{H-1}$ are derived. Assume these functions are Lipschitz continuous, with respective Lipschitz constants $K_{\alpha}$, $K_{\transitionfn}$, $K_{\paymentfn}$, and $K_{\utilityfn}$.

    Consider the corresponding finite-agent Batched Auction model with $ N $ agents $\Gauc $, which is approximated by $ \Mauc $. Let $ \vecpi = \{ \pi_h^i \}_{h=0,\dots,H-1,\, i \in [N]} \in \Pi_H^N $ denote the joint policy of the population. For each round $ h $, denote by $ \widehat{L}_h \in \Delta_{\mathcal{S} \times \mathcal{A}} $ the empirical state-action distribution and by $ \widehat{\mu}_h \in \Delta_{\mathcal{S}} $ the corresponding empirical state distribution.
    
    Let $\overline{\pi} \in \Pi_H$ be an arbitrary policy, and define the associated mean field state-action distribution flow $L^{\overline{\pi}} := \lpopmfa(\overline{\pi})$, with corresponding marginal state distribution $\mu_h^{\overline{\pi}} := \sum_a L_h^{\overline{\pi}}(\cdot, a)$. Then, for all $h \in \{0, \dots, H-1\}$, it holds that
    \begin{align*}
        \mathbb{E}\left[ \| \mu_h^{\overline{\pi}} - \widehat{\mu}_h \|_1 \right]
        &\leq \frac{1 - (K_\xi(1 + K_\alpha))^{h+1}}{1 - K_\xi(1 + K_\alpha)}\frac{\sqrt{\vert \setS \vert}}{\sqrt{N}}\\
        & \quad + K_\xi\frac{1 - (K_\xi(1 + K_\alpha))^h}{1- K_\xi(1 + K_\alpha)}\biggl(\frac{\sqrt{\vert \setS \vert }}{2\sqrt{N}} + \frac{1}{N} +  (1 + K_\alpha) \frac{\sqrt{\vert \setS \vert \vert \setA \vert}}{\sqrt{N}} \biggr)\\
        &\quad + \sum_{h' < h} (K_\xi\cdot (K_\alpha + 1 ))^{h - h'} \cdot \frac{1}{N} \sum_{i \in [N]} \| \overline{\pi}_{h'} - \pi_{h'}^i \|_1,
        \end{align*}
        if $ K_\xi(1 + K_\alpha) \neq 1 $, and
    \begin{align*}
        \mathbb{E}\left[ \| \mu_h^{\overline{\pi}} - \widehat{\mu}_h \|_1 \right]
        &\leq (h+1) \frac{\sqrt{\vert \setS \vert}}{\sqrt{N}} + h\frac{\sqrt{\vert \setS \vert \vert \setA \vert}}{\sqrt{N}} + h K_\xi \biggl(\frac{\sqrt{\vert \setS \vert}}{2 \sqrt{N}} + \frac{1}{N}\biggr)\\
        &\quad + \sum_{h' < h} \frac{1}{N} \sum_{i \in [N]} \| \overline{\pi}_{h'} - \pi_{h'}^i \|_1,
    \end{align*}
    if $ K_\xi(1 + K_\alpha) = 1 $.
\end{lemma}
\begin{proof}
    Let $\{s_h^{i}\}_{h = 0, \dots, H-1, i\in [N]}, \{a_h^{i}\}_{h = 0, \dots, H-1, i\in [N]}$ and $\{z_h^{i}\}_{h = 0, \dots, H-1, i\in [N]}$ as in \Cref{definition:n_batched_auction}.
    We prove the lemma inductively over $h$. 
    For $h= 0$ we have $\mu_0^{\overline{\pi}} = \mu_0 = \mathbb{E}[\widehat{\mu}_0]$. 
    Let $X_s := \sum_{i\in [N]}\mathbbm{1}_{\{s_0^i = s\}}$, which is by definition $X_s$ a binomial random variable with parameters $N$ and $\mu_0(s)$. 
    Since it is a sum of independent Bernoulli random variables, its variance is $\Var{X_s} = N\mu_0(s)(1-\mu_0(s))$.
    By using Jensen's, we can upper bound the expected absolute deviation for each state $s\in \setV$
    \begin{align*}
        \mathbb{E}[\vert \mu_0(s) - \widehat{\mu}_0(s)\vert ] 
        \leq \sqrt{\Var{\widehat{\mu}_0(s)}}
        = \sqrt{\frac{\mu_0(s)(1-\mu_0(s))}{N}}\leq \sqrt{\frac{\mu_0(s)}{N}}
    \end{align*}
    By summing over all states $s\in \setV$ and applying Cauchy-Schwarz's inequality, we get:
    \begin{align*}
        \mathbb{E}[\|\mu_0 - \widehat{\mu}_0\|_1] = \sum_{s \in \setV}\mathbb{E}[\vert \mu_0(s) - \widehat{\mu}_0(s)\vert ] \leq \sum_{s \in \setV}\sqrt{\frac{\mu_0(s)}{N}} \stackrel{C.S.}{\leq} \frac{\sqrt{\vert \setV\vert }}{\sqrt{N}}.
    \end{align*}
    Next, for $h \geq 0$, we compute an upper bound for the deviation at step $h+1$.
    In particular, we analyze the conditional expectation
    \begin{align*}
        \mathbb{E} &\left[ \| \mu_{h+1}^{\overline{\pi}} - \widehat{\mu}_{h+1} \|_1 \big\vert \{z_h^i\}_{i=1}^N \right] \\
        & \quad \leq \underbrace{\mathbb{E}\left[\|\mu_{h+1}^{\overline{\pi}} - \mathbb{E}[\widehat{\mu}_{h+1}\vert \{z_h^i\}_{i=1}^N]\|_1 \middle\vert \{z_h^i\}_{i=1}^N \right]}_{:= (\triangle)} + \underbrace{ \mathbb{E}\left[\|\widehat{\mu}_{h+1} - \mathbb{E}[\widehat{\mu}_{h+1}\vert \{z_h^i\}_{i=1}^N]\|_1 \middle\vert \{z_h^i\}_{i=1}^N\right]}_{:= (\square)},
    \end{align*}
    almost surely,
    where $\{z_h^i\}_{i=1}^N$ are the states of agents after the item allocation in round $h$ as before.
    We upper bound the two terms $(\square)$ and $(\triangle)$ separately. 
    For $(\square)$ we have:
    \begin{align*}
        (\square) &= \sum_{s\in \setS}\mathbb{E}\left[|\widehat{\mu}_{h+1}(s) - \mathbb{E}[\widehat{\mu}_{h+1}(s) \, \vert  \{z_h^i\}] |_1 \middle\vert  \{z_h^i\}\right]  \leq \sum_{s\in \setS} \sqrt{\Var{\widehat{\mu}_{h+1}(s)\middle\vert  \{z_h^i\}}}.
    \end{align*}
    For arbitrary $s\in \setS$, noting that $\widehat{\xi}_h$ is $ \{z_h^i\}$-measurable,
    \begin{align*}
        \Var{\widehat{\mu}_{h+1}(s)\middle\vert \, \{z_h^i\} } &= \frac{1}{N^2}\sum_{s'\in \setS}N\widehat{\xi}_h(s')\transitionfn(s\vert s', \widehat{\xi}_h)(1- \transitionfn(s\vert s', \widehat{\xi}_h))
        \leq \frac{1}{N}\sum_{s'\in \setS }\widehat{\xi}_h(s')\transitionfn(s\vert s', \widehat{\xi}_h).
    \end{align*}
    By the Cauchy-Schwarz inequality, we have
    \begin{align*}
            \left(\sum_{s\in \setS}\sqrt{\sum_{s'\in \setS }\widehat{\xi}_h(s')\transitionfn(s\vert s', \widehat{\xi}_h)}\right)^2 
            \leq \left(\sum_{s\in \setS}\sum_{s'\in \setS }\widehat{\xi}_h(s')\transitionfn(s\vert s', \widehat{\xi}_h))\right) |\setS|
            =  \vert \setS \vert.
    \end{align*}
    Therefore, we (almost surely) have
    \begin{align*}
        (\square) &\leq \frac{1}{\sqrt{N}}\sum_{s\in \setS}\sqrt{\sum_{s'\in \setS }\widehat{\xi}_h(s')\transitionfn(s\vert s', \widehat{\xi}_h)}\leq \frac{\sqrt{\vert \setS \vert}}{\sqrt{N}}.
    \end{align*}
    For $(\triangle)$ it holds that
    \begin{align*}
        (\triangle) &= \mathbb{E}\left[\|\mu_{h+1}^{\overline{\pi}} - \mathbb{E}[\widehat{\mu}_{h+1}\vert \{z_h^i\}]\|_1 \middle\vert \{z_h^i\}\right]
        =\|\mu_{h+1}^{\overline{\pi}} - \mathbb{E}[\widehat{\mu}_{h+1}\vert \{z_h^i\}]\|_1\\
        &=\left\|\Gamma_g(\xi_h^{\overline{\pi}}) - \Gamma_g(\widehat{\xi}_h)\right\|_1 \leq \left(K_\transitionfn + \frac{K_s}{2}\right)\| \xi_h^{\overline{\pi}} - \widehat{\xi}_h\|_1,
    \end{align*}
    where in the last step we applied Lemma 2.2 from \cite{yardim2024meanfield}.
    Merging the upper bounds for $(\triangle)$ and $(\square)$ and taking expectations, we obtain
    \begin{equation*}
        \mathbb{E} \big[ \| \mu_{h+1}^{\overline{\pi}} - \widehat{\mu}_{h+1} \|_1 \big] \leq \frac{\sqrt{\vert \setS \vert}}{\sqrt{N}} + \left(K_\transitionfn + \frac{K_s}{2}\right)\mathbb{E}\left[\|\xi_h^{\overline{\pi}} - \widehat{\xi}_h\|_1\right]. 
    \end{equation*}
    
    Next, we bound the expected deviation between the post-allocation state distribution in the mean field model and its empirical counterpart in the finite-agent system. 
    Specifically, we consider $\mathbb{E}\left[\| \xi_h^{\overline{\pi}} - \widehat{\xi}_h \|_1\right]$,
     where $ \xi_h^{\overline{\pi}} := \Gamma_{\alpha_h}(L_h^{\overline{\pi}}) $ is the mean field post-allocation distribution induced by policy $ \overline{\pi} $, and $ \widehat{\xi}_h $ is the corresponding empirical distribution in the finite-agent auction, computed from realized allocations rather than via the operator $ \Gamma_{\alpha_h} $. 
    Denote the $\sigma$-algebra induced by $\{s_h^i, a_h^i\}_{i}$ as $\setF_h$ for simplicity, then
    \begin{equation*}
        \mathbb{E} \left[ \| \xi_{h}^{\overline{\pi}} - \widehat{\xi}_{h} \|_1 \big\vert \setF_h \right] \leq \underbrace{\mathbb{E} \left[ \| \xi_{h}^{\overline{\pi}} - \mathbb{E}[\widehat{\xi}_{h}\vert \setF_h ] \|_1 \big\vert \setF_h \right]}_{:= (\diamondsuit)} + \underbrace{\mathbb{E} \left[ \| \widehat{\xi}_{h}- \mathbb{E}[\widehat{\xi}_{h}\vert \setF_h] \|_1 \big\vert \setF_h \right]}_{:= (\heartsuit)}.
    \end{equation*}
    We upper-bound the two terms separately once again. 
    For $(\heartsuit)$ we have:
    \begin{align*}
        (\heartsuit) = \mathbb{E} \left[ \| \widehat{\xi}_{h}- \mathbb{E}[\widehat{\xi}_{h}\vert \setF_h] \|_1 \big\vert \setF_h \right] &= \mathbb{E}\left[\sum_{s\in \setV}\left\vert \widehat{\xi}_{h}(s) - \mathbb{E}[\widehat{\xi}_{h}(s)\vert \setF_h]\right\vert\,\middle\vert \setF_h\right]\\
        &=\sum_{s\in \setV}\mathbb{E}\left[\left\vert \widehat{\xi}_{h}(s) - \mathbb{E}[\widehat{\xi}_{h}(s)\vert \setF_h]\right\vert\,\middle\vert \setF_h\right].
    \end{align*}

     We establish an upper bound on the absolute deviation for each state $ s \in \setV $. 
     Let $ s \in \setV $ be arbitrary. 
     Given the empirical state-action distribution $\widehat{L}_h \in \Delta_{\mathcal{S} \times \mathcal{A}}$, the corresponding empirical state distribution is given by marginalizing over actions: $\widehat{\mu}_h = \sum_{a \in \mathcal{A}} \widehat{L}_h(\cdot ,a)$. By definition of $ \widehat{\mu}_h $, there are $ N\widehat{\mu}_h(s) $ agents in state $ s $ at round $ h $. Denoting these agents as $ i_1, \dots, i_{N\widehat{\mu}_h(s)} $, we can express $ \widehat{\xi}_{h}(s) $ as:
    \begin{equation*}
        \widehat{\xi}_{h}(s) = \frac{1}{N} \sum_{j \in [N\widehat{\mu}_h(s)]} \ind{z_{h}^{i_j} = s}
    \end{equation*}
    Two key observations can be made regarding these indicator variables:  

    1. The indicators are negatively correlated due to the structure of the auction. To illustrate this, assume without loss of generality that the first $M$ agents have not won yet, i.e., $ s_h^i \neq \perp $ for all $ i \in [M] $. Since the number of items in each round is fixed, when conditioning on $\widehat{L}_h$, we have 

   \begin{equation*}
       \left\lfloor\alpha_h(\nu^{-\perp}(\widehat{L}_h)) N \right\rfloor = \sum_{i = 1}^{M} \ind{z_{h}^i = \perp},
   \end{equation*}

   which implies that the indicator variables $\ind{z_{h}^i = \perp}$ are negatively correlated. Consequently, their complements $\ind{z_{h}^i = s_h^i} = 1 - \ind{z_{h}^i = \perp}$ are also negatively correlated. This implies that any subset of these indicator variables retains this negative correlation property. Specifically, for every state $ s $, the random variables  $\ind{z_{h}^{i_j} = s}, \, j \in [N\widehat{\mu}_h(s)],$ are negatively correlated.

    2. Since these are Bernoulli random variables, their variance is at most $ 1/4 $.
    
    It follows from the two observations above that the variance of $\widehat{\xi}_{h}$ conditioned on $\setF_h$ can be upper bounded almost surely as follows:
    \begin{align*}
        \Var{\widehat{\xi}_{h}(s)\middle\vert \setF_h} &= \Var{\frac{1}{N}\sum_{j \in [N\widehat{\mu}_h(s)]}\ind{z_{h}^{i_j} = s} \middle\vert \setF_h}
        = \frac{1}{N^2}\Var{\sum_{j \in [N\widehat{\mu}_h(s)]}\ind{z_{h}^{i_j} = s} \middle\vert \setF_h}\\
        &\leq \frac{1}{N^2}\sum_{j \in [N\widehat{\mu}_h(s)]}\Var{\ind{z_{h}^{i_j} = s} \middle\vert \setF_h} \leq \frac{\widehat{\mu}_h(s)}{4N}.
    \end{align*}
    Using Jensen's inequality, we can bound the absolute deviation using the variance:
    \begin{align*}
        \mathbb{E}\left[\left\vert \widehat{\xi}_{h}(s) - \mathbb{E}[\widehat{\xi}_{h}(s)\vert \setF_h]\right\vert\,\middle\vert \setF_h\right] \leq \sqrt{\Var{\widehat{\xi}_{h}(s)\middle\vert \setF_h}} \leq \sqrt{\frac{\widehat{\mu}_h(s)}{4N}}
    \end{align*}
    Using this result, together with Cauchy-Schwarz's inequality, we can further bound $(\heartsuit)$:
    \begin{align*}
        (\heartsuit) &= \sum_{s\in \setV}\mathbb{E}\left[\left\vert \widehat{\xi}_{h}(s) - \mathbb{E}[\widehat{\xi}_{h}(s)\vert \setF_h]\right\vert\,\middle\vert \setF_h\right]
        \leq \sum_{s\in \setV}\sqrt{\frac{\widehat{\mu}_h(s)}{4N}} \stackrel{C.S.}{\leq} \frac{\sqrt{\vert \setV\vert }}{2\sqrt{N}}
    \end{align*}
    For the term $(\diamondsuit)$, applying the result of \Cref{corollary:lipschitz_Gamma_alpha} yields:
    \begin{align*}
        (\diamondsuit) = \mathbb{E} \left[ \| \xi_{h}^{\overline{\pi}} - \mathbb{E}[\widehat{\xi}_{h}\vert \setF_h] \|_1 \big\vert \setF_h \right] &= \|\Gamma_{\alpha_h}(L_h^{\overline{\pi}})- \Gamma_{\alpha_h}(\widehat{L}_h) \|_1 +  \|\mathbb{E}[\widehat{\xi}_h | \setF_h] - \Gamma_{\alpha_h}(\widehat{L}_h)\|_1\\
        &\leq |\alpha_h(\nu^{-\perp}(L_h^{\overline{\pi}})) - \alpha_h(\nu^{-\perp}(\widehat{L}_h))| + \|L_h^{\overline{\pi}} - \widehat{L}_h\|_1 + \frac{1}{N}\\
        &\leq K_\alpha \|(\nu^{-\perp}(L_h^{\overline{\pi}}) - (\nu^{-\perp}(\widehat{L}_h)\|_1 + \|L_h^{\overline{\pi}} - \widehat{L}_h\|_1 + \frac{1}{N}\\
        &\leq (K_\alpha + 1)\|L_h^{\overline{\pi}} - \widehat{L}_h\|_1 + \frac{1}{N},
    \end{align*}
    where the second to last step follows from the Lipschitz continuity of $\alpha_h$, while the last step follows from $\|(\nu^{-\perp}(L_h^{\overline{\pi}}) - (\nu^{-\perp}(\widehat{L}_h)\|_1 \leq \|L_h^{\overline{\pi}} - \widehat{L}_h\|_1.$ 
    Additionally, $ \|\mathbb{E}[\widehat{\xi}_h | \setF_h] - \Gamma_{\alpha_h}(\widehat{L}_h)\|_1 \leq \frac{1}{N}$ comes from \Cref{lemma:lipschitz_Xi_alpha}, as $| \lfloor N \alpha_h(\nu^{-\perp}(L_h^{\overline{\pi}}))\rfloor - N\alpha_h(\nu^{-\perp}(L_h^{\overline{\pi}})) | \leq 1$.
    
    Combining the upper bounds for $(\heartsuit)$ and $(\diamondsuit)$ we get:
    \begin{align*}
        \mathbb{E}\bigl[\| \xi_h^{\overline{\pi}} - &\widehat{\xi}_h \|_1\bigr] \\
        &\leq (K_\alpha + 1)\mathbb{E}\left[\|L_h^{\overline{\pi}} - \widehat{L}_h\|_1\right] + \frac{\sqrt{\vert \setV \vert}}{2 \sqrt{N}} + \frac{1}{N}\\
        & \leq (K_\alpha + 1 )\biggl(\mathbb{E}\left[\|\mu_h^{\overline{\pi}} - \widehat{\mu}_h\|_1\right]
        + \frac{1}{N} \sum_{i \in [N]} \| \overline{\pi}_{h} - \pi_{h}^i \|_1 +  \sqrt{\frac{\vert \setA \vert \cdot \vert \setV\vert}{N}}\biggr) + \frac{\sqrt{\vert \setV \vert}}{2 \sqrt{N}} + \frac{1}{N},
    \end{align*}
    where the last step follows from \Cref{lemma:state_action_pop_deviation}.
    
    Combining this result with the bound on $ \mathbb{E}[\| \mu_{h+1}^{\overline{\pi}} - \widehat{\mu}_{h+1} \|_1] $, we apply induction on $ h $ to conclude the proof of the lemma.
\end{proof}

\begin{corollary}[Deviation Between Empirical and Mean Field State-Action Population]\label{lemma:state_action_pop_deviation}

    Under the conditions of \Cref{lemma:state_pop_deviation},
    for all $h \in \{0, \dots, H-1\}$, it holds that:
    \begin{align*}
        \mathbb{E}\left[ \| L_h^{\overline{\pi}} - \widehat{L}_h \|_1 \right]
        &\leq \frac{1 - (K_\xi(1 + K_\alpha))^{h+1}}{1 - K_\xi(1 + K_\alpha)}\biggl(\frac{\sqrt{\vert \setS \vert}}{\sqrt{N}} + \frac{\sqrt{\vert \setS \vert \vert \setA \vert}}{\sqrt{N}} \biggr)\\
        & \quad + K_\xi\frac{1 - (K_\xi(1 + K_\alpha))^h}{1- K_\xi(1 + K_\alpha)}\biggl(\frac{\sqrt{\vert \setS \vert }}{2\sqrt{N}} + \frac{1}{N} \biggr)\\
        &\quad + \sum_{h' \leq h} (K_\xi\cdot (K_\alpha + 1 ))^{h - h'} \cdot \frac{1}{N} \sum_{i \in [N]} \| \overline{\pi}_{h'} - \pi_{h'}^i \|_1,
        \end{align*}
        if $ K_\xi(1 + K_\alpha) \neq 1 $, and
    \begin{align*}
        \mathbb{E}\left[ \| L_h^{\overline{\pi}} - \widehat{L}_h \|_1 \right]
        &\leq (h +1) \biggl(\frac{\sqrt{\vert \setS \vert}}{\sqrt{N}} + \frac{\sqrt{\vert \setS \vert \vert \setA \vert}}{\sqrt{N}}\biggr) + hK_\xi \biggl(\frac{\sqrt{\vert \setS \vert}}{2 \sqrt{N}} + \frac{1}{N} \biggr)\\
        &\quad + \sum_{h' \leq h} \frac{1}{N} \sum_{i \in [N]} \| \overline{\pi}_{h'} - \pi_{h'}^i \|_1,
    \end{align*}
    if $ K_\xi(1 + K_\alpha) = 1 $.
\end{corollary}
\begin{proof}
    The upper bound is obtained easily from \Cref{lemma:state_action_pop_to_mu,lemma:state_pop_deviation}.
\end{proof}

\subsubsection{Step 2: Expected Deviation of Winning probabilities}
We derive an explicit upper bound on the expected deviation in winning distributions between the mean field auction and its finite-agent counterpart, under a single-agent deviation from a common policy.

\begin{lemma}[Expected Deviation in Allocation Dynamics]\label{lemma:allocation_dynamics_deviation}
     Let $\Mauc = (\setS, \setA, H, \mu_0, \{\Pauc_h\}_{h=0}^{H-1}, \{\Rauc_h\}_{h=0}^{H-1})$ be a BA-MFG. 
     Let $\{\alpha_h\}_{h=0}^{H-1}$, $\{\transitionfn_h\}_{h=0}^{H-1}$, $\{\paymentfn_h\}_{h=0}^{H-1}$, and $\{\utilityfn_h\}_{h=0}^{H-1}$ denote the allocation thresholds, transition dynamics, payment, and utility functions, respectively, from which $\{\Pauc_h\}_{h=0}^{H-1}$ and $\{\Rauc_h\}_{h=0}^{H-1}$ are derived. 
     Assume these functions are Lipschitz continuous, with respective Lipschitz moduli $K_{\alpha}$, $K_{\transitionfn}$, $K_{\paymentfn}$, and $K_{\utilityfn}$.
    Consider the corresponding finite-agent Batched Auction model with $N$ agents approximated by $\Mauc$. 
    Let $\vecpi = \{ \pi_h^i \}_{h=0,\dots,H-1,\, i \in [N]} \in \Pi_H^N$ denote the joint policy of the population, and let $\widehat{L}_h \in \Delta_{\mathcal{S}}$ be the empirical state-action distribution at round $h$.
    
    Let $\overline{\pi} \in \Pi_H$ arbitrary, and define the associated mean field state-action distribution flow $L^{\overline{\pi}} := \lpopmfa(\overline{\pi})$. Then, for all $h \in \{0, \dots, H-1\}$, the following bounds hold:
    \begin{align*}
        \mathbb{E}\Big[  \big\vert \pwin(s_h^i, a_h^i, L_h^{\overline{\pi}}, \nu^{-\perp}(L_h^{\overline{\pi}})) 
        &- \pwin(s_h^i, a_h^i, \widehat{L}_h, \tfrac{\lfloor N \nu^{-\perp}(\widehat{L}_h)\rfloor}{N}) \big\vert \Big] \\
        &\leq C_{L_h^{\overline{\pi}}, \alpha_h} \cdot \mathbb{E}\left[\|L_h^{\overline{\pi}} - \widehat{L}_h\|_1 \right]
        + \frac{C_{L_h^{\overline{\pi}}, \alpha_h}}{K_\alpha + 1} \cdot \frac{1}{N}
    \end{align*}
    and 
    \begin{align*}
        \mathbb{E}\left[\|P_{\alpha_h}(s_h^i, a_h^{i}, L_h^{\overline{\pi}}) -  P_{N, \alpha_h}^{i}(\vecs_h, \veca_h)\|_1\right] \leq 2C_{L_h^{\overline{\pi}}, \alpha_h} \mathbb{E}\left[\|L_h^{\overline{\pi}} - \widehat{L}_h\|_1 \right] + 2\frac{C_{L_h^{\overline{\pi}}, \alpha_h}}{K_\alpha + 1}\frac{1}{N},
    \end{align*}
    where for an arbitrary state-action distribution $ L \in \Delta_{\mathcal{S} \times \mathcal{A}} $ and $\alpha:\Delta_{\setS}^{\leq 1 } \to [0,1]$, the constant $ C_{L, \alpha} $ is defined as in \Cref{lemma:deviation_winning_prob}.
\end{lemma}

\begin{proof}
    For the first inequality, we have
    \begin{align*}
         &\mathbb{E}\Big[  \big\vert \pwin(s_h^i, a_h^i, L_h^{\overline{\pi}}, \nu^{-\perp}(L_h^{\overline{\pi}})) 
        - \pwin(s_h^i, a_h^i, \widehat{L}_h, \tfrac{\lfloor N \nu^{-\perp}(\widehat{L}_h)\rfloor}{N}) \big\vert \Big]\\
        &\leq \sum_{s,a,L}\mathbb{P}[s_h^i = s, a_h^i = a, \widehat{L}_h = L]\left|\pwin(s_h^i, a_h^i, L_h^{\overline{\pi}}, \nu^{-\perp}(L_h^{\overline{\pi}})) 
        - \pwin(s_h^i, a_h^i, \widehat{L}_h, \tfrac{\lfloor N \nu^{-\perp}(\widehat{L}_h)\rfloor}{N})\right|\\
        &\leq \sum_{s,a,L}\mathbb{P}[s_h^i = s, a_h^i = a, \widehat{L}_h = L]C_{L_h^{\overline{\pi}}, \alpha_h} \cdot \|L_h^{\overline{\pi}} - L\|_1 
        + \frac{C_{L_h^{\overline{\pi}}, \alpha_h}}{K_\alpha + 1} \cdot \frac{1}{N}\\
        &\leq C_{L_h^{\overline{\pi}}, \alpha_h} \mathbb{E}\left[\|L_h^{\overline{\pi}} - \widehat{L}_h\|_1 \right] + \frac{C_{L_h^{\overline{\pi}}, \alpha_h}}{K_\alpha + 1}\frac{1}{N},
    \end{align*}
    where in the second-to-last step we used \Cref{lemma:deviation_winning_prob}.

    Additionally, for the allocation dynamics $P_{N, \alpha_h}$, since each state $s \in \mathcal{V}$ can transition only either to itself or to the inactive state $\perp$, it follows that:
    \begin{align*}
        \mathbb{E}\left[\|P_{\alpha_h}(s_h^i, a_h^{i}, L_h^{\overline{\pi}}) -  P_{N, \alpha_h}^{i}(\vecs_h, \veca_h)\|_1\right] &= \sum_{s\in \setS} \mathbb{E}\left[\|P_{\alpha_h}(s \vert s_h^i, a_h^{i}, L_h^{\overline{\pi}}) -  P_{N, \alpha_h}^{i}(s \vert \vecs_h, \veca_h)\|_1\right]\\
        &= \mathbb{E}\left[\|P_{\alpha_h}(s_h^i \vert s_h^i, a_h^{i}, L_h^{\overline{\pi}}) -  P_{N, \alpha_h}^{i}(s_h^i \vert \vecs_h, \veca_h)\|_1\right]\\
        &\quad+ \mathbb{E}\left[\|P_{\alpha_h}(\perp \vert s_h^i, a_h^{i}, L_h^{\overline{\pi}}) -  P_{N, \alpha_h}^{i}(\perp \vert \vecs_h, \veca_h)\|_1\right]\\
        &= 2 \mathbb{E}\left[\|P_{\alpha_h}(\perp \vert s_h^i, a_h^{i}, L_h^{\overline{\pi}}) -  P_{N, \alpha_h}^{i}(\perp \vert \vecs_h, \veca_h)\|_1\right].
    \end{align*}
    Since the marginal probability of transitioning to the state $ \perp $ corresponds to the winning probability, the bound follows directly from the first inequality.

\end{proof}

\subsubsection{Step 3: Exploitability Deviation for BA-MFG}
Finally we prove the absolute difference in expected reward due to a single-side policy deviation.
\begin{theorem}
    Let $\Mauc = (\setS, \setA, H, \mu_0, \{\Pauc_h\}_{h=0}^{H-1}, \{\Rauc_h\}_{h=0}^{H-1})$ be a Batched Auction Mean Field Game (BA-MFG). Let $\{\alpha_h\}_{h=0}^{H-1}$, $\{\transitionfn_h\}_{h=0}^{H-1}$, $\{\paymentfn_h\}_{h=0}^{H-1}$, and $\{\utilityfn_h\}_{h=0}^{H-1}$ denote the allocation thresholds, transition dynamics, payment, and utility functions, respectively, from which $\{\Pauc_h\}_{h=0}^{H-1}$ and $\{\Rauc_h\}_{h=0}^{H-1}$ are derived. Assume these functions are Lipschitz continuous, with respective Lipschitz constants $K_{\alpha}$, $K_{\transitionfn}$, $K_{\paymentfn}$, and $K_{\utilityfn}$.

    Consider the corresponding finite-agent Batched Auction $\Gauc$ with $ N $ agents, which is approximated by $ \Mauc $. Let $\overline{\pi} \in \Pi_H$ an arbitrary policy satisfying the no-zero dominance property. Then, for any policy $\pi \in \Pi_H, \tau \geq 0$ it holds
    \begin{align*}
        \biggl\vert \Vmfa^\tau(\lpopmfa(\overline{\pi}), \pi) - \Jauc^{\tau, 1}(\pi,\underbrace{ \overline{\pi}, \dots,\overline{\pi} }_{N-1\text{ times}}) \biggr\vert &= \mathcal{O}\left(\frac{1}{\sqrt{N}}\right)
    \end{align*}
\end{theorem}
\begin{proof}
    Define the random variables $ \{s_h^i, a_h^i, z_h^i\}_{i \in [N],\, h \in \{0, \dots, H-1\}} $, along with $ \{\widehat{L}_h\}_{h=0}^{H-1} $, $ \{\widehat{\mu}_h\}_{h=0}^{H-1} $, and $ \{\widehat{\xi}_h\}_{h=0}^{H-1} $, as in the definition of the $ N $-player Batched Auction (see \Cref{definition:n_batched_auction}). Here, $ s_h^i $ denotes the state of agent $ i $ at round $ h $, $ a_h^i $ its action, and $ z_h^i $ its hidden state following the allocation step. The random variables $ \widehat{L}_h $, $ \widehat{\mu}_h $, and $ \widehat{\xi}_h $ represent, respectively, the empirical state-action distribution, the empirical state distribution, and the empirical post-allocation state distribution at round $ h $.

    For the Mean-Field Batched Auction, define $ \{s_h, a_h, z_h\}_{h=0}^{H-1} $, where $ s_h $ and $ a_h $ represent the state and action of a representative agent at round $ h $, and $ z_h $ denotes its post-allocation hidden state. These evolve deterministically according to the mean-field population flows $ L^{\overline{\pi}}, \mu^{\overline{\pi}}, \xi^{\overline{\pi}} $ induced by the population policy $ \overline{\pi} $.

    We divide the proof into three steps:
    \begin{enumerate}
        \item We show that for every $h \in \{0, \dots, H-1\}$ we have:
        \begin{align*}
            \bigl\| \mathbb{P}[s_h = \cdot] - &\mathbb{P}[s_h^1 = \cdot] \bigr\|_1\\
            &\leq \sum_{h' < h} \mathbb{E}\left[\|L_{h'}^{\overline{\pi}} - \widehat{L}_{h'}\|_1\right](K_\transitionfn(K_\alpha + 1) + 2C_{L_{h'}^{\overline{\pi}}, \alpha_{h'}}) + \frac{2 C_{L_{h'}^{\overline{\pi}}, \alpha_{h'}}}{(K_\alpha + 1) N}\\
            &\quad + h K_\transitionfn\biggl(\frac{\sqrt{\vert \setS \vert}}{2\sqrt{N}} + \frac{1}{N} \biggr)
        \end{align*}
        \item We show that for every $h \in \{0, \dots, H-1\}$ we have:
        \begin{align*}
            \bigl| \mathbb{E}[ \Rauc_h(s_h,a_h, L_h^{\overline{\pi}}) - &\Rauc_{N,h}(s_h^1, a_h^1, \widehat{L}_h)] \bigr|\\
            & \leq (K_\utilityfn K_\paymentfn + B_\utilityfn C_{L_h^{\overline{\pi}}, \alpha_h})\mathbb{E}\left[\|L_h^{\overline{\pi}} - \widehat{L}_h\|_1 \right] + \frac{C_{L_h^{\overline{\pi}}, \alpha_h}}{K_\alpha + 1}\frac{B_\utilityfn}{N}\\
            &\quad +  B_{\utilityfn} \|\mathbb{P}[s_h = \cdot] - \mathbb{P}[s_h^1 = \cdot ]\|_1.       
        \end{align*}
        \item In the last step we combine the results of the previous steps to prove the man claim of the theorem.
    \end{enumerate}

    \textbf{Step 1:} We prove the inequality by induction on $h$. For the base case $h = 0$, we have $\|\mathbb{P}[s_0 = \cdot] - \mathbb{P}[s_0^1 = \cdot]\|_1 = 0$, since both distributions are equal to the initial distribution $\mu_0$. Now, assuming the bound holds for some $h \geq 0$, we show that it also holds for round $h + 1$.
    \begin{align*}
        \|\mathbb{P}[s_{h+1} = \cdot] - \mathbb{P}[s_{h+1}^1 = \cdot]\|_1 &= \left\|\sum_{z, \xi}\mathbb{P}[z_h^1  = z, \widehat{\xi}_h = \xi]\transitionfn_h(z, \xi) - \sum_{z}\mathbb{P}[z_h = z]\transitionfn_h(z, \xi_h^{\overline{\pi}}) \right\|_1.
    \end{align*}
    By adding and subtracting $\sum_{z, \xi}\mathbb{P}[z_h^1  = z, \widehat{\xi}_h = \xi]\transitionfn_h(z, \xi_h^{\overline{\pi}}) = \sum_{z}\mathbb{P}[z_h^1  = z]\transitionfn_h(z, \xi_h^{\overline{\pi}}), $ and applying triangular inequality, we get
    \begin{align*}
        \|\mathbb{P}[s_{h+1} = \cdot] - \mathbb{P}[s_{h+1}^1 = \cdot]\|_1
        &\leq\Biggl\|\sum_{z, \xi}\mathbb{P}[z_h^1  = z, \widehat{\xi}_h = \xi]\left(\transitionfn_h(z, \xi) - \transitionfn_h(z, \xi_h^{\overline{\pi}})\right) \Biggr\|_1\\
        &\quad + \Biggl\|\sum_{z}(\mathbb{P}[z_h = z] - \mathbb{P}[z_h^1 = z]) \transitionfn_h(z, \xi_h^{\overline{\pi}}) \Biggr\|_1\\
        &\leq K_\transitionfn\mathbb{E}[\|\xi_h^{\overline{\pi}} - \widehat{\xi}_h\|_1] + \|\mathbb{P}[z_h=\cdot] - \mathbb{P}[z_h^1 = \cdot ]\|_1.
    \end{align*}
    The term $\mathbb{E}[\|\xi_h^{\overline{\pi}} - \widehat{\xi}_h\|_1]$, using the same derivation as in the inductive step of \Cref{lemma:state_pop_deviation}, can be further upper bounded as
    \begin{equation*}
        \mathbb{E}\left[\| \xi_h^{\overline{\pi}} - \widehat{\xi}_h \|_1\right] \leq (K_\alpha + 1)\mathbb{E}\left[\|L_h^{\overline{\pi}} - \widehat{L}_h\|_1\right] + \frac{\sqrt{\vert \setV \vert}}{2 \sqrt{N}} + \frac{1}{N}.
    \end{equation*}
    Finally applying a similar reasoning we can upper bound $\|\mathbb{P}[z_h=\cdot] - \mathbb{P}[z_h^1 = \cdot ]\|_1$.
    \begin{align*}
        &\|\mathbb{P}[z_h=\cdot] - \mathbb{P}[z_h^1 = \cdot ]\|_1\\
        &= \Biggl\| \sum_{\vecs, \veca}\mathbb{P}[\vecs_h = \vecs, \veca_h = \veca]P_{N, \alpha_h}^1(\vecs, \veca) - \sum_{s,a}\mathbb{P}[s_h = s, a_h = a]P_{\alpha_h}( s,a, L_h^{\overline{\pi}}) \Biggr\|_1.
    \end{align*}
    By adding and subtracting $\sum_{\vecs, \veca}\mathbb{P}[\vecs_h = \vecs, \veca_h = \veca]P_{\alpha_h}(s^1,a^1, L_h^{\overline{\pi}}) = \sum_{s,a}\mathbb{P}[s_h^1 = s, a_h^1 = a]P_{\alpha_h}(s,a, L_h^{\overline{\pi}})$, and applying triangular inequality, we get
    \begin{align*}
        \|\mathbb{P}[z_h=\cdot] - \mathbb{P}[z_h^1 = \cdot ]\|_1 &\leq \sum_{\vecs, \veca}\mathbb{P}[\vecs_h = \vecs, \veca_h = \veca]\| P_{N, \alpha_h}^1(\vecs, \veca) - P_{\alpha_h}( s^1,a^1, L_h^{\overline{\pi}})\|_1\\
        &\quad + \sum_{s,a}\|P_{\alpha_h}(s,a, L_h^{\overline{\pi}})\|_1  | \mathbb{P}[s_h = s, a_h = a] - \mathbb{P}[s_h^1 = s, a_h^1 = a]|\\
        &\leq \mathbb{E}\left[\|P_{N, \alpha_h}^{1}(\vecs_h, \veca_h) - P_{\alpha_h}(s_h^1, a_h^{1}, L_h^{\overline{\pi}})\|_1\right]\\
        &\quad + \sum_s\vert \mathbb{P}[s_h = s] -\mathbb{P}[s_h^1 = s] \vert \sum_a\pi_h(a\vert s) \|P_{\alpha_h}(s,a, L_h^{\overline{\pi}})\|_1.
    \end{align*}
    Applying \Cref{lemma:allocation_dynamics_deviation} it follows
    \begin{align*}
         &\|\mathbb{P}[z_h=\cdot] - \mathbb{P}[z_h^1 = \cdot ]\|_1\\
         &\leq 2C_{L_h^{\overline{\pi}}, \alpha_h} \mathbb{E}\left[\|L_h^{\overline{\pi}} - \widehat{L}_h\|_1 \right] + 2\frac{C_{L_h^{\overline{\pi}}, \alpha_h}}{K_\alpha + 1}\frac{1}{N} + \|\mathbb{P}[s_h = \cdot] - \mathbb{P}[s_h^1 = \cdot ]\|_1.
    \end{align*}
    By combining the two bounds, we obtain
    \begin{align*}
        \|\mathbb{P}[s_{h+1} = \cdot] - \mathbb{P}[s_{h+1}^1 = \cdot]\|_1 
        &\leq K_{\transitionfn} \biggl( (K_\alpha + 1)\mathbb{E}\left[\|L_h^{\overline{\pi}} - \widehat{L}_h\|_1\right] + \frac{\sqrt{|\mathcal{V}|}}{2\sqrt{N}}  + \frac{1}{N}\biggr) \\
        &\quad + 2 C_{L_h^{\overline{\pi}}, \alpha_h} \mathbb{E}\left[\|L_h^{\overline{\pi}} - \widehat{L}_h\|_1 \right] + \frac{2 C_{L_h^{\overline{\pi}}, \alpha_h}}{(K_\alpha + 1) N} \\
        &\quad + \|\mathbb{P}[s_h = \cdot] - \mathbb{P}[s_h^1 = \cdot]\|_1.
    \end{align*}
    Applying the induction hypothesis completes the proof for round $h+1$.

    \textbf{Step 2:} We prove the inequality by using the result obtained in step 1. Let $h \in \{0, \dots, H-1\}$ arbitrary, then
    \begin{align*}
        &\bigl| \mathbb{E}[ \Rauc_h(s_h,a_h, L_h^{\overline{\pi}}) - \Rauc_{N,h}(s_h^1, a_h^1, \widehat{L}_h)] \bigr|\\
        & \leq \underbrace{\bigl| \mathbb{E}[ \Rauc_h(s_h,a_h, L_h^{\overline{\pi}}) - \Rauc_h(s_h^1, a_h^1, L^{\overline{\pi}}_h)] \bigr|}_{(\square)} + \underbrace{\bigl| \mathbb{E}[ \Rauc_h(s_h^1, a_h^1, L^{\overline{\pi}}_h) - \Rauc_{N,h}(s_h^1, a_h^1, \widehat{L}_h)] \bigr|}_{(\triangle)}. 
    \end{align*}
    The first term $(\square)$ can be upper bounded by
    \begin{align*}
        (\square) = \Bigl| \sum_{s}(\mathbb{P}[s_h = s] - \mathbb{P}[s_h^1 = s])\sum_a \pi_a(a| s) \Rauc_h(s,a, L_h^{\overline{\pi}})\Bigr| \leq B_{\utilityfn} \|\mathbb{P}[s_h = \cdot] - \mathbb{P}[s_h^1 = \cdot ]\|_1.
    \end{align*}
    For the second term $(\triangle)$ we have:
    \begin{align*}
        (\triangle) &\leq \sum_{s,a,L}\mathbb{P}[s_h^1 = s, a_h^1 = a, \widehat{L} = L] \bigl| \Rauc_h(s, a, L^{\overline{\pi}}_h) - \Rauc_{N,h}(s, a, L)\bigr|\\
        &\leq \sum_{s,a,L}\mathbb{P}[s_h^1 = s, a_h^1 = a, \widehat{L} = L] \bigl| \utilityfn_h(s, \paymentfn_h(a, \nu^{-\perp}(L_h^{\overline{\pi}}))) -  \utilityfn_h(s, \paymentfn_h(a, \nu^{-\perp}(L))) \bigr|\\
        &\quad + B_\utilityfn\mathbb{E}\Big[  \big\vert \pwin(s_h^1, a_h^1, L_h^{\overline{\pi}}, \nu^{-\perp}(L_h^{\overline{\pi}})) 
        - \pwin(s_h^1, a_h^1, \widehat{L}_h, \tfrac{\lfloor N \nu^{-\perp}(\widehat{L}_h)\rfloor}{N}) \big\vert \Big]\\
        &\leq K_\utilityfn K_\paymentfn \sum_{s,a,L} \mathbb{P}[s_h^1 = s, a_h^1 = a, \widehat{L} = L]\|L_h^{\overline{\pi}} - L\|_1\\
        &\quad  + B_\utilityfn C_{L_h^{\overline{\pi}}, \alpha_h} \mathbb{E}\left[\|L_h^{\overline{\pi}} - \widehat{L}_h\|_1 \right] + \frac{C_{L_h^{\overline{\pi}}, \alpha_h}}{K_\alpha + 1}\frac{B_\utilityfn}{N}\\
        &\leq (K_\utilityfn K_\paymentfn + B_\utilityfn C_{L_h^{\overline{\pi}}, \alpha_h})\mathbb{E}\left[\|L_h^{\overline{\pi}} - \widehat{L}_h\|_1 \right] + \frac{C_{L_h^{\overline{\pi}}, \alpha_h}}{K_\alpha + 1}\frac{B_\utilityfn}{N}, 
    \end{align*}
    where in the second to last step we used the Lipschitz continuity of $\utilityfn \circ \paymentfn$ and \Cref{lemma:allocation_dynamics_deviation}. Combining both results we get:
    \begin{align*}
        \bigl| \mathbb{E}[ \Rauc_h(s_h,a_h, L_h^{\overline{\pi}}) - \Rauc_{N,h}(s_h^1, a_h^1, \widehat{L}_h)] \bigr| & \leq (K_\utilityfn K_\paymentfn + B_\utilityfn C_{L_h^{\overline{\pi}}, \alpha_h})\mathbb{E}\left[\|L_h^{\overline{\pi}} - \widehat{L}_h\|_1 \right] + \frac{C_{L_h^{\overline{\pi}}, \alpha_h}}{K_\alpha + 1}\frac{B_\utilityfn}{N}\\
        &\quad +  B_{\utilityfn} \|\mathbb{P}[s_h = \cdot] - \mathbb{P}[s_h^1 = \cdot ]\|_1.       
    \end{align*}
    \textbf{Step 3: }We now combine the results from the previous two steps to establish the final bound stated in the theorem.
    \begin{align*}
        &\biggl\vert \Vmfa^\tau(\lpopmfa(\overline{\pi}), \pi) - \Jauc^{\tau, 1}(\pi, \overline{\pi}, \dots,\overline{\pi}) \biggr\vert\\
        &= \left\vert \mathbb{E}\left[\sum_{h=0}^{H-1}\Rauc_h(s_h,a_h, L_h^{\overline{\pi}}) + \tau \entropy(\pi_h(s_h)) - \Rauc_{N,h}(s_h^1, a_h^1, \widehat{L}_h) - \tau \entropy(\pi_h(s_h^1))\right] \right\vert \\
        &\leq \sum_{h=0}^{H-1}\left| \mathbb{E}\left[ \Rauc_h(s_h,a_h, L_h^{\overline{\pi}}) + \tau \entropy(\pi_h(s_h)) - \Rauc_{N,h}(s_h^1, a_h^1, \widehat{L}_h) - \tau \entropy(\pi_h(s_h^1))\right]\right|.
    \end{align*}
    We proceed by bounding each term individually for every round $h$. Let $h \in \{0, \dotsm H-1\}$ arbitrary, then
    \begin{align*}
        &\left| \mathbb{E}\left[ \Rauc_h(s_h,a_h, L_h^{\overline{\pi}}) + \tau \entropy(\pi_h(s_h)) - \Rauc_{N,h}(s_h^1, a_h^1, \widehat{L}_h) - \tau \entropy(\pi_h(s_h^1))\right]\right|\\
        &\leq \left| \mathbb{E}\left[ \Rauc_h(s_h,a_h, L_h^{\overline{\pi}}) - \Rauc_{N,h}(s_h^1, a_h^1, \widehat{L}_h)\right]\right| + \tau \left| \mathbb{E}\left[ \entropy(\pi_h(s_h)) - \entropy(\pi_h(s_h^1))\right]\right|.
    \end{align*}
    The first term is bounded using the result from Step 2, while the second term can be handled as follows:
    \begin{align*}
        \left| \mathbb{E}\left[ \entropy(\pi_h(s_h)) - \entropy(\pi_h(s_h^1))\right]\right| &\leq \sum_s |\mathbb{P}[s_h = s] - \mathbb{P}[s_h^1 = s]\entropy(\pi_h(s))|\\
        &\leq \|\mathbb{P}[s_h = \cdot] - \mathbb{P}[s_h^1 = \cdot]\|_1 \log(| \setA |).
    \end{align*}
    Therefore the (entropy regularized) absolute difference in rewards at round $h$ can be upper bounded by
    \begin{align*}
    \biggl| \mathbb{E}\biggl[ \Rauc_h(s_h,a_h, L_h^{\overline{\pi}}) + &\tau \entropy(\pi_h(s_h)) - \Rauc_{N,h}(s_h^1, a_h^1, \widehat{L}_h) - \tau \entropy(\pi_h(s_h^1))\biggr]\biggr|\\
    &\leq (K_\utilityfn K_\paymentfn + B_\utilityfn C_{L_h^{\overline{\pi}}, \alpha_h})\mathbb{E}\left[\|L_h^{\overline{\pi}} - \widehat{L}_h\|_1 \right] + \frac{C_{L_h^{\overline{\pi}}, \alpha_h}}{K_\alpha + 1}\frac{B_\utilityfn}{N}\\
    &\quad + (B_{\utilityfn} + \tau\log(| \setA |)) \|\mathbb{P}[s_h = \cdot] - \mathbb{P}[s_h^1 = \cdot ]\|_1.
    \end{align*}
    By applying the bound on $|\mathbb{P}[s_h = \cdot] - \mathbb{P}[s_h^1 = \cdot]|_1$ derived in Step 1, and summing the per-round deviations over all $h$, it follows:
    \begin{align*}
        &\biggl\vert \Vmfa^\tau(\lpopmfa(\overline{\pi}), \pi) - \Jauc^{\tau, 1}(\pi,\underbrace{ \overline{\pi}, \dots,\overline{\pi} }_{N-1\text{ times}}) \biggr\vert\\
        & \leq \sum_{h=0}^{H-1}(K_\utilityfn K_\paymentfn + B_\utilityfn C_{L_h^{\overline{\pi}}, \alpha_h})\mathbb{E}\left[\|L_h^{\overline{\pi}} - \widehat{L}_h\|_1 \right] + \frac{C_{L_h^{\overline{\pi}}, \alpha_h}}{K_\alpha + 1}\frac{B_\utilityfn}{N}\\
        &\quad + (B_\utilityfn + \tau \log(| \setA |))\sum_{h=0}^{H-1} \sum_{h' < h} \mathbb{E}\left[\|L_{h'}^{\overline{\pi}} - \widehat{L}_{h'}\|_1\right](K_\transitionfn(K_\alpha + 1) + 2C_{L_{h'}^{\overline{\pi}}, \alpha_{h'}}) + \frac{2 C_{L_{h'}^{\overline{\pi}}, \alpha_{h'}}}{(K_\alpha + 1) N}\\
        &\quad + (B_\utilityfn + \tau \log(| \setA |)) \frac{H(H-1)}{2}K_\transitionfn \biggl(\frac{\sqrt{\vert \setS \vert}}{2\sqrt{N}} + \frac{1}{N} \biggr)
    \end{align*}
    In particular, \Cref{lemma:state_action_pop_deviation} implies that the total (entropy-regularized) reward difference is of order $\mathcal{O}(\tfrac{1}{\sqrt{N}})$.
\end{proof}

\textbf{Conclusion and Statement of Result.} Let $\Mauc$ be a BA-MFG with Lipschitz-continuous $\{\utilityfn_h\}_{h=0}^{H-1}, \{\transitionfn_h\}_{h=0}^{H-1}, \{\alpha_h\}_{h=0}^{H-1}, \{\paymentfn_h\}_{h=0}^{H-1}$ Let $\pi_\delta \in \Pi_H$ be a policy that satisfy the no zero-dominance property. Let further assume $\pi_{\delta}$ is a $\delta-$MFG-NE, namely
\begin{equation*}
    \delta \geq \max_{\pi' \in \Pi_H}\Vmfa^\tau(\lpopmfa(\pi_\delta), \pi') - \Vmfa^\tau(\lpopmfa(\pi_\delta), \pi_\delta). 
\end{equation*}
Then, for $\vecpi = (\pi_\delta, \dots, \pi_\delta)$, we have:
\begin{align*}
    &\max_{\pi' \in \Pi_H} \Jauc^{\tau, i}(\pi', \vecpi^{-i}) - \Jauc^{\tau, i}(\vecpi)\\
    &= \max_{\pi' \in \Pi_H} \Jauc^{\tau, i}(\pi', \vecpi^{-i}) - \Jauc^{\tau, i}(\vecpi) \\
    &\quad + \Vmfa^\tau(\lpopmfa(\pi_\delta), \pi') - \Vmfa^\tau(\lpopmfa(\pi_\delta), \pi') + \Vmfa^\tau(\lpopmfa(\pi_\delta), \pi_\delta) - \Vmfa^\tau(\lpopmfa(\pi_\delta), \pi_\delta)\\
    &\leq \max_{\pi' \in \Pi_H}\Vmfa^\tau(\lpopmfa(\pi_\delta), \pi') - \Vmfa^\tau(\lpopmfa(\pi_\delta), \pi_\delta)\\
    &\quad + \left|\Vmfa^\tau(\lpopmfa(\pi_\delta), \pi') - \Jauc^{\tau, i}(\pi', \vecpi^{-i})\right| + \left|\Vmfa^\tau(\lpopmfa(\pi_\delta), \pi_\delta) - \Jauc^{\tau, i}(\vecpi)\right|\\
    &\leq \delta + \frac{C_1}{\sqrt{N}} + \frac{C_2}{N} = \delta + \mathcal{O}\left(\frac{1}{\sqrt{N}}\right).
\end{align*}

In case $K_\xi(K_\alpha + 1) \not= 1$, the constants $C_1$ and $C_2$ are given by
\begin{align*}
    C_1 &=2\sum_{h=0}^{H-1}(K_\utilityfn K_\paymentfn + B_{\utilityfn}C_{L_h^{\pi_\delta}, \alpha_h})\frac{1 - (K_\xi(1 + K_\alpha))^{h+1}}{1 - K_\xi(1 + K_\alpha)}\left(\sqrt{\vert \setS \vert} + \sqrt{\vert \setS \vert \vert \setA \vert} \right)\\
    &\quad + \sum_{h=0}^{H-1}(K_\utilityfn K_\paymentfn + B_{\utilityfn}C_{L_h^{\pi_\delta}, \alpha_h}) K_\xi\frac{1 - (K_\xi(1 + K_\alpha))^h}{1- K_\xi(1 + K_\alpha)}\sqrt{\vert \setS \vert }\\
    &\quad + \sum_{h=0}^{H-1}\sum_{h' < h} 
    \left(K_\transitionfn(K_\alpha + 1) + 2C_{L_{h'}^{\pi_\delta}, \alpha_{h'}}\right) \\
    &\qquad \cdot \left[ 
        2(B_\utilityfn + \tau \log|\setA|) \frac{1 - (K_\xi(1 + K_\alpha))^{h'+1}}{1 - K_\xi(1 + K_\alpha)} 
        \left( \sqrt{|\setS|} + \sqrt{|\setS||\setA|} \right) \right. \\
    &\qquad\quad + \left. (B_\utilityfn + \tau \log|\setA|) K_\xi \frac{1 - (K_\xi(1 + K_\alpha))^{h'}}{1 - K_\xi(1 + K_\alpha)} \sqrt{|\setS|} 
        \right]\\
    &\quad + (B_\utilityfn + \tau \log(| \setA |))H(H-1)\frac{K_\transitionfn\sqrt{|\setS|}}{2},
\end{align*}
\begin{align*}
    C_2 &= 2\sum_{h = 0}^{H-1}(K_\utilityfn K_\paymentfn + B_\utilityfn C_{L_h^{\pi_\delta}, \alpha_h })\biggl(\frac{1 - (K_\xi(1 + K_\alpha))^{h+1}}{1 - (K_\xi(1 + K_\alpha))} + K_\xi\frac{1 - (K_\xi(1 + K_\alpha))^{h}}{1 - (K_\xi(1 + K_\alpha))}\biggr)\\
    &\quad +  2\sum_{h = 0}^{H-1}B_\utilityfn \frac{C_{L_h^{\pi_\delta}, \alpha_{h}}}{K_\alpha + 1 }\\
    &\quad + 2(B_\utilityfn + \tau \log|\setA|)\sum_{h=0}^{H-1}\sum_{h' < h}(K_\transitionfn(K_\alpha + 1) + 2C_{L_{h'}^{\pi_\delta}, \alpha_{h'}})\frac{1 - (K_\xi(1 + K_\alpha))^{h'+1}}{1 - (K_\xi(1 + K_\alpha))}\\
    &\quad + 2(B_\utilityfn + \tau \log|\setA|)\sum_{h=0}^{H-1}\sum_{h' < h}(K_\transitionfn(K_\alpha + 1) + 2C_{L_{h'}^{\pi_\delta}, \alpha_{h'}})K_\xi\frac{1 - (K_\xi(1 + K_\alpha))^{h'}}{1 - (K_\xi(1 + K_\alpha))}\\
    &\quad + 2(B_\utilityfn + \tau \log|\setA|)\sum_{h=0}^{H-1}\sum_{h' < h} \frac{2C_{L_{h'}^{\pi_\delta}, \alpha_{h'}}}{K_\alpha + 1}\\
    &\quad + (B_\utilityfn + \tau \log|\setA|) K_\transitionfn H(H-1).
\end{align*}
In case $K_\xi(K_\alpha + 1) = 1$, the constants $C_1$ and $C_2$ are given by
\begin{align*}
    C_1 &= 2\sum_{h=0}^{H-1}(h + 1) \cdot (K_\utilityfn K_\paymentfn + B_{\utilityfn}C_{L_h^{\pi_\delta}, \alpha_h}) \left(\sqrt{\vert \setS \vert} + \sqrt{\vert \setS \vert \vert \setA \vert} \right)\\
    &\quad + \sum_{h=0}^{H-1}h\cdot (K_\utilityfn K_\paymentfn + B_{\utilityfn}C_{L_h^{\pi_\delta}, \alpha_h}) K_\xi\sqrt{\vert \setS \vert }\\
    &\quad  + 2(B_\utilityfn + \tau \log|\setA|)\sum_{h=0}^{H-1}\sum_{h' < h} (h'+1) \cdot
    \left(K_\transitionfn(K_\alpha + 1) + 2C_{L_{h'}^{\pi_\delta}, \alpha_{h'}}\right) \left( \sqrt{|\setS|} + \sqrt{|\setS||\setA|} \right)\\
    &\quad + (B_\utilityfn + \tau \log|\setA|)\sum_{h=0}^{H-1}\sum_{h' < h} h' \cdot
    \left(K_\transitionfn(K_\alpha + 1) + 2C_{L_{h'}^{\pi_\delta}, \alpha_{h'}}\right)K_\xi \sqrt{|\setS|}\\
    &\quad + (B_\utilityfn + \tau \log(| \setA |))H(H-1)\frac{K_\transitionfn\sqrt{|\setS|}}{2},
\end{align*}
\begin{align*}
    C_2 &= 2 \sum_{h=0}^{H-1} (K_\utilityfn K_\paymentfn + B_\utilityfn C_{L_h^{\pi_\delta}, \alpha_h})(h + 1 + K_\xi h) + B_\utilityfn \frac{C_{L_h^{\pi_\delta}, \alpha_h }}{K_\alpha + 1}\\
    &\quad + 2(B_\utilityfn + \tau \log(|\setA|)) \sum_{h= 0}^{H-1}\sum_{h' < h}(K_\transitionfn(K_\alpha + 1) + 2C_{L_{h'}^{\pi_\delta}, \alpha_{h'}})(h' + 1 + K_\xi h')\\
    &\quad + 2(B_\utilityfn + \tau \log(|\setA|)) \sum_{h= 0}^{H-1}\sum_{h' < h}\frac{2C_{L_{h'}^{\pi_\delta}, \alpha_{h'}}}{K_\alpha + 1}\\
    &\quad + (B_\utilityfn + \tau \log(|\setA|))K_\transitionfn H (H-1).
\end{align*}

\begin{remark}
    For policies with full support, the Lipschitz constant associated with the winning probability is of order $ \frac{1}{\epsilon} $, as discussed in \Cref{remark:Lipschitz_pwin_full_support}. Consequently, the constants $ C_1 $ and $ C_2 $ scale as $ \mathcal{O}(\frac{1}{\epsilon} H^2 \cdot \frac{1 - (K_\xi(K_\alpha + 1))^{H}}{1 - K_\xi(K_\alpha + 1)}) $ when $ K_\xi(K_\alpha + 1) \neq 1 $, and as $ \mathcal{O}(\frac{1}{\epsilon} H^3) $ when $ K_\xi(K_\alpha + 1) = 1 $.
\end{remark}

\textbf{Explanation of Constants.} In the above expression, the constants represent key components of the BA-MFG dynamics:
\begin{itemize}
    \item $B_\utilityfn$ and $K_\utilityfn$ are the bound and Lipschitz constant of the utility function $\utilityfn$ respectively.
    \item $B_\paymentfn$ and $K_\paymentfn$ are the bound and Lipschitz constant of the payment function $\paymentfn$ respectively.
    \item $K_\transitionfn$ denotes the Lipschitz constant of the transition function $\transitionfn$.
    \item $K_\alpha$ is the Lipschitz constant of the allocation threshold function $\alpha$.
    \item $K_s =  \sup_{s, s', \xi} \| \transitionfn(s, \xi) - \transitionfn( s', \xi) \|_1$.
    \item $K_\xi = K_\transitionfn + \frac{1}{2}K_s$.
    \item $C_{L, \alpha}$ is the Lipschitz constant of the winning probability function evaluated at the distribution $L$, assuming $L$ satisfies the no zero-dominance property. For its precise definition, see \Cref{lemma:deviation_winning_prob}.
    \item $\tau$ is the entropy regularization parameter.
\end{itemize}

\subsection{Proof of \texorpdfstring{\Cref{theorem:approximation_auction}}{}, part 2 (Approximation in Objective)}
\label{section:objective_convergence}

We show that, under Lipschitz conditions, the objective computed under the mean field approximation closely matches its expected value under a finite population of agents.

\begin{theorem}[Convergence of the Mechanism Objective]\label{theorem:lipschitz_objective}
    Let $ g: \Delta_{\setS \times \setA}^{H} \to \mathbb{R} $ be a Lipschitz-continuous objective defined over the class of Batched Auction Mean Field Games (BA-MFGs). Let $\Mauc$ be a BA-MFG with Lipschitz-continuous $\{\utilityfn_h\}_{h=0}^{H-1}, \{\transitionfn_h\}_{h=0}^{H-1}, \{\alpha_h\}_{h=0}^{H-1}, \{\paymentfn_h\}_{h=0}^{H-1}$. Let $\vecpi = (\pi, \dots, \pi)$ be the joint population policy for some $\pi \in \Pi_H$. Then:
    \begin{equation*}
        \left| g(\lpopmfa(\pi)) - G(\vecpi)\right| = \mathcal{O}\left(\tfrac{1}{\sqrt{N}}\right),
    \end{equation*}
    where $G(\vecpi) = \mathbb{E}[g(\{\widehat{L}_h\}_{h=0}^{H-1}) | \vecpi]$.
\end{theorem}

\begin{proof}
    We use a decomposition over the support of $ \widehat{\vecpop} $ and apply the triangle inequality:
    \begin{align*}
    \left| g(\lpopmfa(\pi)) - G(\vecpi) \right| 
    &= \left| \sum_{\vecpop} \mathbb{P}[\widehat{\vecpop} = \vecpop | \vecpi] \left(g(\lpopmfa(\pi)) - g(\vecpop) \right) \right| \\
    &\leq \sum_{\vecpop} \mathbb{P}[\widehat{\vecpop} = \vecpop | \vecpi] \left| g( \lpopmfa(\pi)) - g(\vecpop) \right| \\
    &\leq \sum_{\vecpop} \mathbb{P}[\widehat{\vecpop} = \vecpop | \vecpi ] K_g \| \lpopmfa(\pi) - \vecpop \|_1 \\
    &= K_g\mathbb{E}\left[ \| \lpopmfa(\pi) - \widehat{\vecpop} \|_1 \right],
    \end{align*}
    where $ K_{g}$ is the Lipschitz constant of $ g $. The result follows by applying the bound from \Cref{lemma:state_action_pop_deviation}.
\end{proof}
\begin{lemma}[Lipschitz Continuity of Expected Revenue]\label{lemma:lipschitz_revenue}
    Let $\objrev$ denote the expected revenue objective. Let $\Mauc$ be a BA-MFG with Lipschitz-continuous payment functions $\{\paymentfn_h\}_{h=0}^{H-1}$ and allocation functions $\{\alpha_h\}_{h=0}^{H-1}$. Let $\vecpop = \{L_h\}_{h=0}^{H-1}$ and $\vecpop' = \{L_h'\}_{h=0}^{H-1}$ be two arbitrary state-action distribution trajectories over $H$ rounds. Then,
    \[
    \left| \objrev(\vecpop) - \objrev(\vecpop') \right| \leq (B_\paymentfn(2 + K_\alpha) + K_\paymentfn) \| \vecpop - \vecpop' \|_1,
    \]
    where $B_{\paymentfn}$ is a uniform bound on the absolute value of the payment functions, $K_{\paymentfn}$ is their Lipschitz constant, and $K_\alpha$ is the Lipschitz constant of the allocation threshold functions.
\end{lemma}

\begin{proof}
    The expected revenue, for $\vecpop = \{L_h\}_{h=0}^{H-1}$, can be rewritten using the operator $\Xi$
    \begin{equation*}
        \objrev(\vecpop) := \sum_{h = 0}^{H-1} \sum_{s,a} 
        \left( L_h(s, a) - \Xi_{\alpha_h(\nu^{-\perp}(L_h))}(L_h)(s, a) \right)
        \paymentfn_h\left(a, \nu^{-\perp}(L_h) \right).
    \end{equation*}
    To simplify notation, we define
    \begin{equation*}
        \overline{\Xi}_\alpha(L) := L - \Xi_{\alpha}(L),
    \end{equation*}
    representing the residual (unallocated) mass at each state-action pair. Applying the triangle inequality:
    \begin{align*}
        \left| \objrev(\vecpop) - \objrev(\vecpop') \right|
        &\leq \sum_{h = 0}^{H-1} 
        \Bigg| \sum_{s,a} \left( \overline{\Xi}_{\alpha_h(\nu^{-\perp}(L_h))}(L_h)(s, a) \right) \paymentfn_h(a, \nu^{-\perp}(L_h))\\
        &\quad \quad- \left(\overline{\Xi}_{\alpha_h(\nu^{-\perp}(L_h'))}(L'_h)(s, a) \right) \paymentfn_h(a, \nu^{-\perp}(L_h')) \Bigg|\\
        &\leq \sum_{h = 0}^{H-1}\sum_{s,a}\left| \overline{\Xi}_{\alpha_h(\nu^{-\perp}(L_h))}(L_h)(s, a) - \overline{\Xi}_{\alpha_h(\nu^{-\perp}(L_h'))}(L_h')(s, a)\right| \left|\paymentfn_h(a, \nu^{-\perp}(L_h))\right|\\
        &\quad + \sum_{h = 0}^{H-1}\sum_{s,a}\overline{\Xi}_{\alpha_h(\nu^{-\perp}(L_h'))}(L_h')(s, a) \left|\paymentfn_h(a, \nu^{-\perp}(L_h)) - \paymentfn_h(a, \nu^{-\perp}(L_h')) \right|.
    \end{align*}
    Using the boundedness of the payment function $\paymentfn$ and the Lipschitz continuity of the allocation operator $\Xi_\alpha$, the first term can be bounded by $B_\paymentfn(2 + K_\alpha)\sum_{h = 0}^{H-1} \|L_h - L_h'\|_1$. For the second term, the Lipschitz property of $\paymentfn$ implies a bound of $\sum_{h = 0}^{H-1} K_\paymentfn \|L_h - L_h'\|_11$. Combining these, we conclude that the revenue objective $\objrev$ is Lipschitz continuous with constant $B_\paymentfn(2 + K_\alpha) + K_\paymentfn$, and satisfies the bound 
    \begin{equation*}
        \left| \objrev(\vecpop) - \objrev(\vecpop') \right| \leq (B_\paymentfn(2 + K_\alpha) + K_\paymentfn)\|\vecpop - \vecpop'\|_1
    \end{equation*}
    
\end{proof}

\begin{corollary}[Convergence of Expected Revenue]
Let $ \objrev$  the expected revenue objective. Let $\Mauc$ be a BA-MFG with Lipschitz-continuous $\{\utilityfn_h\}_{h=0}^{H-1}, \{\transitionfn_h\}_{h=0}^{H-1}, \{\alpha_h\}_{h=0}^{H-1}, \{\paymentfn_h\}_{h=0}^{H-1}$. Let $\vecpi = (\pi, \dots, \pi)$ be the joint population policy for some $\pi \in \Pi_H$. Then:
\[
\left| \objrev(\lpopmfa(\vecpi)) - G_{\text{rev}}(\vecpi) \right| = \mathcal{O}\left( \tfrac{1}{\sqrt{N}} \right).
\]
\end{corollary}

\begin{proof}
The result follows by combining \Cref{theorem:lipschitz_objective,lemma:lipschitz_revenue}.
\end{proof}

\subsection{Proof of \texorpdfstring{\Cref{lemma:qauclipschitz}}{}}
\label{section:quaclipschitz}
In this section, we prove that the (entropy regularized) q-functions are Lipschitz continuous with respect to the population policy, assuming full support. We begin by showing that the population flow is Lipschitz in the policy. Next, we establish that both the transition dynamics and the reward function are Lipschitz continuous with respect to the population distribution. Finally, we combine these results to derive a bound on the Lipschitz constant of the (entropy regularized) q-functions.

\begin{lemma}[Lipschitz Continuity of Population Operator]\label{lemma:Lipschitz_lpopmfa}
    Let $\Mauc = (\setS, \setA, H, \mu_0, \{\Pauc_h\}_{h=0}^{H-1}, \{\Rauc_h\}_{h=0}^{H-1})$ be a Batched Auction Mean Field Game (BA-MFG). Let $\{\alpha_h\}_{h=0}^{H-1}$, $\{\transitionfn_h\}_{h=0}^{H-1}$, $\{\paymentfn_h\}_{h=0}^{H-1}$, and $\{\utilityfn_h\}_{h=0}^{H-1}$ denote the allocation thresholds, transition dynamics, payment, and utility functions, respectively, from which $\{\Pauc_h\}_{h=0}^{H-1}$ and $\{\Rauc_h\}_{h=0}^{H-1}$ are derived. Assume these functions are Lipschitz continuous. Consider two arbitrary policies $\pi, \pi' \in \Pi_H$, then
    \begin{equation*}
        \|\lpopmfa(\pi) - \lpopmfa(\pi')\|_1 \leq \sum_{h'\leq h}(K_\xi (1 + K_\alpha))^{h- h'}\|\pi_{h'} - \pi'_{h'}\|_1
    \end{equation*}
\end{lemma}
\begin{proof}
    We prove the bound inductively. Let $ L^{\pi} := \lpopmfa(\pi) $ and $ L^{\pi'} := \lpopmfa(\pi') $ denote the population flows induced by policies $ \pi $ and $ \pi' $, respectively. Similarly, let $ \mu^{\pi} $ and $ \mu^{\pi'} $ denote the corresponding marginal state distributions.

    For $h = 0$ we have $\mu_0^{\pi} = \mu_0^{\pi'} = \mu_0$. Therefore by \Cref{lemma:L_mu_relation} we have $\|L_0^{\pi} - L_0^{\pi'}\|_1 \leq \|\pi_0 - \pi_0'\|_1$.

    For $h + 1 > 0$ by applying the result of \Cref{lemma:L_mu_relation} we have
    \begin{equation*}
        \|L_{h+1}^{\pi} - L_{h+1}^{\pi'}\|_1 \leq \|\pi_{h+1} - \pi_{h+1}'\|_1 + \|\mu_{h+1}^{\pi} - \mu_{h+1}^{\pi'}\|_1.
    \end{equation*}
    We then bound the variational difference in state distribution:
    \begin{align*}
        \|\mu_{h+1}^{\pi} - \mu_{h+1}^{\pi'}\|_1 &= \|\Gamma_{\transitionfn}(\xi_h^{\pi}) - \Gamma_{\transitionfn}(\xi_h^{\pi'})\|_1
        \leq K_\xi\|\xi_h^{\pi} -\xi_h^{\pi'}\|_1\\
        & = K_\xi\|\Gamma_{\alpha_h}(L_h^{\pi}) - \Gamma_{\alpha_h}(L_h^{\pi'})\|_1
        \leq K_\xi (1 + K_\alpha)\|L_h^{\pi} - L_h^{\pi'}\|_1,
    \end{align*}
    where in the last step we used \Cref{corollary:lipschitz_Gamma_alpha}. 
    By induction over $h$ the claim follows.
\end{proof}

\begin{lemma}[Lipschitz Continuity of Transitions and Rewards under Full-Support Policies]\label{lemma:lipschitz_mfg_rewards_transition}
    Let $\Mauc$ be a BA-MFG with utility functions $\{\utilityfn_h\}_{h=0}^{H-1}$, transition dynamics $\{\transitionfn_h\}_{h=0}^{H-1}$, payment functions $\{\paymentfn_h\}_{h=0}^{H-1}$, and allocation thresholds $\{\alpha_h\}_{h=0}^{H-1}$, all of which are Lipschitz-continuous with constants $K_{\utilityfn}$, $K_{\transitionfn}$, $K_{\paymentfn}$, and $K_\alpha$, respectively. Consider two policies $ \pi, \pi' \in \Pi_H $ with full support; that is, for all $ s \in \mathcal{S} $, $ a \in \mathcal{A} $, and $ h \in {0, \dots, H-1} $, we have $ \pi_h(a \vert s) > \epsilon $ and $ \pi_h'(a \vert s) > \epsilon $ for some constant $ \epsilon > 0 $. Then
    \begin{align*}
        \bigl\vert \Rauc_h(s, a, L_h^{\pi}) - &\Rauc_h(s, a, L_h^{\pi'}) \bigr\vert\\
        &\leq \left(\frac{B_\utilityfn }{(1- \alpha_{\text{max}})\epsilon} + K_\utilityfn K_\paymentfn\right)\sum_{h'\leq h}(K_\xi (1 + K_\alpha))^{h- h'}\|\pi_{h'} - \pi'_{h'}\|_1,
    \end{align*}
    and 
    \begin{align*}
        \bigl\|\Pauc_h(s,a,L_h^{\pi}) - &\Pauc_h(s,a,L_h^{\pi'})\bigr\|_1\\
        &\leq \left(\frac{2}{(1- \alpha_{\text{max}})\epsilon}  + K_\alpha + 1\right)\sum_{h'\leq h}(K_\xi (1 + K_\alpha))^{h- h'}\|\pi_{h'} - \pi'_{h'}\|_1.
    \end{align*}
\end{lemma}
\begin{proof}
    Let $\pi, \pi' \in \Pi_H$ be two arbitrary policies with full support Let $L_h^{\pi} = \lpopmfa(\pi)$ and $L_h^{\pi'} = \lpopmfa(\pi')$. We then bound separately the rewards and the transition probabilities.

    We first bound bound the absolute difference in rewards, let $s \in \setS, a \in \setA$ arbitrary, assume $s \not= \perp$, else the claim holds trivially, then, by applying triangular inequality we have:
    \begin{align*}
        &\left| \Rauc_h(s,a, L_h^{\pi'}) - \Rauc_h(s,a, L_h^{\pi'})\right|\\
        & \leq B_\utilityfn \left| \pwin(s,a, L_h^{\pi}, \alpha_h(\nu^{-\perp}(L_h^{\pi}))) - \pwin(s,a, L_h^{\pi'}, \alpha_h(\nu^{-\perp}(L_h^{\pi'})))\right|\\
        &\quad + \left| \utilityfn_h(s, \paymentfn_h(a, \nu^{-\perp}(L_h^{\pi}))) - \utilityfn_h(s, \paymentfn_h(a, \nu^{-\perp}(L_h^{\pi'})))\right|\\
        &\leq \frac{B_\utilityfn }{(1- \alpha_{\text{max}})\epsilon} \|L_h^{\pi} - L_h^{\pi'}\|_1 + K_\utilityfn K_\paymentfn \|L_h^{\pi} - L_h^{\pi'}\|_1,
    \end{align*}
    where the last step follows by \Cref{remark:Lipschitz_pwin_full_support} and the Lipschitz continuity of $\paymentfn$ and $\utilityfn$.
    By \Cref{lemma:Lipschitz_lpopmfa} the bound for the rewards follows.
    
    Similarly for the transition probabilities we have:
    \begin{align*}
        &\|\Pauc_h(s,a,L_h^{\pi}) - \Pauc_h(s,a,L_h^{\pi'})\|_1\\
        & \leq \sum_z \|\transitionfn_h(z | \Gamma_{\alpha_h}(L_h^{\pi'}))\|_1 \left|P_{\alpha_h}(z | s,a,L_h^{\pi'}) - P_{\alpha_h}(z | s,a,L_h^{\pi}) \right|\\
        & \quad + \sum_z P_{\alpha_h}(z | s,a, L_h^{\pi}) \|\transitionfn_h(z, \Gamma_{\alpha_h}(L_h^{\pi})) - \transitionfn_h(z, \Gamma_{\alpha_h}(L_h^{\pi'}))\|_1\\
        &\leq \|P_{\alpha_h}(s,a,L_h^{\pi'}) - P_{\alpha_h}(s,a,L_h^{\pi})\|_1 + K_\transitionfn \|\Gamma_{\alpha_h}(L_h^{\pi}) - \Gamma_{\alpha_h}(L_h^{\pi'})\|_1. 
    \end{align*}
    From \Cref{remark:Lipschitz_pwin_full_support} and \Cref{corollary:lipschitz_Gamma_alpha} it follows
    \begin{equation*}
        \|\Pauc_h(s,a,L_h^{\pi}) - \Pauc_h(s,a,L_h^{\pi'})\|_1 \leq \left(\frac{2}{(1- \alpha_{\text{max}})\epsilon}  + K_\alpha + 1\right)\|L_h^{\pi} - L_h^{\pi'}\|_1,
    \end{equation*}
    by \Cref{lemma:Lipschitz_lpopmfa} the bound follows.
    
\end{proof}

\begin{lemma}[Lipschitz Continuity of Regularized Value Functions]
    Let $\Mauc$ be a BA-MFG with utility functions $\{\utilityfn_h\}_{h=0}^{H-1}$, transition dynamics $\{\transitionfn_h\}_{h=0}^{H-1}$, payment functions $\{\paymentfn_h\}_{h=0}^{H-1}$, and allocation thresholds $\{\alpha_h\}_{h=0}^{H-1}$, all of which are Lipschitz-continuous with constants $K_\utilityfn$, $K_\transitionfn$, $K_\paymentfn$, and $K_\alpha$, respectively.
    
    Let $ \pi, \pi' \in \Pi_H $ be two policies with full support, i.e., for some $ \epsilon > 0 $, $ \pi_h(a \vert s), \pi_h'(a \vert s) > \epsilon $, for all $ h, s, a $. Let $ V_h^\tau(s\vert L^{\pi}, \pi) $ and $ V_h^\tau(s\vert L^{\pi'}, \pi') $ denote the entropy-regularized value functions under policies $ \pi $ and $ \pi' $, defined recursively as
    \begin{equation*}
        V_h^\tau(s\vert L^{\pi}, \pi) := \tau H(\pi_h(s)) + \sum_{a \in \mathcal{A}} \pi_h(a \vert s) \left[ \Rauc_h(s, a, L_h^\pi) + \sum_{s'} \Pauc_h(s' \vert s, a, L_h^\pi) V_{h+1}^\tau(s'\vert L^{\pi}, \pi) \right],    
    \end{equation*}
    where $ \{L_h^\pi\}_{h=0}^{H-1} := \lpopmfa(\pi) $ is the population distribution induced by policy $ \pi $.
    
    Then, there exists a constant $ C > 0 $ such that for all $ h $ and all $ s \in \mathcal{S} $, the following bound holds:
    \begin{align*}
       \left| V_h^\tau(s\vert L^{\pi}, \pi) - V_h^\tau(s\vert L^{\pi'}, \pi') \right|&\leq C \|\pi - \pi'\|_1,
    \end{align*}
    where $C = \mathcal{O}(\frac{1}{\epsilon}H^2)$ when $K_\xi(1 + K_\alpha) = 1$, and $C = \mathcal{O}(\frac{1}{\epsilon}H \tfrac{1- (K_\xi(1 + K_\alpha))^H}{1 -K_\xi(1 + K_\alpha)})$ when $K_\xi(1 + K_\alpha) \not= 1$.

\end{lemma}

\begin{proof}
    We prove a stronger inductive bound that implies the result of the lemma. To simplify notation, we introduce the following constants:
    \begin{itemize}
        \item $\beta := K_\xi(1 + K_\alpha)$,
        \item $B := B_\utilityfn + \tau \log(|\setA|)$,
        \item $g_0 := \frac{B_\utilityfn}{(1 - \alpha_{\max})\epsilon} + K_\utilityfn K_\paymentfn$,
        \item $g_1 := (B_\utilityfn + \tau \log(|\setA|))\left(\frac{2}{(1 - \alpha_{\max})\epsilon} + K_\alpha + 1\right)$.
    \end{itemize}
    
    We also define the following $h$-dependent quantities, which will be used in the inductive argument:
    \begin{itemize}
        \item $A_h := \tau(\log(\tfrac{1}{\epsilon}) + 1) + (H - h)B$,
        \item $G_h := g_0 + g_1(H - h - 1)$,
        \item $\Delta_h := \|\pi_h - \pi_h'\|_1$.
    \end{itemize}
    
    We proceed by backward induction on the round $h$, and show that the following bound holds:
    \begin{align*}
        \left| V_h^\tau(s \mid L^\pi, \pi) - V_h^\tau(s \mid L^{\pi'}, \pi') \right| 
    \leq \sum_{h' = h}^{H-1} A_{h'} \Delta_{h'} 
    + \sum_{h' = h}^{H-1} G_{h'} \sum_{h'' \leq h'} \beta^{h' - h''} \Delta_{h''}.    
    \end{align*}

    For $h = H$ by definition, the value function at round $H$ is zero for all states, i.e.,
    \begin{equation*}
        \Vmfg_H^{\tau}(s\vert L^{\pi}, \pi) = \Vmfg_H^{\tau}(s\vert L^{\pi}, \pi) = 0 
    \end{equation*}

    Assume that for value function at time $h+1$, $\Vmfg_{h+1}^{\tau}$ the upper bound holds. We now prove the same for $V_h^{\tau}$.
    The regularized value function is:
    \begin{equation*}
        \Vmfg_h^{\tau}(s\vert L^{\pi}, \pi) = \tau \entropy(\pi_h(s)) + \sum_a \pi_h(a\vert s)\qmfg_h^{\tau}(s,a \vert L^{\pi}, \pi). 
    \end{equation*}
    Then,
    \begin{align*}
        &\vert \Vmfg_h^{\tau}(s\vert L^{\pi}, \pi) - \Vmfg_h^{\tau}(s\vert L^{\pi'}, \pi') \vert\\
        &\leq \underbrace{\tau \vert \entropy(\pi_h(s)) -  \entropy(\pi_h'(s)) \vert}_{:= (\square)} + \underbrace{\left\vert \sum_a \pi_h(a\vert s)\qmfg_{h}^{\tau}(s,a \vert L^{\pi}, \pi) - \sum_a \pi_h'(a\vert s)\qmfg_{h}^{\tau}(s,a \vert L^{\pi'}, \pi')\right\vert}_{:= (\triangle)}.
    \end{align*}
    For the entropy term $(\square)$ we have
    \begin{equation*}
         \tau \vert \entropy(\pi_h(s)) -  \entropy(\pi_h'(s)) \vert \leq \tau\left(\log(\tfrac{1}{\epsilon}) + 1\right) \|\pi_h(s) - \pi_h'(s)\|_1
    \end{equation*}

    By applying triangular inequality we can upper bound $(\triangle)$ as follows:
    \begin{align*}
        (\triangle) \leq \sum_a \vert \pi_h(a\vert s) - \pi_h'(a \vert s) \vert \qmfg_h^{\tau}(s,a \vert L^{\pi}, \pi)\vert + \sum_a \pi_h'(a \vert s) \vert \qmfg_h^{\tau}(s,a \vert L^{\pi}, \pi) -  \qmfg_h^{\tau}(s,a \vert L^{\pi'}, \pi') \vert
    \end{align*}
    The first term can be bounded as a function of $H$ and the maximal absolute utility $B_{\utilityfn}$:
    \begin{align*}
        \sum_a \vert \pi(a\vert s) - \pi'(a \vert s) \vert \qmfg_h^{\tau}(s,a \vert L^{\pi}, \pi)\vert \leq (H - h)\left(B_{\utilityfn} + \tau\log(\vert \setA \vert )\right) \|\pi_h(s) - \pi_h'(s)\|_1,
    \end{align*}
    where the term $\tau\log(\vert \setA \vert )$ comes from the maximal additional reward from the entropy regularizer.s

    For the second term we bound the absolute difference in $\qmfg$ functions:
    \begin{align*}
        \vert \qmfg_h^{\tau}(s,a \vert L^{\pi}, \pi) -  &\qmfg_h^{\tau}(s,a \vert L^{\pi'}, \pi') \vert\\
        &\leq \bigl| \Rauc_h(s,a, L_h^{\pi}) - \Rauc_h(s,a, L_h^{\pi'})\bigr| \\
        & \quad + \sum_{s'}\vert \Pauc_h(s' \vert s,a,L_h^{\pi}) - \Pauc_h(s' \vert s,a,L_h^{\pi'})\vert \vert \Vmfg_{h+1}^{\tau}(s,a \vert L^{\pi}, \pi)\vert\\
        &\quad + \sum_s \Pauc_h(s' | s,a ,L_h^{\pi'}) |\Vmfg_{h+1}^{\tau}(s' \vert L^{\pi}, \pi) - \Vmfg_{h+1}^{\tau}(s' \vert L^{\pi'}, \pi')|\\
        &\leq \bigl| \Rauc_h(s,a, L_h^{\pi}) - \Rauc_h(s,a, L_h^{\pi'})\bigr|\\
        &\quad + (H - h - 1)(B_\utilityfn + \tau \log(|\setA|)) \|\Pauc_h(s,a,L_h^{\pi}) - \Pauc_h(s,a,L_h^{\pi'})\|_1\\
        &\quad \max_s |\Vmfg_{h+1}^{\tau}(s' \vert L^{\pi}, \pi) - \Vmfg_{h+1}^{\tau}(s' \vert L^{\pi'}, \pi')|.
    \end{align*}
    By combining all intermediate bounds and applying \Cref{lemma:lipschitz_mfg_rewards_transition} we have
    \begin{align*}
        \left| V_h^\tau(s \mid L^\pi, \pi) - V_h^\tau(s \mid L^{\pi'}, \pi') \right| &\leq A_h \Delta_h + G_h\sum_{h' \leq h}\beta^{h-h'}\Delta_h' \\
        & \quad + \max_s |\Vmfg_{h+1}^{\tau}(s' \vert L^{\pi}, \pi) - \Vmfg_{h+1}^{\tau}(s' \vert L^{\pi'}, \pi')|,
    \end{align*}
    by induction the bound holds. As next we compute the global Lipschitz constant.
    \begin{align*}
        \left| V_h^\tau(s \mid L^\pi, \pi) - V_h^\tau(s \mid L^{\pi'}, \pi') \right| 
    \leq \underbrace{\sum_{h' = h}^{H-1} A_{h'} \Delta_{h'}}_{(\square)} 
    + \underbrace{\sum_{h' = h}^{H-1} G_{h'} \sum_{h'' \leq h'} \beta^{h' - h''} \Delta_{h''}}_{(\triangle)}.    
    \end{align*}
    For $(\square)$ we have:
    
    \begin{align*}
        \sum_{h' = h}^{H-1} A_{h'} \Delta_{h'} \leq \sum_{h = 0}^{H-1} A_{h} \Delta_{h} \leq A_0\sum_{h = 0}^{H-1}\Delta_{h}
    \end{align*}
    
    For $(\triangle)$ we have:
    \begin{align*}
        \sum_{h' = h}^{H-1} G_{h'} \sum_{h'' \leq h'} \beta^{h' - h''} \Delta_{h''} &\leq \sum_{h = 0}^{H-1} G_{h} \sum_{h' \leq h} \beta^{h - h'} \Delta_{h'}\\
        &= \sum_{h' = 0}^{H-1} \Delta_{h'}\sum_{j = 0}^{H-1- h'}G_{h' + j}\beta^j\\
        &= \sum_{h' = 0}^{H-1} \Delta_{h'}\left(G_{h'}\sum_{j=0}^{H-1-h'}\beta^{j} - g_1 \sum_{j=0}^{H-1-h'}j\beta^{j}\right)\\
        &\leq \sum_{h' = 0}^{H-1} \Delta_{h'}G_{h'}\sum_{j=0}^{H-1-h'}\beta^{j}
    \end{align*}
    Using the definition of the geometric sum, and observing that the constants $G_{h'}$ as well as the nested geometric sum decrease as $h'$ increases, we can further simplify the bound by pulling out the leading terms. For $\beta \not= 1$:
    \begin{equation*}
        \sum_{h' = 0}^{H-1} \Delta_{h'}G_{h'}\sum_{j=0}^{H-1-h'}\beta^{j} \leq \frac{1 - \beta^{H}}{1- \beta}(g_0 + g_1(H-1))\sum_{h=0}^{H-1}\Delta_h,
    \end{equation*}
    while for $\beta = 1$:
    \begin{equation*}
        \sum_{h' = 0}^{H-1} \Delta_{h'}G_{h'}\sum_{j=0}^{H-1-h'}\beta^{j} \leq H(g_0 + g_1(H-1))\sum_{h=0}^{H-1}\Delta_h.
    \end{equation*}
    Combining the results (for $\beta \not= 1$) of $(\square)$ and $(\triangle)$ we have:
    \begin{align*}
        \left| V_h^\tau(s \mid L^\pi, \pi) - V_h^\tau(s \mid L^{\pi'}, \pi') \right|  (\tau (\log(\tfrac{1}{\epsilon}) + 1) + HB + \frac{1-\beta^{H}}{1- \beta}(g_0 + g_1(H-1)))\sum_{h=0}^{H-1}\Delta_h.
    \end{align*}
    For $\beta = 1$ the geometric term is replaced by $H$. By the definitions of $B, g_0$ and $g_1$ the claim follows.
    
\end{proof}

\section{Experiment Details}
\label{app:exp_details}

\paragraph{Implementation Details.}
The experiments were implemented in JAX and PyTorch, the code is provided in the supplementary material.
We implement the adjoint method in JAX.
For the PyTorch implementation, some code was adapted from \cite{mfglib}.
All error bars in experiments are one standard deviation away from the mean.

\paragraph{Hardware and Compute Time.}
We run our experiments on a single NVIDIA H100 GPU with an AMD EPYC 16-core CPU.
One run of \myalgname{} for 1000 iterations takes 6 minutes, apart from the experiment (A8) described below with a long time horizon $H=100$, which takes 20 minutes for 1000 iterations.
The beach bar process experiments take roughly 3 minutes for 1000 iterations.

\paragraph{Parameterizing $\theta$ in $\setM_{\text{bb}}$.}
For the beach bar process, we parameterize the per state payments as $\theta_s = p_{\text{max}} \operatorname{sigmoid}(\xi_s)$, where $p_{\text{max}}$ is the maximum per state payment and the unconstrained parameters $\xi\in\mathbb{R}^\setS$ are learned via \myalgname{}.

\paragraph{Parameterizing $\theta$ in $\Mauc$.}
We parameterize $\paymentfn_h^\theta$ and sold goods $\alpha_h^\theta$ as residual neural networks sharing a base.
The base network, $f^\theta_{\text{base}}$, has $d_{\text{in}} = H+|\setA|+1$ inputs consisting of one-hot encoded time vector $\vece_h$, $|\setA|$-dimensional vector of bid distribution $\nu^{-\perp}$, and remaining goods at round $h$ denoted $r_h$, (given by $\alphamax - \sum_{h'=0}^{h-1}\alpha_{h'}$).
For the input vector $x_{\text{in}} \in \mathbb{R}^{d_\text{in}}$, the base residual network is defined as 
\begin{align*}
    x_{\mathrm{in}}
&= \begin{bmatrix}\mathbf{e}_h \\[0.3em]\nu^{-\perp} \\[0.3em] r_h\end{bmatrix}
\in\mathbb{R}^{H+d+1}, \\
h^{(1)}
&= \mathrm{ReLU}\bigl(W^{(1)}\,x_{\mathrm{in}} + b^{(1)}\bigr), \\
y_{\mathrm{base}} := h^{(2)}
&= \mathrm{ReLU}\bigl(W^{(2)}\,h^{(1)} + b^{(2)} + V^{(2)}\,x_{\mathrm{in}} + c^{(2)}\bigr)
\in\mathbb{R}^{d_{\mathrm{hidden}}}.
\end{align*}
The goods to be sold this round are then computed by:
\begin{align*}
    \alpha_h
= r \;\times\;\sigma\bigl(w_g^\top\,y_{\mathrm{base}} + b_g\bigr),
\in\mathbb{R}.
\end{align*}
and the payments functions for bids is computed then by:
\begin{align*}
    h^{(3)} 
= \mathrm{ReLU}\bigl(W^{(3)}\,y_{\mathrm{base}} + b^{(3)}\bigr)\;+\;y_{\mathrm{base}}, \quad
&h^{(4)} 
= \frac{1}{A-1}\,\sigma\bigl(W^{(4)}\,t + b^{(4)}\bigr)
\in\mathbb{R}^{A-1}, \\[0.6em]
x_{\mathrm{payment},1} = 0, 
\quad 
&x_{\mathrm{payment},i}
= \sum_{j=1}^{i-1} h^{(4)}_j,
\quad i=2,\dots,A.
\end{align*}
Note that this parameterization ensures that the payment rule $p_h^\theta$ is a monotonic increasing function of bids $a$.
The parameters $\theta$ of the mechanism overall are $W^{(1)},\,V^{(2)}\in\mathbb{R}^{d_{\mathrm{hidden}}\times (H + d + 1)},\quad
W^{(2)},\,W^{(3)}\in\mathbb{R}^{d_{\mathrm{hidden}}\times d_{\mathrm{hidden}}},\quad
W^{(4)}\in\mathbb{R}^{(A-1)\times d_{\mathrm{hidden}}}$
and $b^{(1)},\,b^{(2)},\,c^{(2)},\,b^{(3)},\,w_g\in\mathbb{R}^{d_{\mathrm{hidden}}},\quad
b^{(4)}\in\mathbb{R}^{A-1},\quad
b_g\in\mathbb{R}$.

\paragraph{Baseline algorithms.}
For the zeroth-order baseline algorithms $0$-\textsc{SGD} and $0$-\text{Adam}, we use the standard 2-point (biased) gradient estimator
\begin{align*}
    \widehat{\nabla}_\theta := \frac{G^{T}_{\textrm{approx}}( \theta + u_{\text{zero}} z ) - G^{T}_{\textrm{approx}}( \theta - u_{\text{zero}} z )}{2  u_{\text{zero}}} D z, 
\end{align*}
where $\theta \in \mathbb{R}^D$, $z$ is uniformly distributed on the sphere $\mathbb{S}_{D-1}$, and $u_{\text{zero}}$ is a tunable hyperparameter.
This estimator satisfies the well-known property
\begin{align*}
    \Exop[\widehat{\nabla}_\theta] = \grad\Exop[G^{T}_{\textrm{approx}}(\theta + u_{\text{zero}} \widebar{z})], \quad \widebar{z} \propto \operatorname{Uniform}(\mathbb{B}_D).
\end{align*}
That is, $\widehat{\nabla}_\theta$ is an unbiased estimator of the gradient of a smoothed version of the function $G^{T}_{\textrm{approx}}$.
The bias is tunable by the parameter $u_{\text{zero}}$, with smaller values corresponding to less bias but potentially higher variance in estimates.
Since a single evaluation of this gradient estimator takes 2 forward passes over $G^{T}_{\textrm{approx}}$, its run time is comparable to that of \myalgname{} per iteration.
For the baseline \textsc{Anneal}, each iteration, we sample a perturbation $n$ from the $D$-dimensional standard normal distribution.
After evaluating $G^{T}_{\textrm{approx}}(\theta), G^{T}_{\textrm{approx}}(\theta + \sigma_{\text{anneal}} n), G^{T}_{\textrm{approx}}(\theta - \sigma_{\text{anneal}} n)$ for the tunable hyperparameter $\sigma_{\text{anneal}} >0$, \textsc{Anneal} updates $\theta$ to be the best among $\theta, \theta - \sigma_{\text{anneal}} n, \theta + \sigma_{\text{anneal}} n $.

\textbf{Hyperparameters.}
All hyperparameters for the baselines as well as \myalgname{} are presented in \Cref{table:hyperparams}.
For a fair comparison, we perform a grid search on a range of values for the parameters for all baselines and take the best run after 10 repetitions.
In our experiments, \myalgname{} is robust to hyperparameter choices while zeroth order methods require some tuning.

\begin{table*}[h]
\centering
  \begin{tabular}{cp{7cm}p{4cm}} \toprule
    \textbf{Parameter} & \textbf{Explanation} & \textbf{Values}  \\ \midrule
    $\eta$ & Adam/SGD learning rate  &  \numlist[list-final-separator={, }]{3e-5; 1e-4; 3e-4; 1e-3; 1e-3; 1e-2} \\
    $u_{\text{zero}}$ & Noise magnitude for evaluating zeroth-order gradient estimator $\widehat{\nabla}_\theta$, for 0-\textsc{SGD} and 0-\textsc{Adam} & \numlist[list-final-separator={, }]{1e-3; 1e-2; 3e-2} \\
    $\sigma_{\text{anneal}}$ & Perturbation magnitude for \textsc{Anneal} &  \numlist[list-final-separator={, }]{1e-6; 1e-5; 1e-4; 1e-3; 1e-2; 3e-2} \\
     $\tau$ & Entropy regularization  & \num{1e-3} \\
     $\eta_{\text{OMD}}$ & OMD learning rate  & 10 \\
    $T$ & OMD iterations in \eqref{eq:tstepobjfunction}  & 400 \\
    $T_\text{val}$ &  OMD iterations for validation  & 500 \\
    $d_{\text{hidden}}$ & Hidden dimension of residual network parameterizing payments and sold goods &  \num{256} \\
    \bottomrule
  \end{tabular}
  \caption{Hyperparameters for the experiments on auctions.}
  \label{table:hyperparams}
\end{table*}

\subsection{Additional Results on the Beach Bar Process}

As mentioned in the main body of the paper, we first present the payment function $\theta_s$ learnt \myalgname{} after 1000 iterations.
As before, we report these by using a slightly higher OMD iteration step $T_{\text{val}} = 500$ than used for training, to demonstrate the robustness of our method.

\begin{figure}[H]
    \centering
    \input{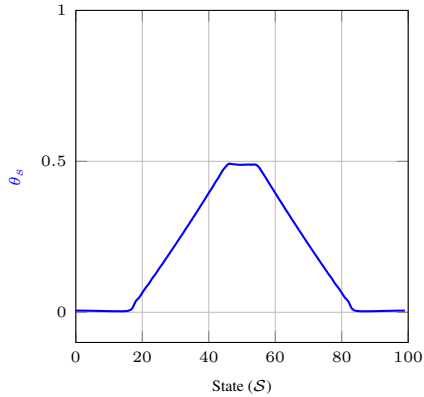}
    \caption{Payment function $s\rightarrow\theta_s$ learned after training with \myalgname{}, where payments are bounded on $[0, \sfrac{1}{2}]$.}
    \label{fig:bb_app_payment}
\end{figure}

A bottleneck in the beach bar experiment is the magnitude of payments $\theta_s$, which is restricted to be bounded on $[0,\sfrac{1}{2}]$.
We also report the experiment when $\theta \in [0,\sfrac{4}{5}]^\setS$ below (i.e., when $p_{\text{max}} = 0.8$), in \Cref{fig:bb_app_large_payment}.
As expected, the population distributions $L_h$ are smoother and closer to uniform in this case.
In both cases, \myalgname{} behaves as expected: the exploitability of $T$ iterates of OMD is consistently low throughout training, suggesting that the $T$ step approximation objective remains close to a NE throughout training iterations (with exploitability $<0.02$).
The fact that the exploitability of the policy induced by the $T$ repeated iterations of OMD is due to the fact that in the beach bar experiments, tuning payments yields a NE distribution closer to uniform, which is also the initialization of OMD iterates.

\begin{figure}[H]
    \centering
    \begin{subfigure}{0.33\textwidth}
        \input{plots/beachbar_appendix/training}
        \centering
    \end{subfigure}
    \begin{subfigure}{0.33\textwidth}
        \centering
         \input{plots/beachbar_appendix/payments_high}
    \end{subfigure}%
    \begin{subfigure}{0.33\textwidth}
        \centering
        \includegraphics[scale=0.35]{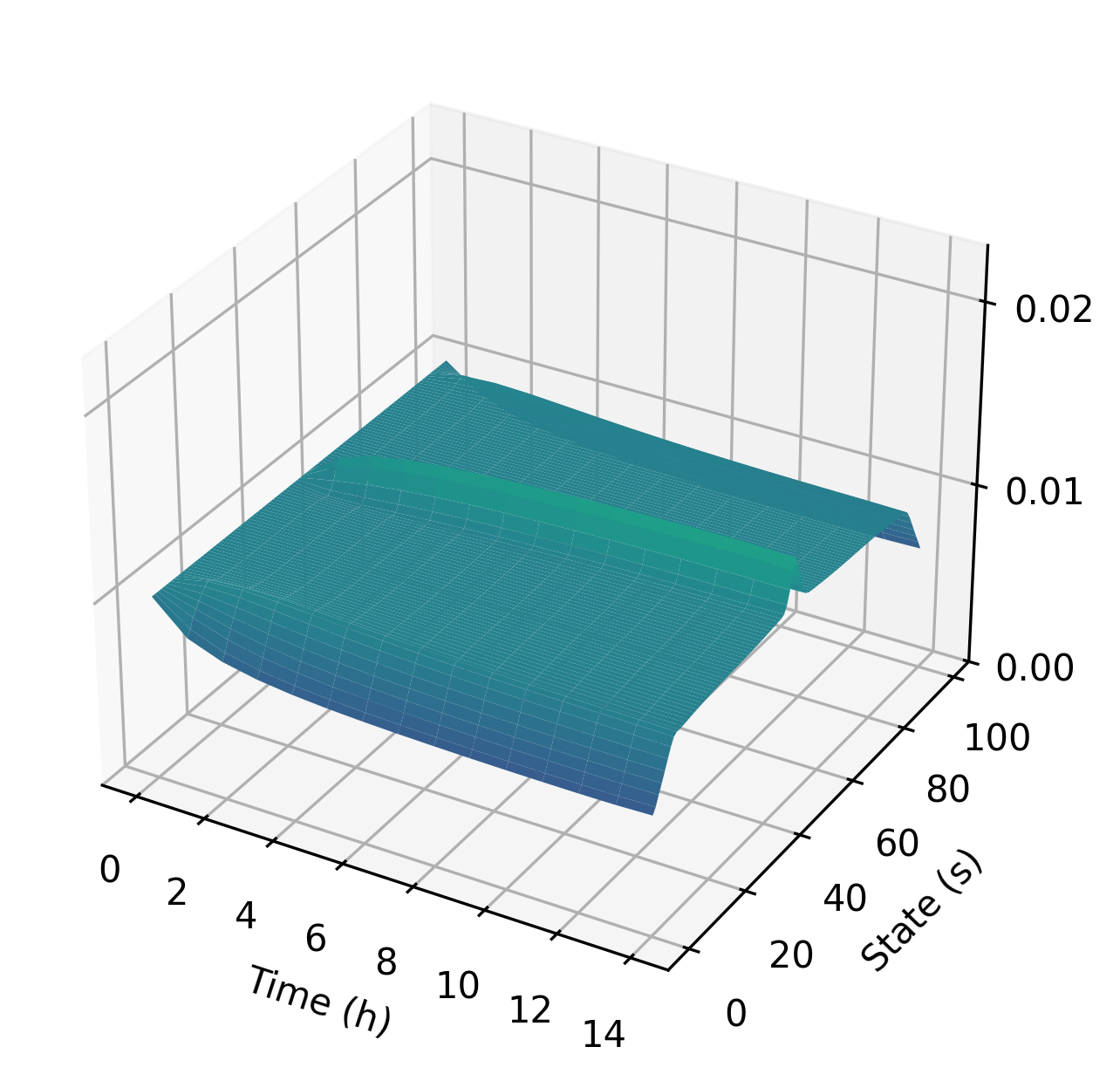}
    \end{subfigure}
    \caption{Payment design with \myalgname{} in $\setM_{\text{bb}}$ with larger payments. 
    \textbf{Left:} {\color{blue} objective} and {\color{red} exploitability} throughout training iterations. \textbf{Middle:} learned payment rule after training with \myalgname{}. \textbf{Right:}
    population flow in time after learning payments.}
    \label{fig:bb_app_large_payment}
\end{figure}

\subsection{Additional Results on Auctions}

We first analyze 3 further auctions with nonlinear utility functions for bidders.
Namely, we take (A1) presented in the main body of the paper where $H=4$, $\mu_0 = \operatorname{Uniform}(\setS)$, $\alpha_{\max}=0.8$, $|\setS|=|\setA|=100$ and bidders are single-minded with no evolution in valuations $s_h^i$ other than to transitions to $\perp$. 
\begin{enumerate}[leftmargin=*,noitemsep,topsep=0pt, label=\textbf{(A\arabic*)}, start=4]
    \item Risk-averse utilities formulated by
    \begin{align*}
        u_h(s,p) = \frac{1 - \exp\{-\beta (s-p)\}}{1 - \exp^{-\beta}},
    \end{align*}
    where we take $\beta=1$.
    \item Risk-seeking utilities formulated by
    \begin{align*}
        u_h(s,p) = \frac{\exp\{\beta (s-p)\} - 1}{\exp^{\beta} - 1},
    \end{align*}
    where we take $\beta=1$.
    \item Hyperbolic time discounting, which discounts future rewards, where the utility at time $h$ is given by:
    \begin{align*}
        u_h(s,p) = \frac{s - p}{1 + \lambda h},
    \end{align*}
    where we take $\lambda = 1$ as the time discount factor. 
\end{enumerate}

\begin{figure}[h]
\centering
\includegraphics{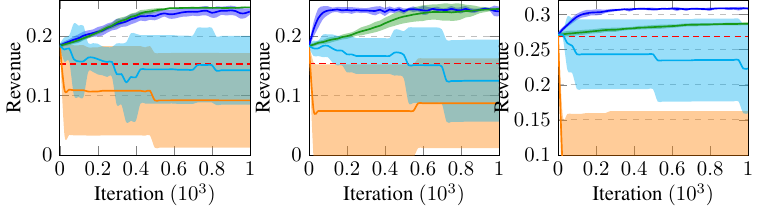}
\begin{tikzpicture}
\input{plots/legend_training_curves}
\end{tikzpicture}
\vspace{-.2cm}
\caption{$\objrev$ throughout iterations of \myalgname{} and baseline algorithms in settings with nonlinear utilities for bidders (A4-6), left to right.}
\label{fig:experiments:training_curves_utility}
\end{figure}

Across utility functions, \myalgname{} manages to beat all baselines.
In our experiments, we also observed significant qualitative changes in both the mean-field Nash equilibrium and the payment rule when nonlinear utility functions are used.
We show NE and payment rules suggested by the neural mechanism as NE in \Cref{fig:experiments:qualitiative_ne_utility}.
\begin{figure}[h]
\centering
\includegraphics{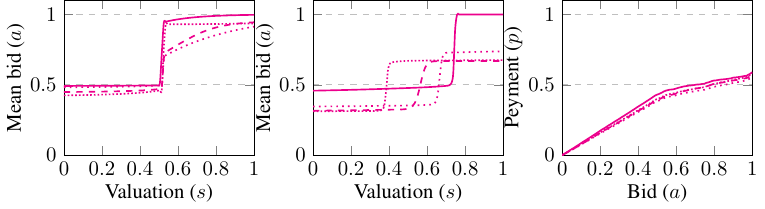}
\vspace{-.2cm}
\caption{\textbf{Left: } mean bids at NE at $\theta^*$ after training with \myalgname{} on the risk-seeking utility experiment (A5) for $h\in\{ 0,1,2,3 \}$.
\textbf{Middle: } mean bids at NE at $\theta^*$ after training with \myalgname{} on the hyperbolic time discounting utility experiment (A6) for $h\in\{ 0,1,2,3 \}$.
\textbf{Right: } payment function in setting (A5) at the bids induced by the NE policy at $\theta^*$ for $h\in\{ 0,1,2,3 \}$.}
\label{fig:experiments:qualitiative_ne_utility}
\end{figure}

\paragraph{Experiments regarding the impact of horizon.}
We report the impact of agent regeneration and time horizon by evaluating \myalgname{} on the following auction setups: 
\begin{enumerate}[leftmargin=*,noitemsep,topsep=0pt, label=\textbf{(A\arabic*)}, start=7]
    \item $H=6$, agents regenerate with probability $1$ (that is, they never transition to $\perp$ even when they win a round), linear utilities for bidders, $\alpha_{\max}=0.8$, $\mu_0(s) \propto \gamma^s$ for $\gamma = 0.9$, dynamic values with $w(s'|s) \propto \exp\{\sfrac{-(1.2 s-s')^2}{2\sigma^2}\}$ for $\sigma=0.2$.
    \item Long horizon $H=100$, $\alphamax=5$, linear utilities for bidders, agents regenerate with probability $0.015$,
    $\mu_0(s) \propto \gamma^s$ for $\gamma = 0.9$, dynamic values with $w(s'|s) \propto \exp\{\sfrac{-(1.2 s-s')^2}{2\sigma^2}\}$ for $\sigma=0.01$.
\end{enumerate}
The results are reported on \Cref{fig:experiments:training_curves_extended}.

\paragraph{General objectives.}
Finally, we explore the impact of optimizing over more general objectives other than revenue.
We define the objective
\begin{align*}
    g_{\text{mix}}(\theta, \vecpop) := g_{\text{rev}}(\theta, \vecpop) + g_{\text{efficiency}}(\theta, \vecpop),
\end{align*}
where we define $g_{\text{efficiency}}$ as:
\begin{align*}
    g_{\text{efficiency}}(\theta, \vecpop) := \sum_{h = 0}^{H-1}\sum_{(s, a) \in \setV \times \setA}L_h(s, a)\pwin( s, a, L_h, \alpha_h^\theta(\bids(L_h))) \utilityfn_h(s, \paymentfn_h^\theta\left(a, \bids(L_h)\right)).
\end{align*}
We modify experiment (A1) and evaluate our \myalgname{} on the following setting.
\begin{enumerate}[leftmargin=*,noitemsep,topsep=0pt, label=\textbf{(A\arabic*)}, start=9]
    \item  $H=4$, $\mu_0 = \operatorname{Uniform}(\setS)$, $\alpha_{\max}=0.8$, and single-minded bidders (after winning stay at state $\perp$) with no evolution in valuations $s_h^i$ otherwise.
    The objective function is $g_{\text{mix}}$.
\end{enumerate}
The results are reported on \Cref{fig:experiments:training_curves_extended}.

\begin{figure}[h]
\centering
\includegraphics{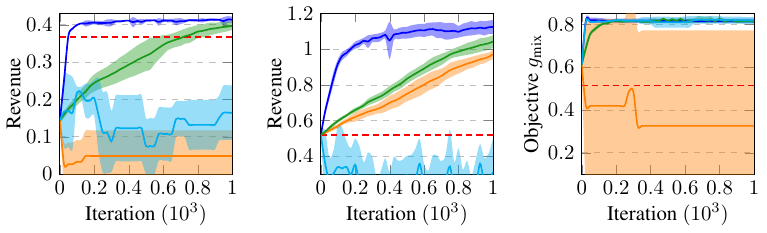}
\begin{tikzpicture}
\input{plots/legend_training_curves}
\end{tikzpicture}
\vspace{-.2cm}
\caption{\textbf{Left-middle:} $\objrev$ throughout iterations of \myalgname{} and baseline algorithms in settings (A7-8).
\textbf{Right:} $g_{\text{mix}}$ throughout iterations of \myalgname{} and baseline algorithms, in setting (A9).}
\label{fig:experiments:training_curves_extended}
\end{figure}

\paragraph{Experiments with static $\alpha_h$ and payment rules.}
Finally, to verify the impact of the mechanism having access to bids $\nu_h^{-\perp}$ at round $h$, we run a final experiment where the mechanism is independent of bids, which we call a \emph{static mechanism}.
In this case, we simply parameterize $p_h^\theta(a) = \operatorname{sigmoid}(\theta^{(1)}_{h,a})\cdot a$ and $\alpha_h^\theta(a) = \alphamax \frac{\exp\{\theta^{(2)}_{h}\}}{\sum_{h'}\exp\{\theta^{(2)}_{h'}\}}$, to ensure no more than $\alphamax$ hoods are sold and the payments never exceed the bid.
The parameter space is then $\theta := [\theta^{(1)}, \theta^{(2)}]$, for $\theta^{(1)} \in \mathbb{R}^{[H]\times \setA}, \theta^{(2)} \in \mathbb{R}^{[H]}$.
For static mechanisms, we observe much less significant improvement over the first price mechanism in general, which most likely originates from better allocation of goods over time when there are dynamics such as regeneration.
\begin{enumerate}[leftmargin=*,noitemsep,topsep=0pt, label=\textbf{(A\arabic*)}, start=10]
    \item $H=4$, static mechanism parameterization (independent of $\nu_h^{-\perp}$), each agent regenerates with probability $0.3$ at the end of every round, and has a linear utility function.
\end{enumerate}
The results in this setting are reported in \Cref{fig:experiments:training_curves_old}.

\paragraph{Experiments on approximating exploitability.}
We also report our attempts to measure the true $N$ player exploitability gap suggested by \Cref{theorem:approximation_auction} when $N$ is finite.
For the first price mechanism, we compute the MFG-NE on the auction setting (A1) introduced in the main paper.
Then, fixing $N=1000$, we simulate trajectories in the batched auction by setting the policies of 999 agents to the MFG-NE, and train PPO on the last bidder.
In this setting, we were not able to achieve a better mean reward for the last bidder than the MFG-NE.
This suggests the exploitability is close to $0$, we report the expected reward achieved by PPO throughout training in \Cref{fig:experiments:training_curves_old}. 
\begin{figure}[h]
\centering
\includegraphics{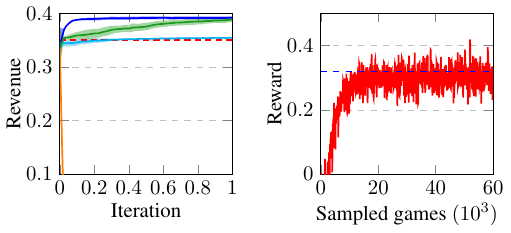}
\vspace{-.2cm}
\caption{\textbf{Left: } Revenue throughout training in the static mechanism setting (A10).
\textbf{Right: } PPO episodic rewards trained on the batched auction (A1) with $N=1000$ agents, all but one playing NE. Blue is the MFG best response expected reward, red is PPOs expected reward throughout training.}
\label{fig:experiments:training_curves_old}
\end{figure}

\end{document}